\documentclass[aps,pra,twocolumn,superscriptaddress,nofootinbib]{revtex4-2}
\usepackage[utf8]{inputenc}
\usepackage{amsmath,amssymb}
\usepackage{graphicx}
\usepackage{braket} %
\usepackage{physics} %
\usepackage{xcolor}
\usepackage{amsthm}
\usepackage{booktabs}

\usepackage[acronym,nolong]{glossaries}
\usepackage[colorlinks=true, linkcolor=blue, citecolor=blue, urlcolor=blue]{hyperref}

\usepackage{quantikz}
\usepackage{tikz}
\usepackage{subcaption}
\usepackage{float}
\usepackage{caption}
\usepackage{ragged2e}
\newacronym{qec}{QEC}{quantum error correction}
\newacronym{mlp}{MLP}{multilayer perceptron}
\newacronym{mld}{MLD}{maximum likelihood decoding}
\newacronym{mwpm}{MWPM}{minimum weight perfect matching}
\newacronym{rbim}{RBIM}{random-bond Ising model}
\newacronym{ci}{CI}{coherent information}
\newacronym{sgd}{SGD}{stoachstic gradient descent}
\newacronym{pde}{PDE}{probability density estimation}
\newacronym{mle}{MLE}{maximum likelihood estimation}
\newacronym{cnn}{CNN}{convolutional neural network}
\newacronym{vae}{VAE}{variational autoencoder}
\newacronym{elbo}{ELBO}{evidence lower bound}
\newacronym{kl}{KL}{Kullback-Leibler}
\newacronym{mse}{MSE}{mean squared error}
\newacronym{bce}{BCE}{binary cross-entropy}
\newacronym{rnn}{RNN}{recurrent neural network}
\newacronym{gnn}{GNN}{graph neural network}
\newacronym{ldpc}{LDPC}{low density parity-check codes}

\newtheorem{theorem}{Theorem}[section]
\newtheorem{definition}{Definition}[section]
\newtheorem{corollary}{Corollary}[section]
\newtheorem{remark}{Remark}[section]

\newlength{\circA}
\newlength{\circB}
\newlength{\circGap}

\definecolor{depolarizing}{rgb}{0.95,0.8,0.85} %
\definecolor{bitflip}{rgb}{0.7,0.9,0.7}        %
\definecolor{correlated}{rgb}{1.,0.5,0.5}

\date{June 2025}

\begin{document}

\title{Machine Learning Optimal Quantum Error Correction Thresholds}
\author{Dominik Seip}
\email{dominik.seip@rwth-aachen.de}
\affiliation{Institute for Theoretical Nanoelectronics (PGI-2), Forschungszentrum Jülich, 52428 Jülich, Germany}
\affiliation{Institute for Quantum Information, RWTH Aachen University, 52056 Aachen, Germany}
\affiliation{Tübingen AI Center, University of Tübingen, 72076 Tübingen, Germany}

\author{Luis Colmenarez}
\email{colmenarez@physik.rwth-aachen.de}
\affiliation{Institute for Theoretical Nanoelectronics (PGI-2), Forschungszentrum Jülich, 52428 Jülich, Germany}
\affiliation{Institute for Quantum Information, RWTH Aachen University, 52056 Aachen, Germany}

\author{Markus Schmitt}
\affiliation{Institute for Quantum Control (PGI-8), Forschungszentrum Jülich, 52428 Jülich, Germany}
\affiliation{Institute of Theoretical Physics, University of Regensburg, 93053 Regensburg, Germany}

\author{Markus Müller}
\affiliation{Institute for Theoretical Nanoelectronics (PGI-2), Forschungszentrum Jülich, 52428 Jülich, Germany}
\affiliation{Institute for Quantum Information, RWTH Aachen University, 52056 Aachen, Germany}

\date{\today}

\begin{abstract}
As quantum computers remain susceptible to noise, \gls{qec} is essential for preserving logical information during computations. However, the performance of \gls{qec} codes breaks down beyond certain noise thresholds, revealing fundamental limits on their ability to protect quantum information. These limitations can be characterized using information-theoretic measures such as the \gls{ci}, which quantifies the maximum rate at which logical information can be reliably transmitted through a noisy quantum channel. In this work, we establish a direct connection between the \gls{ci} and the \gls{bce} loss used when training neural network decoders. Specifically, we show that the \gls{ci} constitutes a sharp lower bound on the achievable loss for decoders that track logical operators across noisy channels. To this end, we develop a transformer-based neural network model based on \gls{mld}.
We train this network to estimate the \gls{ci} and evaluate its performance on the surface code under three distinct noise models: code capacity, phenomenological, and circuit-level noise. Our results demonstrate that the network accurately predicts \gls{ci} and yields threshold estimates that closely match known theoretical limits. When used as a decoder, the network significantly outperforms the \gls{mwpm} decoder in terms of logical error rate. Additionally, we introduce a novel soft post-selection scheme that independently treats uncertainty in both logical operators and relies on confidence-based filtering of the network's output. We prove that such post-selection strategies, based on the \gls{mld} cosets, are optimal, and demonstrate their scalability in terms of both logical error rate and abort probability.
These findings establish transformer-based architectures as powerful tools for \gls{qec} and provide the first numerical evidence supporting the optimality and scalability of \gls{mld}-based post-selection.

\end{abstract}

\maketitle

\section{Introduction}
Quantum computers promise to solve certain computationally hard problems efficiently, with the potential for exponential speedup in specific applications \cite{feynman_simulating_1982,shor_algorithms_1994,grover_fast_1996}. 
However, the reliability of quantum computations remains limited by noise, which continues to be one of the main obstacles to realizing large-scale quantum computing.
To mitigate errors in quantum processors, \gls{qec} protocols encode logical information across multiple physical qubits, thereby introducing redundancy \cite{terhal_quantum_2015,nielsen_quantum_2010}. This redundancy enables the detection and correction of errors by extracting a syndrome, a set of measurement outcomes that reveals information about the errors affecting the system. The particular way in which logical information is distributed among physical qubits is determined by a \gls{qec} code. Unlike in classical error correction \cite{shannon_mathematical_1948}, decoding quantum errors is computationally more challenging due to the degeneracy of many quantum codes: different error patterns can produce the same syndrome and decoding has to be based on determining the likelihood of such error equivalence classes instead of the plain errors \cite{iyer_hardness_2015}. 
The process of interpreting the syndrome and selecting an appropriate recovery operation is known as decoding and constitutes a central, non-trivial component of \gls{qec}~\cite{iolius_decoding_2024}.

Every \gls{qec} code can only tolerate a certain amount of noise beyond which their error correction ability breaks down \cite{shor_scheme_1995,knill_resilient_1998,kitaev_fault-tolerant_2003,shor_fault-tolerant_1997,shor_fault-tolerant_1997,aharonov_fault-tolerant_2008}. Below this threshold, the logical information can be recovered with fidelity approaching unity \cite{schumacher_quantum_1996}. Accurate threshold estimates are essential to develop and improve \gls{qec} codes. However, it is, in general,  computationally hard to determine optimal thresholds, as it involves sampling errors, decoding syndromes and estimating logical failures using an optimal and generally inefficient decoder \cite{iyer_hardness_2015}. In general, suboptimal decoding strategies are used, meaning that these approaches do not fully exploit the available information, resulting in a decoding performance that is lower than the theoretical maximum. As a consequence, they provide a lower bound on the optimal threshold. Another possibility to obtain optimal thresholds is by mapping the optimal decoding problem onto a phase transition in disordered classical spin models \cite{dennis_topological_2002}. This approach allows for estimating the threshold by solving partition functions with, for example, Monte Carlo methods.

More recently, the \gls{ci} of noisy mixed states has emerged as a powerful tool for accurately estimating optimal error thresholds~\cite{fan_diagnostics_2024,colmenarez_accurate_2024,colmenarez_2025_fundamental,shor_1995_scheme,steane_1954_multiple,gottesman_stabilizer_1997,knill_resilient_1998}. As an upper bound on the recoverable logical information, the \gls{ci} exhibits a sharp transition at the optimal threshold, thereby revealing fundamental limits of \gls{qec} codes and providing bounds on the performance of sub-optimal decoders. This insight is particularly important for realistic circuit-level noise models, where errors occur during gates, measurements, and idling, and where optimal thresholds remain relatively unexplored~\cite{fowler_surface_2012,rispler_random_2024,vodola_fundamental_2022}.
Originally introduced as an achievable rate for reliable quantum communication through noisy channels~\cite{schumacher_quantum_1996}, the \gls{ci} was later identified as the quantum code capacity and applied to threshold estimates of concatenated repetition codes~\cite{divincenzo_quantum-channel_1998}. Since then, it has been widely studied via mappings to statistical mechanical models, enabling analyses of toric and surface codes under various noise models~\cite{fan_diagnostics_2024,wang_intrinsic_2025,hauser_information_2024,colmenarez_2025_fundamental,behrends_surface_2025,tang_phases_2025}, as well as more general stabilizer and Floquet codes~\cite{lyons_understanding_2024,su_tapestry_2024}. Additional applications include characterizing teleportation fidelity~\cite{eckstein_robust_2024} and reformulating the problem in terms of disordered Majorana fermion models~\cite{huang_coherent_2024}. Notably, the \gls{ci} has been computed without the replica trick by relating it to the free energy of a random-bond Ising model~\cite{lee_quantum_2023}, and more recently by direct diagonalization of the output density matrix, yielding precise threshold estimates for surface and color codes in the code-capacity setting~\cite{colmenarez_accurate_2024}.

As the determination of optimal thresholds plays a critical role in guiding the design and evaluation of \gls{qec} codes, their successful experimental implementation has marked a significant milestone toward practical quantum computation and has been achieved across various quantum computing architectures, such as superconducting qubits \cite{acharya_quantum_2025,acharya_suppressing_2023,chen_exponential_2021,krinner_realizing_2022,zhao_realization_2022,andersen_repeated_2020,sivak_real-time_2023,gupta_encoding_2024,hetenyi_creating_2024,lacroix_scaling_2024}, trapped ions \cite{postler_demonstration_2024,ryan-anderson_high-fidelity_2024,egan_fault-tolerant_2021,paetznick_demonstration_2024,berthusen_experiments_2024,pogorelov_experimental_2025,postler_demonstration_2022,huang_comparing_2023}, and neutral atoms \cite{bluvstein_logical_2024,bedalov_fault-tolerant_2024,reichardt_logical_2024,Rodriguez2025Experimental}, already reaching and surpassing the break-even point \cite{acharya_suppressing_2023,acharya_quantum_2025,sivak_real-time_2023,gupta_encoding_2024,paetznick_demonstration_2024,Rodriguez2025Experimental}. Among the various \gls{qec} codes, surface codes, as a planar realization of toric codes, are particularly promising for near-term quantum devices, as they exhibit a high tolerance for errors \cite{fowler_surface_2012,bravyi_quantum_1998}. In the surface code, qubits are arranged on a two-dimensional grid, requiring only local connectivity between qubits, which makes them particularly well-suited for experimental realizations \cite{acharya_quantum_2025,acharya_suppressing_2023,chen_exponential_2021,krinner_realizing_2022,zhao_realization_2022,andersen_repeated_2020,gupta_encoding_2024,hetenyi_creating_2024,berthusen_experiments_2024,bluvstein_logical_2024}. \\

On the other hand, machine learning has become a powerful tool for quantum error correction, with neural network decoders extensively explored across a wide range of noise models, from idealized code-capacity settings \cite{torlai_neural_2017,krastanov_deep_2017,andreasson_quantum_2019,maskara_advantages_2019,wagner_symmetries_2020,fitzek_deep_2020,ni_neural_2020,meinerz_scalable_2022,overwater_neural-network_2022,matekole_decoding_2022,gicev_scalable_2023,egorov_end_2023,cao_qecgpt_2023} to more realistic scenarios \cite{sweke_reinforcement_2021,choukroun_deep_2024,wang_transformer-qec_2023,bordoni_convolutional_2023,baireuther_neural_2019, baireuther_machine-learning-assisted_2018,chamberland_deep_2018,varsamopoulos_decoding_2017,zhang_scalable_2023,chamberland_techniques_2023,hall_artificial_2024,varbanov_neural_2023,lange_data-driven_2023,bausch_learning_2024,blue_machine_2025,gu2026scalableneuraldecoderspractical}, and even in actual experiments~\cite{sivak_real-time_2023,hall_artificial_2024,varbanov_neural_2023,lange_data-driven_2023,bausch_learning_2024}. These approaches are attractive because they can learn error correlations, approximate unknown noise distributions, and operate without explicit knowledge of code geometry, though first methods to interpret the behavior of such NN decoders has been developed \cite{Boedeker_2025_Interpretability}. Several classes of network architectures have been used so far. For instance, recurrent architectures are commonly used for repeated syndrome measurements \cite{baireuther_neural_2019, baireuther_machine-learning-assisted_2018,varbanov_neural_2023}, alternative approaches include graph neural networks, Boltzmann machines, and reinforcement learning~\cite{torlai_neural_2017,krastanov_deep_2017,chamberland_deep_2018,baireuther_neural_2019,sweke_reinforcement_2021,fitzek_deep_2020,lange_data-driven_2023}. 
Inspired by advances in natural language processing and computer vision, transformer-based architectures have recently been applied to \gls{qec}, where their ability to capture long-range and global correlations has led to improved decoding performance~\cite{vaswani_attention_2023,cao_qecgpt_2023,wang_transformer-qec_2023,choukroun_deep_2024,bausch_learning_2024,blue_machine_2025}.  While most works focus on toric and surface codes as canonical topological codes, neural decoding has also been extended to color codes, heavy-hex codes, \gls{ldpc}, repetition codes, and GKP encodings~\cite{baireuther_neural_2019,maskara_advantages_2019,hall_artificial_2024,blue_machine_2025,sivak_real-time_2023}.

Furthermore, post-selection has been proposed as a means to enhance the performance of \gls{qec} by mitigating the complexity of the decoding process and improving output fidelity. Instead of correcting all detected errors deterministically, post-selection involves discarding experimental runs in which errors are likely to have occurred, thereby mitigating the computational overhead associated with decoding \cite{knill_quantum_2005}. This approach can improve the logical accuracy of the resulting encoded quantum states \cite{knill_quantum_2005}.
Recent experiments have demonstrated the practical relevance of post-selection in near-term quantum devices \cite{harper_fault-tolerant_2019,chen_calibrated_2022,postler_demonstration_2022,sundaresan_demonstrating_2023,ye_logical_2023,gupta_encoding_2024,paetznick_demonstration_2024,hetenyi_creating_2024,Bluvstein_2025}. 
However, this approach poses challenges regarding its scalability. The likelihood of obtaining an error-free run diminishes exponentially, rendering strict post-selection impractical for large-scale quantum computation.

To overcome this limitation, recent works have introduced soft post-selection techniques, which utilize the information contained in the syndrome measurements to decide whether to discard \cite{smith_mitigating_2024,english_thresholds_2024,Bluvstein_2025,bombin_fault-tolerant_2024}. The strategies include discarding syndromes with a high density of excited stabilizers~\cite{english_thresholds_2024},
or using modified \gls{mwpm} decoders that compare the likelihood gap between the two most probable corrections~\cite{smith_mitigating_2024,Bluvstein_2025}.
Both approaches can interpolate between the extremes of full error correction and strict error detection.
Such strategies are particularly relevant for fault-tolerant quantum computation, where discarding runs based on syndrome uncertainty can improve logical error rates without incurring the exponential penalty of strict post-selection. 
However, the theoretical understanding of these approaches remains limited, especially the influence of different types of decoding on the discard decision. Further, due to the complexity of optimal decoding, little is known about their impact on the performance of large-scale \gls{qec} codes such as the surface code, including their effect on thresholds and logical failure probabilities when compared to conventional decoding.\\

In this work we address several questions on neural network decoding and soft post-selection. First, we approximate \gls{mld} by means of a neural network decoder trained to extract the maximum amount of quantum information in the noisy state, thereby maximizing the \gls{ci}. To do so, we derive an analytical expression for the \gls{ci} in terms of the logical operators of the noisy \gls{qec} code.
Then, we train a transformer-based neural network on tracking the logical operators across the noise channel, allowing us to estimate the \gls{ci} of the corrupted quantum state. We find that training the network can be interpreted as maximizing the \gls{ci} from the noisy quantum state. Regarding the training process, by measuring the logical information alongside the error syndromes, we facilitate a supervised learning approach with automatically obtained labels. To mitigate difficulties during the training process and enhance convergence at high noise rates, we propose a curriculum learning approach that gradually increases the difficulty of the training examples.

Second, we test the performance of our neural network decoder relative to exact \gls{mld} under different noise models of increasing complexity. Our results demonstrate that our method approaches the exact \gls{ci}, therefore it serves as a good approximation to \gls{mld}. Furthermore, given the closeness to the exact \gls{ci}, our method stands as one approach for estimating the \gls{ci} for distances beyond the reach of conventional computational methods \cite{colmenarez_accurate_2024}. Furthermore, we benchmark the performance of our neural network decoder against the \gls{mwpm} decoder as standard decoding algorithm of the surface code. We find better performance of our neural network decoder over \gls{mwpm}, as expected for a \gls{mld} decoder. 

Third, we explore soft post-selection schemes based on approximate \gls{mld}. There we show how post-selection based on \gls{mld} is an optimal soft post-selection scheme and propose a new post-selection scheme, accounting separately for the uncertainty in the prediction of the $X$- and $Z$-type logical operator. To this end, we investigate the impact of post-selection on the network's predictions. Specifically, we discard syndromes for which the model lacks confidence in determining the appropriate correction. 
We find that both abort probability and logical error rate exhibit the same thresholds, a feature expected from \gls{mld} based post-selection \cite{english_thresholds_2024}.

Broadly speaking, our work differs from others \cite{torlai_neural_2017,krastanov_deep_2017,andreasson_quantum_2019,maskara_advantages_2019,wagner_symmetries_2020,fitzek_deep_2020,ni_neural_2020,meinerz_scalable_2022,overwater_neural-network_2022,matekole_decoding_2022,gicev_scalable_2023,egorov_end_2023,cao_qecgpt_2023,sweke_reinforcement_2021,choukroun_deep_2024,wang_transformer-qec_2023,bordoni_convolutional_2023,baireuther_neural_2019, baireuther_machine-learning-assisted_2018,chamberland_deep_2018,varsamopoulos_decoding_2017,zhang_scalable_2023,chamberland_techniques_2023,hall_artificial_2024,varbanov_neural_2023,lange_data-driven_2023,bausch_learning_2024,blue_machine_2025,gu2026scalableneuraldecoderspractical,sivak_real-time_2023,hall_artificial_2024,varbanov_neural_2023,lange_data-driven_2023,bausch_learning_2024} in that we aim to develop an understanding of the training process. This is reflected in our interpretation of the \gls{ci} as a cost function and the use of curriculum learning as an effective training strategy. Accordingly, we probe the limits of our approach within a specific network architecture; however, our goal is to demonstrate that the observed performance does not rely solely on architecture-specific features. In addition, we investigate thresholds and soft post-selection in the context of \gls{mld}, which is enabled by the improved scalability of the training procedure. In contrast, some works \cite{bluvstein_logical_2024,gu2026scalableneuraldecoderspractical,Bluvstein_2025} focus on the performance of confidence-based post-selection, without addressing threshold behavior.

This paper is organized as follows. In Sec.~\ref{sec:main-results}, we summarize the main results of this work. In Sec.~\ref{sec:qec}, we review the fundamental concepts relevant to this work. In Sec.~\ref{sec:ci-theo} we show the derivations of the \gls{ci} and the training interpretation in detail. In Sec.~\ref{sec:nn-model}, our neural network model and training strategy are presented in detail, alongside our results in estimating the \gls{ci} and comparing the decoding performance of our neural network decoder to the \gls{mwpm} algorithm. In Sec.~\ref{sec:post-selection} we derive theoretical results about soft post-selection and show our numerical experiments for post-selected network outputs. Finally, concluding remarks and outlook are presented in Sec.~\ref{sec:conclusion}.

\glsresetall
\section{Summary of main results} \label{sec:main-results}
In this section, the main results of the paper are briefly summarized. First, in Sec.~\ref{sec:ci-rewriting} we show how the \gls{ci} can be rephrased in terms of logical operators and, secondly how the \gls{bce} corresponds to the \gls{ci} for training neural network decoders. Then, in Sec.~\ref{sec:model-summary} we discuss the relevant architectural concepts of our model and show the results in estimating the \gls{ci} in Sec.~\ref{sec:numerical}. Finally, in Sec.~\ref{sec:post-selection-summary} we present how post-selection on \gls{mld} cosets is optimal, introduce a novel post-selection scheme and discuss that our network retains the optimal thresholds of logical error rate and abort probability under post-selection. 

\subsection{Coherent information in the training of neural network decoders} \label{sec:ci-rewriting}
The central result of this section is a theoretical connection between the loss function used to train neural network decoders and the \gls{ci} as fundamental information-theoretic quantity. Specifically, we show that the \gls{ci} is an upper bound on the \gls{bce} loss. This gives the training objective a clear physical interpretation: the NN is learning to recover as much of the logical information as possible after the noise channel. 

To make this connection concrete, we first express the \gls{ci} in terms of the logical operators after a noise channel. For a \gls{qec} code encoding a single logical qubit, the noise channel induces a logical operator $\lambda \in \{ \mathbb{I},X,Y,Z \} = \Lambda$. Measuring a syndrome $s$, the conditional probability $p(\lambda|s)$ describes the likelihood of each logical operator being induced by the syndrome-triggering error chain. With that, we rewrite the \gls{ci} as
\begin{equation}
\begin{aligned}
    I &= S(\rho_Q) - S(\rho_{RQ}) \\ &= \log{2} + \sum_s p(s) \sum_\lambda p(\lambda | s) \log p(\lambda | s).
\end{aligned}
\end{equation}
The coherent information is maximal ($I = \log 2$) if and only if we can determine the logical operator with certainty, i.e.~$\exists \lambda \in \Lambda \ \mathrm{s.t.} \ p(\lambda|s) = 1$ and $\forall \lambda^\prime \neq \lambda: \ p(\lambda^\prime | s) = 0$.\\

Factorizing the logical operator $\lambda$ into its logical $X$- ($\lambda_x)$ and $Z$-component ($\lambda_z$), $p(\lambda|s) = p(\lambda_z|\lambda_x , s) p(\lambda_x|s)$, the neural network is trained for binary classification of both logical operators. Therefore, the binary cross entropy (\gls{bce}) is employed as loss function. For a batch of samples ($s,\lambda$), predicting two binary output tokens, corresponding to a single encoded logical qubit, the loss can be expressed as 
\begin{equation}
\begin{aligned}
    &- \mathbb{E}_{s, \lambda \sim p(s, \lambda)} \mathcal{L}_\theta(s, \lambda) \\
    = &\sum_{s, \lambda \sim p(s, \lambda)} \sum_{j=x,z} \Bigl( \lambda_{j} \log \hat{\lambda}_{j} + (1 - \lambda_{j}) \log \bigl( 1 - \hat{\lambda}_{j} \bigr) \Bigr) \\
    = &\sum_{s} p(s) \sum_\lambda p(\lambda | s) \log q_\theta(\lambda | s) \\
    \leq &\sum_{s} p(s) \sum_\lambda p(\lambda | s) \log p(\lambda | s) = I - \log 2.
\end{aligned}
\end{equation}
Here, $\hat{\lambda}_x = q_\theta(\lambda_x | s)$ and $\hat{\lambda}_z = q_\theta(\lambda_z | \lambda_x, s)$ are the predicted probabilities of the logical operator component $j$ by the neural network $q_\theta$ with weights $\theta$, $\mathcal{L}_\theta$ is the BCE loss function depending on the weights of the NN $\theta$. The BCE loss (2nd line) can be reformulated using the joint writing $\lambda = (\lambda_z, \lambda_x)$, introducing the joint probability distribution $q_\theta (\lambda | s)$. As $q_\theta (\lambda | s)$ serves as a lower bound for the true probability distribution $p (\lambda | s)$, we show in Sec.~\ref{sec:ci-theo} that the CI serves as an upper bound on the negative BCE loss (up to a constant $\log 2$). Therefore, the loss defines how much of the information, i.e.~the logical operators, can be recovered after the noise channel, with the \gls{ci} as theoretical upper bound. Therefore we have shown that the loss has a natural interpretation in terms of the \gls{ci} and has a unique minimum, represented by the \gls{ci}. This ensures that a neural-network decoder trained with this loss function can, in principle, approach the theoretical limits imposed by \gls{qec}.

\subsection{Transformer for estimating logical operators} \label{sec:model-summary}
\begin{figure}
    \centering
    \includegraphics[width=\linewidth]{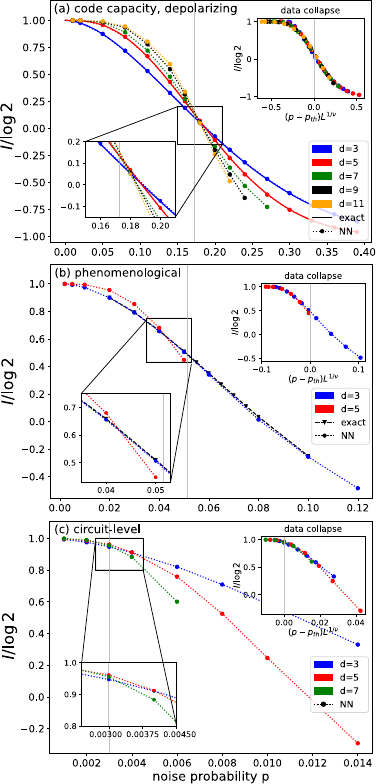}
    \caption[Estimates of the coherent information obtained using our transformer-based model.]{\justifying Estimates of the \gls{ci} obtained using our transformer-based model. We show (a) depolarizing noise in the code capacity setting, (b) phenomenological noise, and (c) circuit-level noise. Each dot corresponds to a single network trained at the particular noise probability. The threshold is marked by a grey vertical line and is either obtained by a finite-size scaling analysis or by the crossing of the CI. (a) The solid lines show the analytically obtained \gls{ci} by computing $p(\lambda|s)$ and $p(s)$ exactly across a densely sampled range of noise values, whereas the markers only indicate selected individual noise values, while (b) the black dots show the numerical value of the \gls{ci}, obtained following Ref.~\cite{colmenarez_accurate_2024}.}
    \label{fig:results_CI}
\end{figure}

Neural network decoders correct errors by applying a recovery operation $R=STL$. This operation is decomposed into a pure error $T$ correcting the excited syndrome, and a logical observable $L$ that accounts for logical flips induced by the error chain. The recovery operation $R$ only needs to account for the initial error $E$ up to some stabilizer $S$. Therefore, the core machine learning task of a neural network decoder reduces to estimating the induced logical observables, i.e.~estimating the conditional probability $p(\lambda|s)$, from which the recovery operation can be derived. We design and train a transformer-based architecture to approximate this density by parameterizing $\lambda = (\lambda_x, \lambda_z)$ as a two-bit binary number and predicting the logical operators autoregressively from the encoded stabilizer values, $p(\lambda|s) = p(\lambda_z|\lambda_x , s) p(\lambda_x|s)$, with two key contributions. %

First, for noise models requiring repeated syndrome measurements across multiple time steps, we introduce a divided space-time attention mechanism. Standard recurrent transformer architectures, such as in Ref.~\cite{bausch_learning_2024}, introduce a bias toward spatially local errors over timelike ones. Our approach removes this bias by attending separately to spatially and temporally aligned stabilizer tokens, while remaining computationally feasible compared to full attention over all stabilizers. The encoding is processed by a transformer decoder to output the conditional probabilities autoregressively \cite{vaswani_attention_2023,brown_language_2020} to account for the correlations in the logical operator, i.e. $p(Y) \neq p(X) p(Z)$.\\

Second, we observe a phenomenon we call conditional collapse. Hereby, the training gets stuck in a local minimum where the model ignores the input syndrome and outputs the average logical operator across the whole dataset, i.e.~$p(\lambda|s) \approx p(\lambda)$. This occurs when a model built upon a strong decoder, e.g.~a transformer decoder, finds it easier to memorize marginal statistics than to learn syndrome-conditioned predictions. To prevent this, we propose a curriculum learning approach. Hereby, models are trained at increasing noise rate therefore gradually increasing the complexity of the training samples. Those contributions allow us to advance the design of transformer architectures for decoding and \gls{qec}.

\subsection{Numerical experiments} \label{sec:numerical}

We validate our model across three noise settings of increasing complexity: the code capacity model, where only data qubits experience errors; the phenomenological noise model, which additionally introduces errors on ancilla qubits during syndrome measurement; and circuit-level noise, where every gate operation is assumed imperfect. In all three settings, our transformer produces accurate estimates of the CI across a range of code distances (see Fig.~\ref{fig:results_CI}). With varying code distances, we observe crossings of the CI, indicating the error-correction thresholds. Our estimates recover thresholds consistent with known optimal values from the literature and, importantly, demonstrate that the method can extract previously unknown thresholds where no analytical reference exists.

\subsection{Soft post-selection on MLD cosets} \label{sec:post-selection-summary}
Post-selection, discarding experimental shots that are unlikely to be corrected successfully, is a practically important tool for improving logical error rates at the cost of a higher abort probability. The central question we try to answer is: given a fixed abort probability budget, which samples should be discarded to minimally increase the logical error rate? For that, soft post-selection strategies \cite{smith_mitigating_2024,english_thresholds_2024,Bluvstein_2025} make discard decisions based on the decoder output, rather than strict post-selection, in which only samples with trivial syndromes are retained. The idea is that soft post-selection schemes will discard samples where a logical error is likely and retain those with a high probability of successful correction.\\

Formally, soft post-selection is formulated in terms of a selection function $f$ that maps a post-selection parameter $c$ onto a set of accepted syndromes $\{ s_1, ..., s_n\} \in \mathcal{P}(\mathcal{S})$. We prove that post-selection based on the \gls{mld} cosets is optimal, where samples are discarded if the maximum coset probability is below a threshold $c$, $\underset{\lambda \in \Lambda}{\mathrm{max}} \ p(\lambda|s) \leq c$. That is:
\begin{theorem} \label{thm:main}
Let $c, c^\prime \in [0, 1]$ be post-selection parameters, $f_\mathrm{MLD}$ the \gls{mld} post-selection function. Then any post-selection scheme $f: [0,1] \rightarrow \mathcal{P}(\mathcal{S})$ with $f(c) \subset f_\mathrm{MLD}(c)$ is optimal as for any other post-selection scheme $f^\prime: [0,1] \rightarrow \mathcal{P}(\mathcal{S})$ we have 
\[
p_\mathrm{abort}(c^\prime, f^\prime) \leq p_\mathrm{abort}(c, f) \Longrightarrow p_L(c^\prime, f^\prime) \geq p_L(c, f).
\]
\end{theorem}
To improve post-selection over the standard \gls{mld} post-selection, we propose a split post-selection scheme that considers the uncertainties of both logical operators separately. This is based on the factorization of the conditional probability $p(\lambda|s)=p(\lambda_x|s) p(\lambda_z|\lambda_x, s)$ and discards the sample if the probability of any individual logical operator lies within the interval $[\sqrt{c}, 1-\sqrt{c}]$.
\begin{definition} 
Split MLD post-selection: Let $c \in [0,1]$. Then the split MLD post-selection function is given by 
\begin{equation}
\begin{aligned}
&f_\mathrm{NN}: [0,1] \rightarrow \mathcal{P}(\mathcal{S}), \\ 
&f_\mathrm{NN}(c) = \{ s \in \mathcal{S}: \ (p(\lambda_x | s) > \sqrt{c} \ \lor \ p(\lambda_x | s) < 1 - \sqrt{c})  \\
&\quad  \land \ (p(\lambda_z | \lambda_x, s) > \sqrt{c} \ \lor \ p(\lambda_z | \lambda_x, s) < 1 - \sqrt{c}) \}.
    \end{aligned}
    \end{equation}
\end{definition}
This scheme is inspired by a confidence filtering on the output of the autoregressive decoding of the employed transformer architecture.\\ 

We prove and numerically observe that optimal thresholds are preserved under optimal post-selection, as the maximum coset probability converges to unity in probability:
\begin{theorem}
    Let $g_d: \mathcal{S} \rightarrow [0,1], \ g_d(s)=\underset{\lambda \in \Lambda}{\mathrm{max}} \ p(\lambda | s)$ be the mapping of a syndrome to its maximum coset probability for a distance-$d$ \gls{qec} code. Then, for $p<p_{th}$ we have
    \begin{equation}
        g_d \overset{P}{\rightarrow} 1 \ \mathrm{as} \ d \rightarrow \infty.
    \end{equation}
\end{theorem}

\glsresetall
\section{Background} \label{sec:qec}
In this section, we provide a brief introduction to the fundamental principles and theoretical background of \gls{qec}. Readers already familiar with these concepts may skip Sec.~\ref{sec:qec} and proceed directly to Sec.~\ref{sec:ci-theo}.
\subsection{Quantum Error Correction}
\subsubsection{Stabilizer codes}
A \gls{qec} code $\mathcal{C}$ defines an encoding of $k$ logical qubits into $n$ physical qubits. It specifies a codespace $C \subset \mathbb{C}^{2n}$, which is a subset of dimension $2^k$ of the full $n$-qubit Hilbert space. Denote by $\{ \psi_i: i=1,...,2^k \}$ a basis of $C$, i.e.~the codewords of the codespace, and by $\mathcal{E} \subset \mathcal{P}^{\otimes n}$ a set of Pauli errors acting on the encoded state. Then, $\mathcal{C}$ can correct for $\mathcal{E}$ iff
\begin{equation}
    \bra{\psi_i} E_a^\dagger E_b \ket{\psi_j} = C_{ab} \delta_{ij} \ \forall i,j \in \{1,...,2^k\}, E_a,E_b \in \mathcal{E}. \label{eq:error-correction}
\end{equation}
This means that two correctable errors should not confuse two distinct logical codewords and by detecting the errors we should not learn any information about the encoded quantum state, i.e.~$C_{ab}$ is independent of $i$ and $j$ \cite{knill_resilient_1998}.

Many QEC codes can be formulated within the stabilizer formalism, first introduced by Gottesman \cite{gottesman_stabilizer_1997}. A stabilizer code that encodes $k$ logical qubits into an $n$-qubit Hilbert space, is defined by a stabilizer group $\mathcal{S}$. The stabilizer group is generated by $n-k$ elements and defines the codespace $C$ as $+1$-eigenstate of the stabilizer group
\begin{equation}
    C = \{ \ket{\psi} : S \ket{\psi} = \ket{\psi} \ \forall S \in \mathcal{S} \}. 
\end{equation}
For the stabilizers to have joint eigenvectors, the stabilizer group $\mathcal{S}$ needs to be abelian, i.e.~all stabilizers commute. The negative identity operator $- \mathbb{I}$ is not part of the stabilizer group.

Stabilizer codes are able to detect errors that anticommute with at least one of the stabilizers in $\mathcal{S}$. Let $E \in \mathcal{E}$ s.t. $\{ S, E \}=0$ for some $S \in \mathcal{S}$ and $\ket{\psi_i} \in C$, then $S E \ket{\psi_i} = - E \ket{\psi_i}$ and inserting in Eq.~\eqref{eq:error-correction} gives
\begin{equation}
    \bra{\psi_i} E \ket{\psi_j} = \bra{\psi_i} S E \ket{\psi_j} = - \bra{\psi_i} E \ket{\psi_j} = 0.
\end{equation}
Thus, errors that anticommute with the stabilizer group can be detected. However, only errors that either differ by a stabilizer or produce a distinguishable syndrome can be corrected, as otherwise the overlap of the corrupted states is nonzero (see Eq.~\eqref{eq:error-correction}).

To identify logical operators in the given stabilizer code, we look for operators that rearrange states in $C$ but do not map the state out of the codespace $C$. Let us define the normalizer of $\mathcal{S}$ as $\mathcal{N}(\mathcal{S}) = \{A \in \mathcal{P}_n : A^\dagger S A \in \mathcal{S} \} \subset \mathcal{P}_n$, which is a subset of the $n$-qubit Pauli group and contains all operators that map $\mathcal{S}$ on itself. Note that $\mathcal{S} \subset \mathcal{N}(\mathcal{S})$.

In fact, we have $s \in \mathcal{S} \implies -s \notin \mathcal{S}$ (since $-\mathbb{I} \notin \mathcal{S}$) and for the elements of the $n$-qubit Pauli group $\mathcal{P}_n$: $\forall u,v \in \mathcal{P}_n: u v = \pm v u$. This implies that the normalizer of $\mathcal{S}$ under $\mathcal{P}_n$ is equal the centralizer of $\mathcal{S}$ under $\mathcal{P}_n$, which is defined as $\mathcal{C}(\mathcal{S}) = \{ A \in \mathcal{P}_n: As=sA \ \forall s \in \mathcal{S} \}$. Thus, $\mathcal{N}(\mathcal{S})$ contains all operators that commute with all stabilizers.
Then, errors $E \in \mathcal{N}(\mathcal{S}) \setminus \mathcal{S}$ cannot be detected, but change the logical information of the encoded quantum state and therefore act as logical operators. Hence
there exists an automorphism $\mathcal{N}(\mathcal{S}) \setminus \mathcal{S} \rightarrow \mathcal{P}_k$. These logical operators $\bar{X}_i, \bar{Z}_i$ fulfill
\begin{align}
    [\bar{X}_i, S] &= 0, \\
    [\bar{Z}_i, S] &= 0, \\
    \{\bar{X}_i, \bar{Z}_i\} &= 0, \ \forall i \in \{1,...,k\}, S \in \mathcal{S}.
\end{align}
The error syndrome of a code is a function that maps an error to a $(n-k)$-bit number, i.e.~$f_S: \mathcal{P}_n \rightarrow \{0,1\}$, such that
\begin{equation}
    f_S (E) =
    \begin{cases} 
    0, & \text{if } [S,E] = 0, \\
    1, & \text{if } \{S,E\}=0.
    \end{cases}
\end{equation}
Then $f(E) = (f_{S_1} (E), ..., f_{S_{n-k}} (E))$ is called the syndrome. 
In \gls{qec}, codes are commonly denoted by $[[n,k,d]]$, where $n$ is the number of physical qubits, $k$ the number of logical qubits encoded, and $d$ is the code distance. The distance quantifies the minimum weight of an error chain to cause a logical error. Specifically, a $[[n,k,d]]$ code can detect up to $d-1$ errors and correct at least all errors up to weight $\lfloor (d - 1)/2 \rfloor$\cite{gottesman_stabilizer_1997,nielsen_quantum_2010}.

\subsubsection{Error channels}

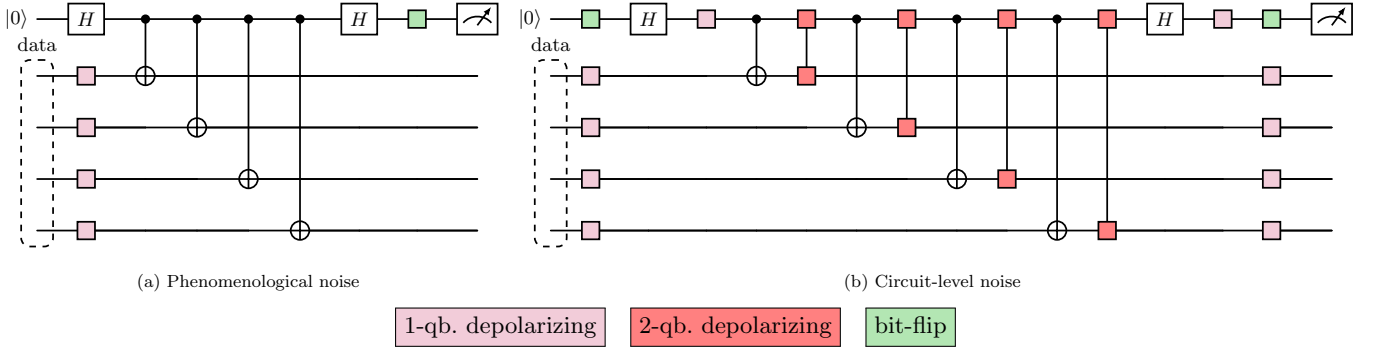
\begin{figure*}
  \centering
  \settowidth{\circA}{$\begin{quantikz}
    \lstick{\(|0\rangle\)} & \gate{H} & \ctrl{1} & \ctrl{2} & \ctrl{3} & \ctrl{4} & \gate{H} & \gate[style={fill=bitflip}]{} & \meter{} \\
    \gategroup[wires=4,steps=1,style={dashed, rounded corners,
      fill=white!30}, background]{{data}}
    & \gate[style={fill=depolarizing}]{} & \targ{} & \qw & \qw & \qw & \qw & \qw & \qw \\
    & \gate[style={fill=depolarizing}]{} & \qw & \targ{} & \qw & \qw & \qw & \qw & \qw \\
    & \gate[style={fill=depolarizing}]{} & \qw & \qw & \targ{} & \qw & \qw & \qw & \qw \\
    & \gate[style={fill=depolarizing}]{} & \qw & \qw & \qw & \targ{} & \qw & \qw & \qw \\
  \end{quantikz}$}%
  \settowidth{\circB}{$\begin{quantikz}
    \lstick{\(|0\rangle\)} & \gate[style={fill=bitflip}]{} & \gate{H} & \gate[style={fill=depolarizing}]{} & \ctrl{1} & \gate[style={fill=correlated}]{}\wire[d][1]{q} & \ctrl{2} & \gate[style={fill=correlated}]{}\wire[d][2]{q} & \ctrl{3} & \gate[style={fill=correlated}]{}\wire[d][3]{q} & \ctrl{4} & \gate[style={fill=correlated}]{}\wire[d][4]{q} & \gate{H} & \gate[style={fill=depolarizing}]{} & \gate[style={fill=bitflip}]{} & \meter{} \\
    \gategroup[wires=4,steps=1,style={dashed, rounded corners,
      fill=white!30}, background]{{data}}
    & \gate[style={fill=depolarizing}]{} & \qw & \qw & \targ{} & \gate[style={fill=correlated}]{} & \qw & \qw & \qw & \qw & \qw & \qw & \qw & \qw & \gate[style={fill=depolarizing}]{} & \qw \\
    & \gate[style={fill=depolarizing}]{} & \qw & \qw & \qw & \qw & \targ{} & \gate[style={fill=correlated}]{} & \qw & \qw & \qw & \qw & \qw & \qw & \gate[style={fill=depolarizing}]{} & \qw \\
    & \gate[style={fill=depolarizing}]{} & \qw & \qw & \qw & \qw & \qw & \qw & \targ{} & \gate[style={fill=correlated}]{} & \qw & \qw & \qw & \qw & \gate[style={fill=depolarizing}]{} & \qw \\
    & \gate[style={fill=depolarizing}]{} & \qw & \qw & \qw & \qw & \qw & \qw & \qw & \qw & \targ{} & \gate[style={fill=correlated}]{} & \qw & \qw & \gate[style={fill=depolarizing}]{} & \qw \\
  \end{quantikz}$}%
  \setlength{\circGap}{6pt}%
  \pgfmathsetmacro{\sfactor}{%
    \textwidth / (\circA + \circB + \circGap)%
  }%
  \makebox[\textwidth][c]{%
    \scalebox{\sfactor}{%
      \begin{minipage}{\circA}
        \centering
        \[
        \begin{quantikz}
          \lstick{\(|0\rangle\)} & \gate{H} & \ctrl{1} & \ctrl{2} & \ctrl{3} & \ctrl{4} & \gate{H} & \gate[style={fill=bitflip}]{} & \meter{} \\
          \gategroup[wires=4,steps=1,style={dashed, rounded corners,
            fill=white!30}, background]{{data}}
          & \gate[style={fill=depolarizing}]{} & \targ{} & \qw & \qw & \qw & \qw & \qw & \qw \\
          & \gate[style={fill=depolarizing}]{} & \qw & \targ{} & \qw & \qw & \qw & \qw & \qw \\
          & \gate[style={fill=depolarizing}]{} & \qw & \qw & \targ{} & \qw & \qw & \qw & \qw \\
          & \gate[style={fill=depolarizing}]{} & \qw & \qw & \qw & \targ{} & \qw & \qw & \qw \\
        \end{quantikz}
        \]
        \vspace{-2em}
        \subcaption{Phenomenological noise}\label{fig:circuit-noise-a}
      \end{minipage}%
      \hspace{\circGap}%
      \begin{minipage}{\circB}
        \centering
        \[
        \begin{quantikz}
          \lstick{\(|0\rangle\)} & \gate[style={fill=bitflip}]{} & \gate{H} & \gate[style={fill=depolarizing}]{} & \ctrl{1} & \gate[style={fill=correlated}]{}\wire[d][1]{q} & \ctrl{2} & \gate[style={fill=correlated}]{}\wire[d][2]{q} & \ctrl{3} & \gate[style={fill=correlated}]{}\wire[d][3]{q} & \ctrl{4} & \gate[style={fill=correlated}]{}\wire[d][4]{q} & \gate{H} & \gate[style={fill=depolarizing}]{} & \gate[style={fill=bitflip}]{} & \meter{} \\
          \gategroup[wires=4,steps=1,style={dashed, rounded corners,
            fill=white!30}, background]{{data}}
          & \gate[style={fill=depolarizing}]{} & \qw & \qw & \targ{} & \gate[style={fill=correlated}]{} & \qw & \qw & \qw & \qw & \qw & \qw & \qw & \qw & \gate[style={fill=depolarizing}]{} & \qw \\
          & \gate[style={fill=depolarizing}]{} & \qw & \qw & \qw & \qw & \targ{} & \gate[style={fill=correlated}]{} & \qw & \qw & \qw & \qw & \qw & \qw & \gate[style={fill=depolarizing}]{} & \qw \\
          & \gate[style={fill=depolarizing}]{} & \qw & \qw & \qw & \qw & \qw & \qw & \targ{} & \gate[style={fill=correlated}]{} & \qw & \qw & \qw & \qw & \gate[style={fill=depolarizing}]{} & \qw \\
          & \gate[style={fill=depolarizing}]{} & \qw & \qw & \qw & \qw & \qw & \qw & \qw & \qw & \targ{} & \gate[style={fill=correlated}]{} & \qw & \qw & \gate[style={fill=depolarizing}]{} & \qw \\
        \end{quantikz}
        \]
        \vspace{-2em}
        \subcaption{Circuit-level noise}\label{fig:circuit-noise-b}
      \end{minipage}%
    }%
  }%
  \vspace{0.1cm}
  \fcolorbox{black}{depolarizing}{1-qb.\ depolarizing}\quad
  \fcolorbox{black}{correlated}{2-qb.\ depolarizing}\quad
  \fcolorbox{black}{bitflip}{bit-flip}
  \caption[Noisy $X$-type stabilizer readout circuits.]{\justifying Noisy $X$-type stabilizer readout circuits with error locations marked by colored boxes. (a) Phenomenological noise: a depolarizing channel (pink) is applied to all data qubits before readout; measurement errors are modeled via a bit-flip channel (green). All gates are assumed noiseless. (b) Circuit-level noise: idling data qubits decohere during ancilla initialization and measurement (pink); single- and two-qubit gates are followed by depolarizing channels (red). Measurement errors are again modeled via bit-flip channels. Idling errors during the execution of Hadamard- and CNOT-gate are not included. A specific gate ordering must be chosen based on the code layout to ensure fault tolerance (see Ref.~\cite{Tomita_2014}). $Z$-type stabilizer readout is analogous, omitting the Hadamard and reversing CNOT directions.}
  \label{fig:circuit-noise}
\end{figure*}

We consider three noise models of increasing realism, ranging from idealized code-capacity noise to full circuit-level noise.

\paragraph{Data-qubit noise:}
In the code-capacity setting, all gates are assumed ideal and errors occur only on data qubits, followed by perfect stabilizer readout~\cite{terhal_quantum_2015,dennis_topological_2002,wang_confinement-higgs_2003}. The single-qubit noise channel is
\begin{equation}
\mathcal{N}(\rho) = (1-\sum_j p_j)\rho + \sum_j p_j E_j \rho E_j,
\end{equation}
with $E_j\in\{X,Y,Z\}$. For bit-flip noise, $p_X=p$ and $p_Y=p_Z=0$, yielding
\begin{equation}
\mathcal{N}(\rho) = (1-p)\rho + p X\rho X.
\end{equation}
Due to the symmetry of surface and toric codes, this is equivalent to a phase-flip channel. Depolarizing noise is obtained by setting $p_X=p_Y=p_Z=p/3$:
\begin{equation}
\mathcal{N}(\rho) = (1-p)\rho + \frac{p}{3}(X\rho X + Y\rho Y + Z\rho Z).
\end{equation}

\paragraph{Phenomenological noise:}
Measurement errors are included by applying bit-flip noise to ancilla qubits prior to readout (Fig.~\ref{fig:circuit-noise}a). Since measurement faults are uncorrelated in time, stabilizers are measured repeatedly. A \gls{qec} cycle consists of $d$ rounds, each with depolarizing noise on data qubits and noisy ancilla measurements.

\paragraph{Circuit-level noise:}
The circuit-level model explicitly includes noisy gates, measurements, and idling (Fig.~\ref{fig:circuit-noise}b). Single-qubit gates (and idling) are followed by depolarizing noise with probability $p$, while initialization and measurement are preceded by basis-dependent bit- or phase-flip noise. Two-qubit gates are followed by a two-qubit depolarizing channel,
\begin{equation}
\mathcal{N}(\rho) = (1-p)\rho + \frac{p}{15}\sum_{i=1}^{15} E_2^i \rho E_2^i,
\end{equation}
where $E_2^i\in\{\sigma_k\otimes\sigma_l\}\setminus\{\mathbb{I}\otimes\mathbb{I}\}$. While specific implementations vary~\cite{rispler_random_2024}, we assume the final stabilizer measurement in each \gls{qec} cycle is perfect.

\subsubsection{Maximum Likelihood Decoding (MLD)}\label{sec:mld}

Decoding maps measured stabilizer syndromes to recovery operations by estimating the underlying physical error. A central difficulty is code degeneracy: many distinct error patterns can produce the same syndrome, making quantum decoding substantially harder than its classical counterpart~\cite{iolius_decoding_2024}. We focus on the surface code~\cite{kitaev_fault-tolerant_2003,fowler_surface_2012}, following Ref.~\cite{iolius_decoding_2024}, though the discussion applies more generally.

Decoding strategies can be broadly classified as follows:
\begin{itemize}
    \item \emph{Minimum weight decoding}: Select the most likely error consistent with the observed syndrome,
    \begin{equation}
        R = \arg\max_{E \in \mathcal{P}_n} p(E|s),
    \end{equation}
    where $p(E)$ is the physical error distribution and $p(E|s)$ denotes the probability of the error $E$ given the observed syndrome $s$. This approach ignores code degeneracy.
    \item \emph{\Gls{mld}}: Given an error $E$ producing syndrome $s$, the coset $E\mathcal{C}(\mathcal{S})$ decomposes into logical cosets
    \begin{equation}
    \begin{aligned}
        E \mathcal{C}(\mathcal{S}) = &E\mathcal{S} \cup E\bar{X}\mathcal{S} \\
        &\cup E\bar{Y}\mathcal{S} \cup E\bar{Z}\mathcal{S},
    \end{aligned}
    \end{equation}
    (with straightforward generalization to $k>1$ logical qubits). MLD selects the logical coset with maximal posterior probability and applies a recovery $R \in f(s)\tilde{\lambda}\mathcal{S}$, where
    \begin{equation}
        \tilde{\lambda} = \arg\max_{\lambda \in \{\mathbb{I},\bar{X},\bar{Y},\bar{Z}\}} p(\lambda|s),
    \end{equation}
    and $p(\lambda|s)=\sum_{E\in\lambda\mathcal{S}}p(E|s)$.
\end{itemize}

By summing over all errors in a logical coset, MLD fully accounts for code degeneracy and is therefore optimal. Minimum weight decoding, which neglects this degeneracy, is generally suboptimal.

\subsection{Coherent Information}
The \gls{ci} provides an upper bound on the amount of information that can be recovered after a quantum state undergoes a noise channel \cite{schumacher_quantum_1996,lloyd_capacity_1997}. In fact, choosing an encoding that maximizes the \gls{ci} gives the quantum capacity of a noisy channel \cite{gyongyosi_survey_2018}. Thus, the \gls{ci} of a corrupted logical encoding quantifies how much of the stored information can in principle be recovered. As a result, it shows a crossing at the threshold while approaching a step function for $d \rightarrow \infty$: beyond the threshold, the information cannot be recovered with certainty, while it can be perfectly preserved below threshold. 
Recent works have shown the relation between \gls{ci} to a statistical mechanics models associated to \gls{qec} codes \cite{fan_diagnostics_2024,colmenarez_2025_fundamental,lyons_understanding_2024,eckstein_robust_2024,behrends_surface_2025,lee_exact_2024}. 

In the context of quantum error correction, the \gls{ci} is defined as difference between the entropy of the encoded qubits, $Q$, and of the codestate entangled with a noiseless reference qubit, $RQ$ \cite{fan_diagnostics_2024}:
\begin{equation}
    I = S(\rho_Q) - S(\rho_{RQ}), \ S(\rho) = - \mathrm{Tr}(\rho \log \rho).
\end{equation}
This setup is shown in Fig.~\ref{fig:CI-setting}. Here $S$ denotes the von Neumann entropy and $\rho_Q$ the reduced density matrix of system $RQ$ by tracing out the reference qubit. 
We can view $I$ as indicating how much of the entanglement between $R$ and $Q$ can be recovered after the noisy channel. Before the noise channel, the \gls{ci} is maximal, $I_0 = \log 2$, while after applying a CPTP map, i.e.~a quantum noise channel, the \gls{ci} can only decrease, $I\leq I_0 = \log 2$ \cite{schumacher_quantum_1996}. This implies, that $I = \log 2$ is both sufficient and necessary for the existence of a \gls{qec} protocol that recovers the logical state with certainty \cite{schumacher_quantum_1996}.\\
\begin{figure}
    \centering
    \includegraphics{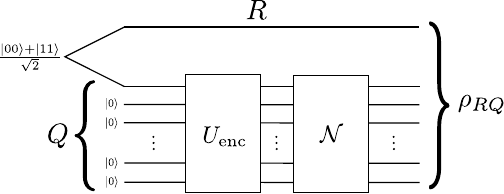}
    \caption[Setup for the coherent information in the context of QEC.]{\justifying Setup for the \gls{ci} in the context of \gls{qec}. A noiseless reference qubit (R) is entangled with another qubit, which is then encoded by $U_\textrm{enc}$ into a logical qubit (Q). This creates an equal superposition $(\ket{0_R 0_L} + \ket{1_R 1_L})/\sqrt{2}$ between reference qubit and logical qubit. The system $Q$ undergoes a noise channel $\mathcal{N}$. Illustration adapted from Fig.~1 of \cite{colmenarez_accurate_2024}.}
    \label{fig:CI-setting}
\end{figure}

\glsresetall
\section{Coherent information in the training process} \label{sec:ci-theo}
In this section, we describe with detailed derivations the role of \gls{ci} in the training process of neural network decoders.
In Sec.~\ref{sec:rewriting-main} we derive the rewriting of the \gls{ci} in terms of the logical operator of the encoding. Then, in Sec.~\ref{sec:interpretation-main} we present the \gls{ci} as lower bound on the loss function and as interpretation of the \gls{bce} for training neural network decoders. Finally, in Sec.~\ref{sec:STL-main} we present a definition of pure errors and the relevant logical operators used as labels for supervised learning.

\subsection{Rewriting the coherent information} \label{sec:rewriting-main}
As discussed in previous sections, the \gls{ci} quantifies the information recoverable from a logical state through a noisy channel, with a crossing at the threshold indicating the transition from high logical error rates to effective error suppression as the code distance increases. 
In the following derivations, we consider the setting in Fig.~\ref{fig:CI-setting} and rewrite the \gls{ci} in terms of error probability useful for decoders. We focus on the case of a single logical qubit ($k=1$); the generalization to multiple logical qubits is straightforward. We initially entangle the logical qubit  with the reference qubit, thereby preparing a Bell state
\begin{equation}
    \ket{\psi_0} = \sum_{l,r} \psi_0(l,r) \ket{s_0,l,r},
\end{equation}
where 
\begin{equation}
    \psi_0 (l,r) = 
\begin{cases} 
      \frac{1}{\sqrt{2}} & l=r, \\
      0 & l \neq r,
\end{cases}
\end{equation}
Here, the state of the qubits forming the system $Q$, representing the error correcting code, is described via the eigenvalue of the stabilizers $S_i$. The state $\ket{s_0,l,r}$ denotes the trivial syndrome, indicating that the state of the qubits is in the codespace. The logical quantum number $l$ specifies the logical state, corresponding to the eigenvalue of the respective logical operator, and the $r$ the state of the reference qubit in the computational basis. 
Considering a noise channel $\mathcal{N}(\rho) = \sum_E p(E) E \rho E^\dagger$ for some Pauli noise operator $E$ acting on system $Q$ and an initial state $\rho = \ket{\psi_0}\bra{\psi_0}$, we can rewrite the joint density matrix as
\begin{widetext}
\begin{equation}\label{eq:rho_rq}
\begin{aligned}
    \rho_{RQ} &= \sum_E p(E) E \ket{\psi_0} \bra{\psi_0} E^\dagger
    = \sum_s p(s) \sum_{l, r, l^\prime, r^\prime} \sum_E \frac{p(E) \delta_{s, s_E}}{p(s)} \psi_E(l,r) \psi_E^{*}(l^\prime, r^\prime) \ket{s, l, r} \bra{s, l^\prime, r^\prime}.
\end{aligned}
\end{equation}
\end{widetext}
Here, $p(s)$ describes the syndrome probability distribution, $\ket{\psi_0}$ the initial state, and $\ket{\psi_E}$ the erroneous state, given by
\begin{equation}
\begin{aligned}
    \ket{\psi_E} = E\ket{\psi_0} = \sum_{s, l, r} \delta_{s, s_E} \psi_E(l,r) \ket{s, l, r}.
\end{aligned}
\end{equation}
In general, the amplitude $\psi_E(l,r)$ depends on the error chain since an error $E$ can flip the value of the logical quantum number. 
In this formulation, it becomes clear that $\rho_{RQ}$ is block-diagonal, with each block corresponding to a different syndrome $s$. We can therefore consider each decoupled block, corresponding to a syndrome $s$, separately.

In the following we use that $\frac{p(E) \delta_{s,s_E}}{p(s)} = \frac{p(E,s)}{p(s)} = p(E|s)$ is the probability of an error chain $E$ given a measured syndrome $s$ and therefore zero for errors that cannot have caused the measured syndrome. Moreover, we find that, for a given syndrome, $\psi_E(l,r)$ depends only on whether and which logical operator was induced, rather than on the specific error chain. Thus we denote the amplitude $\psi_\lambda^s(l,r)$ by the logical operator $\lambda$ and the syndrome $s$ and rewrite the entries corresponding to the block for the syndrome $s$ as
\begin{widetext}
\begin{equation}\label{eq:logical_error}
\begin{aligned}
    \sum_E \frac{p(E) \delta_{s, s_E}}{p(s)} \psi_E(l,r) \psi_E^{*}(l^\prime, r^\prime) 
    = \sum_\lambda \psi_\lambda^s(l,r) (\psi_\lambda^s(l^\prime, r^\prime))^{*} \sum_{E_\lambda} p(E|s) 
    = \sum_\lambda p(\lambda |s) \psi_\lambda^s(l,r) (\psi_\lambda^s(l^\prime, r^\prime))^{*},
\end{aligned}
\end{equation}
\end{widetext}
with $\lambda \in \{1, X, Y, Z \}$ and $E_\lambda$ induces a logical operator $\lambda$, where a logical $X$ ($Y$, $Z$) is induced when $E_\lambda$ anticommutes with the respective logical operator (where anticommutation with both $X$ and $Z$ induces the $Y$ logical operator), i.e.~when, for the surface code, an odd number of qubits is affected by an $X$- ($Y$-, $Z$-) type error. 
Inserting Eq.~\eqref{eq:logical_error} into \eqref{eq:rho_rq} we obtain
\begin{equation}
\vspace{-0.5em}
\begin{aligned}
    \rho_{RQ} &= \sum_s p(s) \sum_\lambda p(\lambda | s) \ket{s,\lambda} \bra{s, \lambda} \\
    &\text{with} \ \ket{s, \lambda} = \sum_{l,r} \psi_\lambda (l,r) \ket{s,l,r},
\end{aligned}
\end{equation}
rewriting $\rho$ in its eigenbasis. Thereby, we have effectively transformed the four basis vectors for each syndrome $s$ to
\begin{align}
    \left\{
\begin{array}{l}
\ket{s,0,0} \\
\ket{s,0,1} \\
\ket{s,1,0} \\
\ket{s,1,1}
\end{array}
\right\}
\quad \longrightarrow \quad
\left\{
\begin{array}{l}
\frac{1}{\sqrt{2}} (\ket{s,0,0} + \ket{s,1,1}) \\
\frac{1}{\sqrt{2}} (\ket{s,1,0} + \ket{s,0,1}) \\
\frac{i}{\sqrt{2}} (\ket{s,0,1} - \ket{s,1,0}) \\
\frac{1}{\sqrt{2}} (\ket{s,0,0} - \ket{s,1,1})
\end{array}
\right\},
\end{align}
corresponding to the effect of each of the four logical operators applied to the first qubit of the initial Bell pair. For multiple logical qubits, this transformation is performed individually for each logical qubit. We obtain $\rho_Q$ by taking the partial trace over subsystem $R$,

\begin{widetext}
\begin{equation}
\begin{aligned}
    \rho_Q &= \text{Tr}_R(\rho_{RQ}) 
    = \sum_s p(s) \sum_\lambda p(\lambda | s) \sum_{l,l^\prime} \ket{s,l}\bra{s,l^\prime} \biggl( \sum_r \psi_\lambda(l,r) \psi_\lambda(l^\prime, r)^* \biggr) 
    = \sum_s \sum_l \frac{p(s)}{2} \ket{s,l} \bra{s,l},
\end{aligned}
\end{equation}
\end{widetext}
where we used $\sum_r \psi_\lambda(l,r) \psi_\lambda(l^\prime, r)^* = \frac{1}{2} \ \delta_{l, l^\prime} $ and $\sum_\lambda p(\lambda | s) = 1$. This is valid for a generic noise model since the Pauli error operators only modify the prefactor $\psi_\lambda^s (l,r)$ by flipping the logical state $l$ and/or introducing a phase factor. 

Having $\rho_{RQ}$ and $\rho_Q$ in a diagonal form we can now rewrite the \gls{ci} as 
\begin{equation}
\begin{aligned}
    I &= S(\rho_Q) - S(\rho_{RQ}) \\ &= \log{2} + \sum_s p(s) \sum_\lambda p(\lambda | s) \log p(\lambda | s). \label{eq:ci-rewriting}
\end{aligned}
\end{equation}

Hence any reduction in the \gls{ci} denotes a non-vanishing Shannon entropy $ \sum_\lambda p(\lambda | s) \log p(\lambda | s)$ of the conditional probability distribution of logical cosets (the same that appears in \gls{mld}, see Sec.~\ref{sec:mld}) averaged over all syndromes $s$. 
When considering noisy measurements, all the redundant degrees of freedom can be absorbed into the syndrome $s$.

\subsection{Interpretation of Training with BCE} \label{sec:interpretation-main}

Now we show how the \gls{ci}, as shown in Eq.~\eqref{eq:ci-rewriting}, is related to the training of neural network decoders. First, let us consider factorizing \[
p(\lambda | s) = p(\lambda_{z} | \lambda_{x}, s) \, p(\lambda_{x} | s).
\] 
The network is trained for binary classification of two logical operators, $X$ and $Z$. Therefore, the \gls{bce} is employed as loss function. For a batch of $n$ samples, predicting two binary output tokens $\hat{\lambda}$, the loss becomes
\begin{equation}
\begin{aligned}
    \mathcal{L}_\theta(s,\lambda) &= - \sum_{i=1}^n \sum_{j=x,z} \Bigl( \lambda_{ij} \log \hat{\lambda}_{ij} (s;\theta) \\ &\quad + (1 - \lambda_{ij}) \log \bigl( 1 - \hat{\lambda}_{ij} (s; \theta) \bigr) \Bigr).
\end{aligned}
\end{equation}
Minimizing the \gls{bce} loss function is equivalent to \gls{mle}, i.e.~maximizing the log-likelihood of the observed samples,
\begin{equation}
\begin{aligned}
    &- \sum_{s, \lambda \sim p(s, \lambda)} \log{q_\theta (\lambda| s)} \\
    = &- \sum_{s, \lambda \sim p(s, \lambda)} \sum_{j=x,z} \log{q_\theta (\lambda_j|\lambda_{<j}, s)}\\
    = &- \sum_{s, \lambda \sim p(s, \lambda)} \sum_{j=x,z} \Bigl( \lambda_{j} \log \hat{\lambda}_{j} (s;\theta) \\ & \quad \quad \quad \quad \quad + (1 - \lambda_{j}) \log \bigl( 1 - \hat{\lambda}_{j} (s; \theta) \bigr) \Bigr).
\end{aligned}
\end{equation}
Here, $\hat{\lambda}_j(s;\theta) = q_\theta(\lambda_j|\lambda_{<j},s)$ denotes the predicted probability of logical operator component $j$ by the neural network $q_\theta$, where the autoregressive ordering $x \prec z$ is adopted, such that $\lambda_{<x} = \emptyset$ and $\lambda_{<z} = \lambda_x$. Minimizing $\mathcal{L}_\theta$ with respect to $\theta$ is therefore equivalent to maximizing the log-likelihood $\log q_\theta(\lambda|s)$. This implies that the negative BCE loss is bounded above by the \gls{ci}, which quantifies the maximum recoverable information about the logical operators after the noise channel. We obtain
\begin{widetext}
\begin{equation}
\begin{aligned}
    &- \mathbb{E}_{s, \lambda \sim p(s, \lambda)} \mathcal{L}_\theta(s, \lambda)\\ 
    &= \sum_{s, \lambda \sim p(s, \lambda)} \sum_{j=x,z} \Bigl( \lambda_{j} \log \hat{\lambda}_{j} + (1 - \lambda_{j}) \log \bigl( 1 - \hat{\lambda}_{j} \bigr) \Bigr) \\
    &= \sum_s p(s) \sum_{\lambda} p(\lambda | s) 
    \sum_{j=x,z} \Bigl(\lambda_{j} \log \hat{\lambda}_{j} + (1 - \lambda_{j}) \log \bigl( 1 - \hat{\lambda}_{j} \bigr) \Bigr) \\
    &= \sum_s p(s) \biggl( 
        p(0,0|s) \bigl( \log (1 - \hat{\lambda}_x) + \log (1 - \hat{\lambda}_z(\lambda_x=0)) \bigr) + p(0,1|s) \bigl( \log (1 - \hat{\lambda}_x) + \log \hat{\lambda}_z(\lambda_x=0) \bigr) \\
        &\quad \quad \quad \quad \quad + p(1,0|s) \bigl( \log \hat{\lambda}_x + \log (1 - \hat{\lambda}_z(\lambda_x=1)) \bigr) + p(1,1|s) \bigl( \log \hat{\lambda}_x + \log \hat{\lambda}_z(\lambda_x=1) \bigr)
    \biggr) \\
    & = \sum_{s} p(s) \sum_\lambda p(\lambda | s) \log q_\theta(\lambda | s)  \leq \sum_{s} p(s) \sum_\lambda p(\lambda | s) \log p(\lambda | s),
\end{aligned}
\end{equation}
\end{widetext}
with equality, once $\hat{\lambda}_x(s) = q_\theta(\lambda_x | s) = p(\lambda_x | s)$ and $\hat{\lambda}_z(s, \lambda_x) = q_\theta(\lambda_z | \lambda_x, s) = p(\lambda_z | \lambda_x, s)$, i.e.~the network approximates the distribution $p$ perfectly. Thus, the loss function has a unique minimum, where the \gls{ci} is exactly recovered. 
This means that the \gls{bce} loss function is naturally suited for this task, as the network learns to accurately track the logical operators for a given syndrome, thereby maximizing the \gls{ci} and approaching the theoretical maximum of information that can be recovered from the syndrome after the noisy quantum channel.

\subsection{\texorpdfstring{$STL$}{STL}-decomposition with Maximum Support Logical Operator} \label{sec:STL-main}
\begin{figure*}
    \centering
    \includegraphics[width=\textwidth]{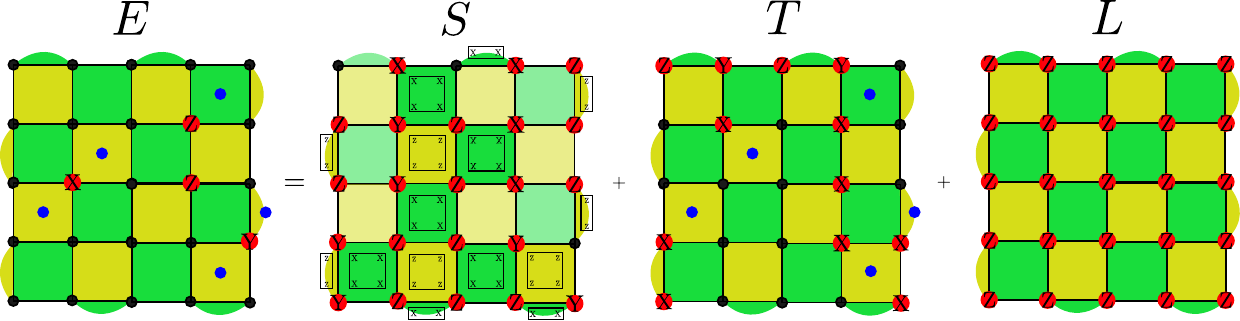}
    \caption[$STL$-decomposition of an error chain $E$ into its stabilizer $S$, pure error $T$, and logical operator $L$.]{\justifying $STL$-decomposition of an error chain $E$ (left) into stabilizer $S$ (middle left), pure error $T$ (middle right), and logical operator $L$ (right) on the surface code. Erroneous data qubits are marked by red dots. The excited stabilizers are marked by blue dots. Non-transparent plaquettes represent the stabilizer components of $S$. The pure error $T$ excites the syndrome but commutes with all stabilizers and logicals. The logical operator $L$ is chosen as the logical operator of maximal support, applying a Pauli operator to each data qubit. Adapted from Ref.~\cite{varsamopoulos_decoding_2017}.}
    \label{fig:STL-decomposition}
\end{figure*}

Most neural network decoders approach the decoding task as a classification problem \cite{torlai_neural_2017, chamberland_deep_2018,varsamopoulos_decoding_2017,maskara_advantages_2019,wagner_symmetries_2020,meinerz_scalable_2022,overwater_neural-network_2022,cao_qecgpt_2023,wang_transformer-qec_2023,blue_machine_2025}. Hereby, the error chain is decomposed as
\begin{equation}
    E = S T L, \label{eq:STL-decomposition}
\end{equation}
where $S$ is an element of the stabilizer group, $T$ is a Pauli operator that triggers the measured syndrome and $L$ a logical Pauli operator \cite{poulin_optimal_2006,varsamopoulos_decoding_2017}. Such a decomposition is exemplarily shown in Fig.~\ref{fig:STL-decomposition}. 
The definition of pure errors is not universal and can vary across different works, as Refs.~\cite{iolius_decoding_2024,varsamopoulos_decoding_2017,poulin_optimal_2006} rely solely on the defining property of pure errors $T$,
\begin{equation}
    [T_i, s_j] = \delta_{i,j} \ \forall i,j,
\end{equation}
i.e.~a pure error does not commute with exactly one generator of the stabilizer group, where the generators represent the measured parity checks. Commutation properties of pure errors with elements of the logical group are not considered. Therefore, we will follow and adapt the definition from Ref.~\cite{cao_qecgpt_2023}: A pure error additionally fulfills
\begin{equation}
    [T_i, \lambda_j] = 0 \ \forall i,j,
\end{equation}
meaning that the pure error commutes with any logical operator on the codespace. This has the advantage that the logical operator in the composition of Eq.~\eqref{eq:STL-decomposition} can be obtained via a global parity check, i.e.~calculating the commutator between $E$ and $\lambda_j$. This avoids costly labeling of the input data and allows for a cheap supervised learning approach.

Since $T$ can be easily determined, and the recovery succeeds for any stabilizer $S^\prime$, the recovery operation 
\begin{equation} 
    R=S^\prime T \tilde{L}
\end{equation} 
depends on the accurate estimation of the logical operator $\hat{L}$. This reduces decoding to a classification task with four possible labels $I, \bar{X}, \bar{Y}, \bar{Z}$ (for a $k=1$ code, encoding a single logical qubit). The syndrome-triggering part $T$ can be efficiently calculated using a look-up table, containing pure errors that each trigger an individual stabilizer measurement \cite{varsamopoulos_decoding_2017}. These pure errors are independent of each other and are multiplied to construct $T$.  

While the choice of logical operators $\lambda_j$ is arbitrary, we find it advantageous to use the maximum-support logical operator rather than a boundary-localized Pauli string. By multiplying the parallel representatives of the $X$- and $Z$-type logical operator, we obtain a logical operator acting with the respective Pauli on all data qubits, where measuring the logical operator effectively implements a global parity check. In contrast, when considering only a single string-like logical operator, only a fraction of roughly $\frac{1}{d}$ of the error chains triggers a flip of the logical operator, rendering most data samples uninformative for training, making it harder for the network to learn accurate error distributions at larger code distances. Predicting all separate string-like representatives autoregressively offers no advantage due to their mutual dependence.

The $STL$-decomposition for the maximum support logical operator is illustrated in Fig.~\ref{fig:STL-decomposition}. The logical operator in the decomposition can be obtained by a global parity check on the data qubits. This parity check can be read out in the setting of Fig.~\ref{fig:CI-setting-msmts}. 
Thus, the $STL$-decomposition does not need to be manually designed, but the value of the logical operator can be measured alongside the syndrome, so that labels are cheap and automatically obtained. 
Hereby, the pure errors $T$ are chosen as the complementary error chains, connecting an excited stabilizer to the respective boundary of the code with an even parity, thereby commuting with the logical operators (see Fig.~\ref{fig:STL-decomposition}). Consequently, every pure error has an even parity on the data qubits.

\glsresetall
\section{A transformer-based model for quantum error correction} \label{sec:nn-model}
In this section, we present our transformer-based neural network for \gls{qec} and demonstrate its performance by estimating the \gls{ci} for surface code sizes up to $d=11$ for depolarizing noise in the code capacity setting and up to $d=7$ for phenomenological and circuit-level noise. In Sec.~\ref{sec:model-architecture} we present the network architecture, while in Sec.~\ref{sec:training-details} and Sec.~\ref{sec:curriculum-details} details about the training are discussed. Finally, in Sec.~\ref{sec:estimating-ci-results} the results in estimating the \gls{ci} are shown and in Sec.~\ref{sec:decoding-results} the decoding performance compared to \gls{mwpm} is demonstrated.\\

The setup for this work is shown in Fig.~\ref{fig:CI-setting-msmts}. A quantum state is first entangled with a noiseless reference qubit $R$ and then encoded into a logical state within a system $Q$. The system undergoes either a single round (code capacity setting) or multiple rounds (phenomenological and circuit-level noise: $d$ rounds) of noise and subsequent stabilizer measurements. The final round of stabilizer readouts is assumed to be perfect and therefore not part of the noisy \gls{qec} cycle.

Prior to the noise channel, the entangled Bell state 
\begin{equation}
    \ket{\psi_0} = \frac{\ket{0_R 0_L} + \ket{1_R 1_L}}{\sqrt{2}}
\end{equation}
is stabilized by $X_R X_L$ and $Z_R Z_L$, extending the stabilizer group of the whole system $RQ$ by two generators, hence the state on RQ is stabilized by $d^2 + 1$ operators for a single logical qubit ($k=1$). The Bell stabilizers commute with each other and with all stabilizers of the \gls{qec} code. However, if a Bell stabilizer measurement flips, it indicates a change in the value of the logical operator, as the reference qubit is not affected by noise. Therefore, the two Bell stabilizers are measured at the end of the \gls{qec} cycle, enabling the simultaneous readout of both logical operators of the encoded quantum state. 

The machine learning task of learning the \gls{ci} reduces to estimating the conditional densities $p(\lambda | s)$ of each logical operator $\lambda$ given a measured syndrome $s$, where $\lambda$ can be written as a two-bit binary number representing the state of the encoded logical qubit. These probabilities will be used to compute the \gls{ci} in terms of $p(\lambda | s)$ according to Eq.~\eqref{eq:ci-rewriting}. 

\begin{figure}
    \centering
    \includegraphics[width=\linewidth]{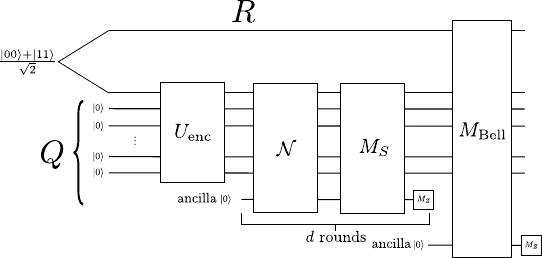}
    \caption[Circuit for estimating conditional logical probabilities $p(\lambda|s)$.]{\justifying Circuit for estimating conditional logical probabilities $p(\lambda|s)$. A noiseless reference qubit $R$ is prepared in a Bell pair with a physical qubit, which is subsequently encoded into a logical qubit, forming the entangled state $\frac{1}{\sqrt{2}}(\ket{0_R 0_L} + \ket{1_R 1_L})$ via a single noiseless stabilizer measurement. The logical system $Q$ then undergoes one round of noise (code capacity) or $d$ rounds (phenomenological/circuit-level) of noise $\mathcal{N}$ and stabilizer readout $M_S$, with the final round assumed noiseless. At each round, fresh ancilla qubits initialized to $\ket{0}$ are introduced to facilitate the stabilizer measurements and subsequently discarded. Stabilizer outcomes are tracked as changes relative to the previous round rather than against a fixed $+1$ reference, accommodating arbitrary initial states. Finally, a Bell measurement $M_\text{Bell}$ on $RQ$ reveals any logical flips, from which the conditional logical probabilities are extracted.}
    \label{fig:CI-setting-msmts}
\end{figure}

\subsection{Model architecture} \label{sec:model-architecture}
\begin{figure*}
    \centering
    \includegraphics[width=0.9\linewidth]{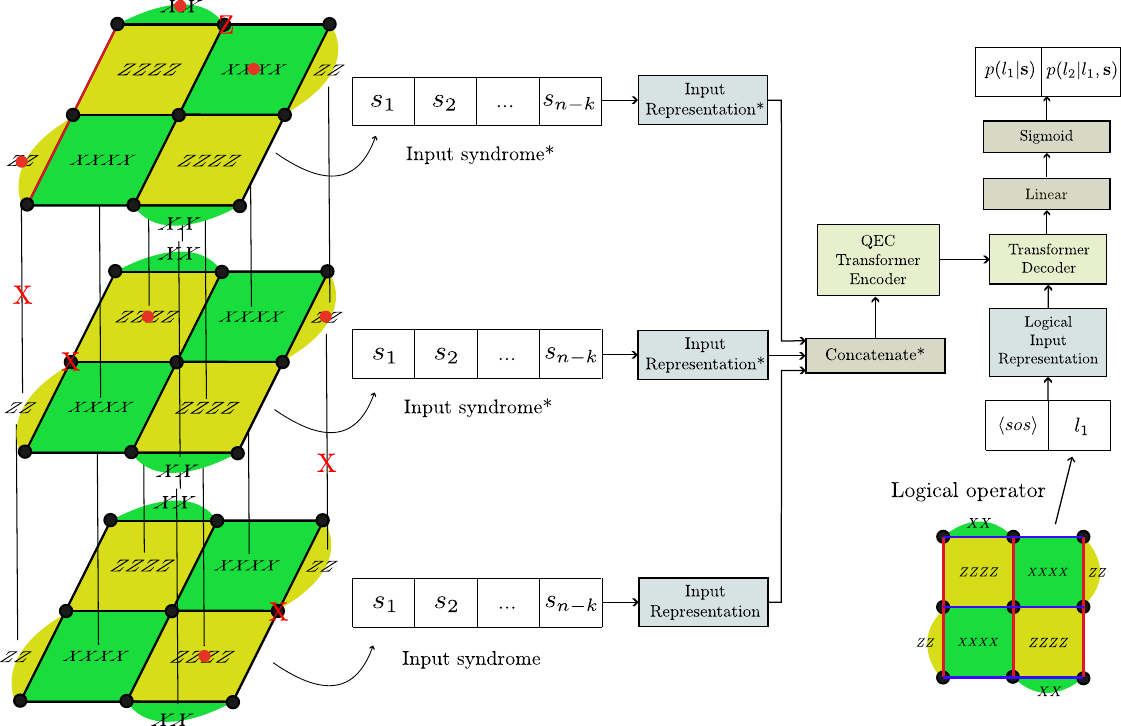}
    \caption{\justifying Model architecture. Full network overview with circuit-level noise affecting data and ancilla qubits. Stabilizer measurements are repeated over $d$ rounds, introducing both spacelike and timelike error correlations (e.g., a timelike $Z$-error connecting two weight-2 stabilizers after round two). Details about the individual model components are shown in Fig.~\ref{fig:structure-depolarizing-components}.}
    \label{fig:structure-depolarizing}
\end{figure*}
\begin{figure*}
    \centering
    \includegraphics[width=\linewidth]{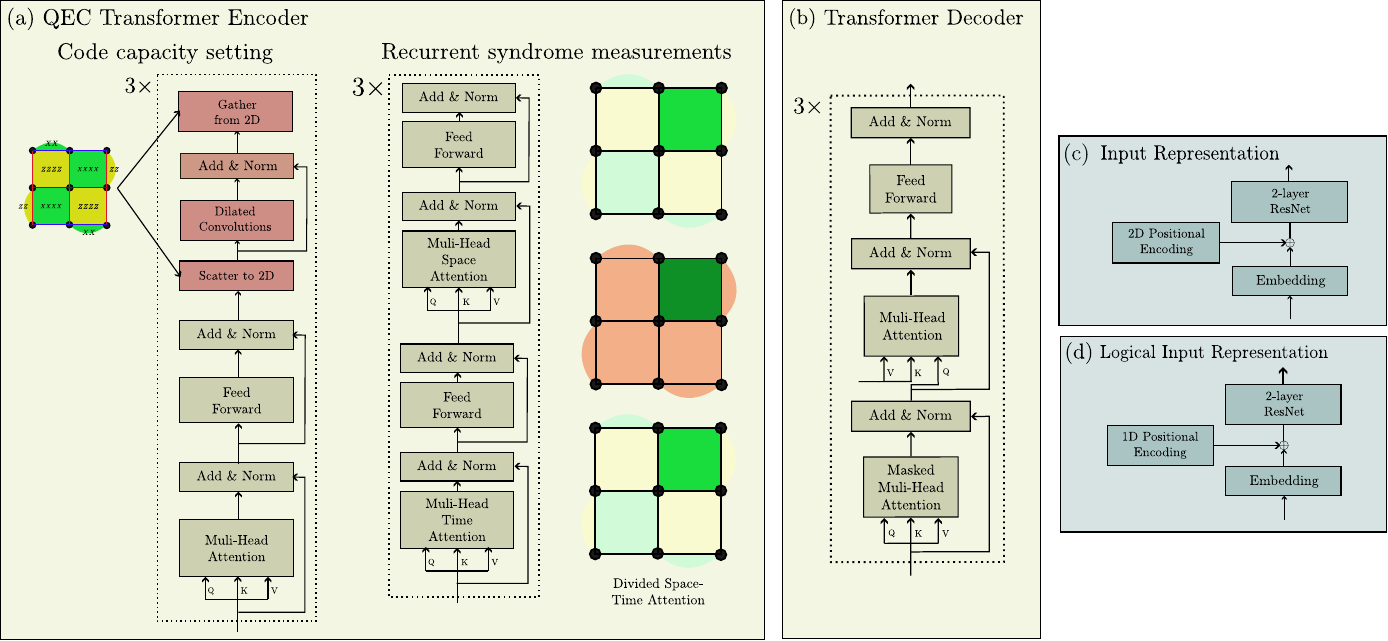}
    \caption{\justifying Model components.   
    (a) Encoder structure. For depolarizing noise, three layers combine multi-head attention, feedforward, and convolutional modules to capture short- and long-range syndrome dependencies. For circuit-level and phenomenological noise, a divided space-time attention mechanism is used: attention alternates between the time axis (tokens at the same spatial position across rounds) and the spatial axis (tokens within one round). Illustration of divided space-time attention: the current token (dark green) attends to same-position tokens across time (green), then to same-round tokens in space (red).
    (b) The decoder processes logical tokens using masked self-attention and cross-attention to all encoded stabilizer tokens, estimating $p(\lambda|s)$ autoregressively.  
    (c) Stabilizer and (d) logical tokens are embedded using positional encodings and residual blocks; a separate embedding is used for the final perfect measurement round.}
    \label{fig:structure-depolarizing-components}
\end{figure*}
We employ a transformer encoder-decoder architecture to estimate $p(\lambda | s)$. Transformers \cite{vaswani_attention_2023} are sequence-to-sequence models that process discrete input tokens via the attention mechanisms, allowing the model to learn correlations between tokens across the full input sequence. Originally developed for natural language processing \cite{vaswani_attention_2023, devlin_bert_2019, brown_language_2020}, they have since been successfully applied to probability density estimation \cite{fakoor_trade_2020} and, more recently, to \gls{qec} \cite{cao_qecgpt_2023, choukroun_deep_2024, wang_transformer-qec_2023, bausch_learning_2024, blue_machine_2025}. In our setting, each stabilizer measurement corresponds to one token, and the attention mechanism captures correlations in the syndrome across space and time. We refer to Figs.~\ref{fig:structure-depolarizing} and \ref{fig:structure-depolarizing-components} for a full overview of our architecture, and list all hyperparameters in App.~\ref{app:hyperparameters}, Tab.~\ref{tab:hyperparameters}.

\paragraph{Input representation:}
Each stabilizer measurement is embedded into a $D$-dimensional vector and combined with a fixed sinusoidal positional encoding \cite{vaswani_attention_2023} that encodes the spatial position (and temporal position for $r > 1$ rounds) of each stabilizer token. Specifically, we use a two-dimensional encoding ($m=2$) for code capacity noise and a three-dimensional encoding ($m=3$) for phenomenological and circuit-level noise, with the surface code placed on a grid of size $(2d+1) \times r$. The embedded tokens are further refined by a two-layer residual network before being passed to the encoder. Logical operator tokens are embedded analogously with a one-dimensional positional encoding.

\paragraph{Encoding:} For the code capacity setting ($r=1$), we extend the encoder layers with two-dimensional convolutions alongside the attention mechanism. The convolutional layers capture local spatial correlations arising from isolated data qubit errors, while the attention mechanism accounts for longer error chains connecting stabilizers further apart, inspired by Ref.~\cite{bausch_learning_2024}.
\paragraph{Divided space-time attention:}
For noise models beyond code capacity, stabilizer measurements are faulty and must be repeated over $d$ rounds, introducing both spacelike and timelike error 
correlations. A key design challenge is to treat these correlations on an equal footing. Recurrent transformer architectures \cite{bausch_learning_2024} introduce a bias toward spatially local errors over timelike ones, while temporal pooling \cite{choukroun_deep_2024} discards much of the temporal structure. Local attention windows \cite{blue_machine_2025} fail to capture non-local error chains that span the full code distance. Full attention over all $d(d^2-1)$ tokens \cite{wang_transformer-qec_2023} avoids these issues but becomes computationally prohibitive at large distances.

To address this, we introduce a divided space-time attention mechanism \cite{bertasius_is_2021, lee_cast_2024}, which, to the best of our knowledge, has not been applied to \gls{qec} prior to this work. Rather than attending over all tokens simultaneously, each token attends separately to all tokens at the same spatial position across time (temporal attention) and then to all tokens within the same measurement round (spatial attention). This reduces the number of tokens each token attends to from $d(d^2-1)$ to $(d^2 - 1 + d)$, while preserving full global attention in both space and time.

\paragraph{Autoregressive decoder:}
The encoded stabilizer state is passed to a transformer decoder, which estimates the logical operator probabilities autoregressively as $p(\lambda|s) = p(\lambda_z|\lambda_x, s)\, p(\lambda_x|s)$. This factorization is necessary because, under depolarizing noise, the $X$- and $Z$-components of the logical operator are correlated: $p(Y) \neq p(X)p(Z)$. The decoder uses masked self-attention on the logical tokens to enforce this autoregressive structure, followed by cross-attention to all encoded stabilizer tokens across all measurement rounds.

\subsection{Training Details} \label{sec:training-details}
The network is trained in a supervised fashion. The labels, i.e.~the logical operators, are sampled alongside the syndromes, eliminating the need for costly labeling of correctability. All data samples are obtained using \textit{Stim} for an efficient simulation of stabilizer circuits \cite{gidney_stim_2021}. Sampling the data is approximately two orders of magnitude faster than training the network for a single epoch and is therefore almost negligible in terms of runtime. Consequently, we choose an online-training approach, where data is sampled on-the-fly from the data distribution $p(s)$, thereby diminishing the effects of a finite-size dataset.

For each selected noise value, a separate network is trained using data sampled specifically from this noise rate. The chosen noise strengths are listed in the appendix in Tab.~\ref{tab:noise-strengths}.

The learning rate is initially increased linearly over five epochs to gradually accumulate gradients and stabilize early training \cite{liu_variance_2021}, before being annealed following a cosine schedule, smoothly decaying from the initial value down to $10^{-8}$ over $500$ to $750$ epochs. If the training loss plateaus before the end of the decay cycle, one or several warm restarts of the cosine schedule are performed. If the training loss stagnates at a value consistent with conditional collapse (see Sec.~\ref{sec:curriculum-details}), the training is restarted from scratch. We set the initial learning rate to $10^{-4}$ to $10^{-5}$ and schedule the learning rate down to $10^{-8}$. A batch size of $1000$ is chosen, as we observe that a larger batch size improves the stability of the training.

Since the machine learning task is \#P-hard, we expect the training time required for convergence to scale exponentially with the code distance. Accordingly, we employ more frequent warm restarts at higher distances to facilitate optimization.

\subsection{Curriculum learning to avoid conditional collapse} \label{sec:curriculum-details}
As separate models are trained for different noise rates, we observe a conditional collapse, $p(\lambda|s) \approx p(\lambda)$, happening at higher noise rates, where tracking the logical information gets increasingly difficult. The conditional output probability $p(\lambda|s)$ collapses onto the overall probability $p(\lambda)$, causing the decoder to ignore the syndromes and to output only the average logical operator. The training gets stuck in a local minimum.

This phenomenon happens since, at high noise rates, the syndrome becomes too complex, so that the network cannot immediately recognize patterns in the training data, leading the training to a local minimum corresponding to the conditional collapse. The labels are effectively averaged over all syndromes, resulting in an estimate of the \gls{ci} of $I=-\log 2$. This value is already reached at low noise rates, when on average, one error occurs on some data qubit of the \gls{qec} code. Since the syndrome is ignored, each logical operator is predicted with equal probability and no information can be retrieved by the network after the noise channel. 

To prevent a conditional collapse of the output probability, $p(\lambda|s) \approx p(\lambda)$, we train the network for distances $d \geq 5$ using a curriculum learning approach. Curriculum learning describes a strategy in which the samples are presented in a structured manner, rather than being randomly shuffled, to guide the learning process \cite{bengio_curriculum_2009}. This technique helps the network converging and improves the achieved minimum, as the training gets less likely trapped in local minima. 

In our case we use this strategy as follows: We begin by training the network for the lowest chosen noise rate. This noise rate is chosen so that the conditional probability $p(\lambda | s)$ is easy to learn for the majority of the syndromes. This initial model then serves as a pre-trained network, which is subsequently fine-tuned on the syndromes it has already seen while gradually learning correlations in more complex syndromes arising from a higher number of data qubit errors. However, this prevents the training from getting trapped in a local minimum associated with conditional collapse, as the number of unseen (but relevant) syndromes is increased only slightly from noise rate to noise rate. Thereby, we not only avoid the conditional collapse but also reduce the computational cost of training the networks for higher noise rates, as we effectively fine-tune models that have already been trained at slightly lower noise rates.

\subsection{Estimating the Coherent Information} \label{sec:estimating-ci-results}
We evaluate the performance of our transformer-based neural network on estimating the \gls{ci} of the surface code in three different settings: code capacity with depolarizing errors, phenomenological noise, and circuit-level noise. The model's accuracy in approximating $p(\lambda|s)$ is analyzed and compared to theoretical expectations of the \gls{ci}, obtained as in Ref.~\cite{colmenarez_accurate_2024}. 
Yet, directly estimating the \gls{ci} via $p(\lambda|s)$ introduces some instabilities, which we discuss next.
During inference, the syndrome is passed to the encoder alongside a start-of-sequence token, and the logical operator probabilities are predicted autoregressively: first $p(\lambda_x|s)$, then $p(\lambda_z|\lambda_x, s)$ by conditioning on both possible outcomes of the first token.

This autoregressive structure introduces a source of instability in the \gls{ci} estimates. For syndromes with no excited $X$-type stabilizers, the probability of a logical $X$ flip is close to zero, meaning that the case $\lambda_x = 1$ is rarely represented in the training data, as this situation is very unlikely to occur (implying that a logical error has happened without exciting any stabilizer).
Estimating $p(\lambda_z | \lambda_x = 1, s)$ in this regime is therefore unreliable, which can degrade the \gls{ci} estimate. 

Since training the network with \gls{bce} loss is equivalent to maximizing the information that the network can recover from the corrupted state after the noise channel, the \gls{ci} serves as an upper bound for the network predictions (see Sec.~\ref{sec:interpretation-main}). We can therefore use the validation loss as a proxy for the \gls{ci} to obtain more robust results than explicitly computing $I = \sum_s p(s) \sum_\lambda p(\lambda|s) \log p(\lambda|s)$.

This instability is specific to probability estimation and does not affect decoding, where only the most likely logical operator is needed. In that case, the model feeds back only the most probable first token to generate the second, avoiding the problematic low-probability conditioning entirely.
The results for the \gls{ci} are presented in Fig.~\ref{fig:results_CI}. 
We verify the finite-size scaling of the \gls{ci} by rescaling the noise probability according to \( p \rightarrow (p - p_{th}) d^{1/\nu} \), which approximately collapses the curves for different code distances onto a single scaling function.

For the following results, the uncertainties of the network predictions for the \gls{ci} were obtained via $100$ bootstrap resamplings. For the \gls{ci} estimates, the statistical uncertainties are too small to be visible and are hidden behind the markers in Fig.~\ref{fig:results_CI}. The uncertainties of the threshold value and the critical exponent were obtained by performing a finite-size scaling analysis on $250$ resampled datasets. 

\paragraph{Code capacity:} We extract a threshold of \( p_{th} = 0.172(1) \), which is close to the optimal threshold \( p_{th} = 0.189(3) \) reported in Ref.~\cite{bombin_strong_2012}, and a critical exponent \( \nu = 1.81(1) \), close to the one estimated in \cite{huang_coherent_2024}. Notably, our estimated threshold is lower than the noise probability at which the \gls{ci} curves for different distances intersect. We attribute this discrepancy to the exponential increase in decoding complexity with the distance, combined with the fact that we did not scale training time accordingly. As a result, the network underperforms relative to the theoretical maximum the larger the distance, leading to a lower apparent threshold in the scaling analysis.
However the threshold obtained is larger than the $p_{th}=0.14$ achieved by \gls{mwpm} decoding \cite{iolius_decoding_2024}.
For smaller distances \( d = 3 \) and \( d = 5 \), the network closely matches the numerically calculated \gls{ci}, successfully recovering the maximally achievable information after the noise channel. Both curves intersect at \( p_{th}^{(3,5)} = 0.185(1) \). 

\paragraph{Phenomenological noise:} Under phenomenological noise, we estimate the threshold as \( p_{th} = 5.15(1) \times 10^{-2} \) and the critical exponent as \( \nu = 2.64(12) \). Unlike in the depolarizing noise scenario, a consensus on the optimal threshold for phenomenological noise is lacking in the literature. For instance, Ref.~\cite{wang_transformer-qec_2023} reports a threshold of \( p_{th} = 3.8 \times 10^{-2} \) using a neural network decoder with \( d \) rounds of syndrome measurements. In contrast, Rispler et al.~\cite{rispler_random_2024} estimate a higher threshold of \( p_{th} = 6 \times 10^{-2} \) for the toric code, based on a mapping of surface code decoding under phenomenological noise to the random coupled-plaquette gauge model. The latter is expected to yield a threshold close to that of the surface code. However, such threshold estimates for finite-size systems can be highly sensitive to implementation details—such as the number of noise rounds, the measurement strategy for the stabilizers, and whether the final measurement round is assumed to be noiseless. It is therefore plausible that the threshold we observe is influenced by finite-size effects, given the limited distances considered. To facilitate reproducibility and comparison with future studies, we provide the explicit stabilizer measurement circuit used for the distance-three code in Appendix~\ref{sec:circuits}.

In contrast to the code capacity setting with depolarizing noise, where the threshold is typically associated with a zero crossing of the \gls{ci}, we observe an intersection at a nonzero value of the \gls{ci} under phenomenological noise. Huang et al.~\cite{huang_coherent_2024} attributed the zero crossing in the case of independent bit/phase flip to the emergent self-duality of the random-bond Ising model (RBIM) in the thermodynamic limit. Based on our results, we conjecture that this self-duality is broken under phenomenological noise or must be re-interpreted to account for the effect of noisy measurements, leading to the observed offset in the intersection point of the \gls{ci} curves.

\paragraph{Circuit-level noise:} Only few studies have explored the optimal threshold of \gls{qec} codes under circuit-level noise, as decoding syndromes sampled from this noise model is particularly challenging. Moreover, the issues concerning the threshold under phenomenological noise apply even more to circuit-level noise: The threshold is highly dependent on the specifics of the noise model, making it even more complicated to compare results across different works, as no universal noise model for circuit-level noise has been established. The detailed circuit of the syndrome measurements for a distance three code that we use, is shown in Appendix \ref{sec:circuits}.

Nevertheless, rough estimates can still be found. Fowler et al.~\cite{fowler_surface_2012} report a suboptimal threshold of approximately $p_{th} \approx 5.7 \times 10^{-3}$ using \gls{mwpm} decoding, and Rispler et al.~\cite{rispler_random_2024} obtain an approximation of the threshold of $p_{th}=1.4 \times 10^{-2}$ for the toric code by using an effective noise model for circuit-level noise that can be mapped to the random coupled-plaquette gauge model, the same statistical mechanical model that has also been used to estimate the phenomenological noise threshold, but with adapted interaction strengths. This suggests that neural networks decoders that perform well under phenomenological noise, should, at least, achieve reasonable results under circuit-level noise. 
\begin{figure*}
    \centering
    \includegraphics[width=\linewidth]{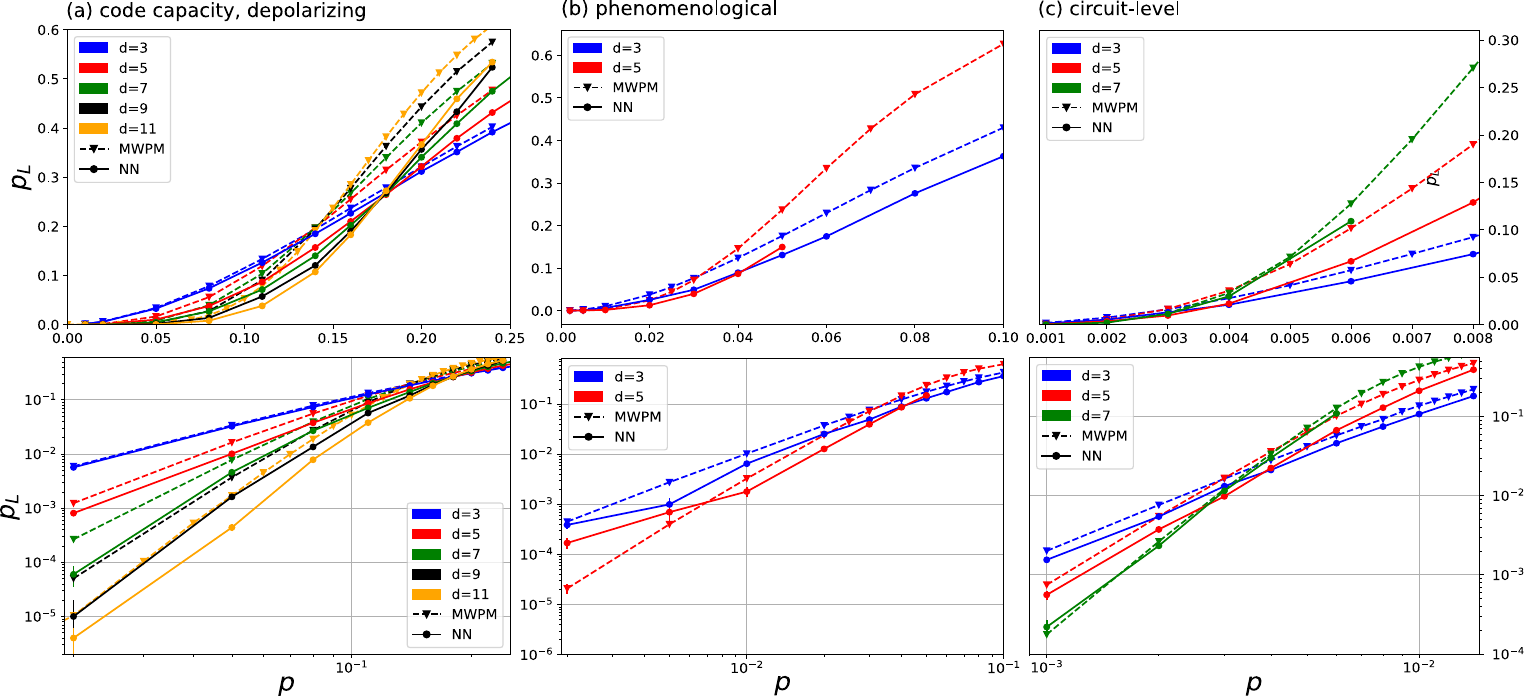}
    \caption[Decoding performance of our neural network compared to MWPM.]{\justifying Decoding performance of our neural network vs.\ the \gls{mwpm} decoder. We show (a) code capacity for $d=3,5,7,9,11$, (b) phenomenological noise for $d=3,5$, and (c) circuit-level noise for $d=3,5,7,9$. Each dot is a network trained at a specific noise probability; dashed lines show the performance of \gls{mwpm}. On the left, logical error rates (linear scale) are plotted, while on the right the scaling behavior (log-log scale) is shown.}
    \label{fig:log-error}
\end{figure*}

The results for the estimated \gls{ci} are shown in Fig.~\ref{fig:results_CI}(c). We obtain a threshold at $p_{th}=3.00(1) \times 10^{-3}$ with a \gls{ci} of around $I \approx 0.95(1)$, and a critical exponent of $\nu = 1.19(4)$. Similar to the phenomenological noise case, the threshold occurs at a nonzero \gls{ci}, implying that much of the information can be recovered during error correction near the threshold for finite distances. In the following section, we further evaluate the network performance in terms of decoding.\\

\subsection{Decoding the surface code} \label{sec:decoding-results}
We choose to approach the decoding problem based on the $STL$-decomposition of an error chain $E$ into a stabilizer $S$, a pure error $T$, and a logical operator $L$, see Sec.~\ref{sec:STL-main}. The goal of the decoder is to correctly classify the logical operator or, more precisely, to estimate the probabilities of each logical operator for a given syndrome. By considering the maximum support logical operator, it is sufficient for decoding to classify the syndrome based on global properties, such as the value of the logical operator, which represents a global parity check. We measure this global parity check alongside the syndrome. The logical operator is generated autoregressively by choosing the most likely logical operator.

Our results for depolarizing noise in the code capacity setting, phenomenological noise, and circuit-level noise are shown in Fig.~\ref{fig:log-error}. The plots in the top row show the logical vs.~physical error rate on a linear scale, illustrating the threshold and performance compared to the \gls{mwpm} algorithm. In contrast, the plots in the bottom row show the scaling of both \gls{mwpm} and our model below threshold, using a log-log scale. We note that our model outperforms \gls{mwpm} significantly and scales similarly, as $\frac{(d+1)}{2}$ below threshold. For depolarizing noise in the code capacity setting, the logical error rates exhibits approximately the same finite-size scaling behavior as the \gls{ci}, with a threshold $p_\mathrm{th} = 0.172(1)$ and a critical exponent $\nu = 1.80(1)$. For phenomenological noise, we find a threshold of $p_\mathrm{th} = 0.043(34)$ and a critical exponent of $\nu = 2.4(4)$. Finally, for circuit-level noise we observe $p_\mathrm{th} = 2.8 (2) \times 10^{-3}$ and $\nu = 1.19(7)$. Under circuit-level noise, we observe the largest performance margin between our method and \gls{mwpm} for distances up to $d=5$, which shrinks for distance $d=7$, possibly due to insufficient training. This indicates that our approach is particularly well-suited for handling complex noise models, as it makes no prior assumptions about the noise and instead learns directly from the raw syndrome data. Notably, the performance gap between our decoder and \gls{mwpm} increases with code distance, as for small distances degeneracy plays a less important role.

\glsresetall
\section{Soft post-selection to improve logical accuracy} \label{sec:post-selection}
In this section, we present both theoretical and numerical details on soft-post selection. In Sec.~\ref{sec:soft-post-selection}, we formally define soft post-selection and state relevant properties. In Sec.~\ref{sec:thresholds-post-selection} we present work on the theoretical threshold under post-selection. Finally, we numerically verify the stated properties with the previously presented neural network decoder in Sec.~\ref{sec:post-selection-results}.

\subsection{Soft post-selection} \label{sec:soft-post-selection}
We start the chapter with a formal definition of post-selection and present some relevant theorems. All statements in this and the following subsection are proven in App.~\ref{sec:post-selection-appendix}. Again, we consider a single logical qubit $k=1$.

Formally, we define soft post-selection in terms of a selection function $f$ that maps a post-selection parameter $c$ onto a set of accepted syndromes $\{ s_1, ..., s_n\} \in \mathcal{P}(\mathcal{S})$. In the following, $\mathcal{S}$ denotes the set of measured stabilizers, i.e.~the generators of the stabilizer group, and therefore $\mathcal{P}(\mathcal{S})$ the set of all subsets of syndromes. As the logical error rate is only well-defined if $p_\mathrm{abort} < 1$, we consider only post-selection schemes $f$ with post-selection parameters $c$ and noise rates such that $f(c) \neq \emptyset$.
\begin{definition} \label{def:std-mld-post-selection}
Standard \gls{mld} post-selection. Let $c \in [0,1]$. Then the \gls{mld} post-selection function is given by 
\begin{equation}
    \begin{aligned}
        &f_\mathrm{MLD}: [0,1] \rightarrow \mathcal{P}(\mathcal{S}), \\ &f_\mathrm{MLD}(c) = \{ s \in \mathcal{S}: \underset{\lambda \in \Lambda}{\max} \  p(\lambda | s) > c\}.
    \end{aligned}
\end{equation}
\end{definition}
Since our model predicts the logical operators autoregressively, it naturally 
outputs a confidence score for each logical token individually. This motivates 
a per-token confidence threshold, where a syndrome is accepted only if the model 
is sufficiently confident about both $\lambda_x$ and $\lambda_z$ separately. 
We refer to this as split \gls{mld} post-selection, indexed by NN.
\begin{definition} \label{def:split-post-selection}
Split \gls{mld} post-selection: Let $c \in [0,1]$. Then the split MLD post-selection function is given by 
\begin{equation}
\begin{aligned}
&f_\mathrm{NN}: [0,1] \rightarrow \mathcal{P}(\mathcal{S}), \\ 
&f_\mathrm{NN}(c) = \{ s \in \mathcal{S}: \ (p(\lambda_x | s) > \sqrt{c} \ \lor \ p(\lambda_x | s) < 1 - \sqrt{c})  \\
&\quad  \land \ (p(\lambda_z | \lambda_x, s) > \sqrt{c} \ \lor \ p(\lambda_z | \lambda_x, s) < 1 - \sqrt{c}) \}.
    \end{aligned}
    \end{equation}
\end{definition}

\begin{figure}
    \centering
    \includegraphics[width=\linewidth]{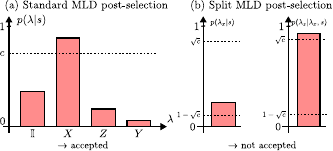}
    \caption{\justifying Visualization of (a) standard MLD post-selection and (b) split MLD post-selection. (a) Shown are the coset probabilities $p(\lambda|s)$ for each possible logical configuration. As $p(X|s)>c$, the syndrome is accepted. (b) Shown are the probabilities of the first logical ($X$) operator and of the second logical operator ($Z$), conditioned on the value of the first one. We only accept the syndrome when certain that a logical operator is likely (close to $1$ probability) or unlikely (close to $0$) probability to appear. As $1-\sqrt{c}<p(\lambda_x|s)<\sqrt{c}$, the syndrome gets discarded.}
    \label{fig:post-selection-visualization}
\end{figure}
While this scheme naturally emerges from the autoregressive output structure, 
it can also be applied under \gls{mld} directly, since the per-token 
probabilities can be recovered from the full coset probabilities $p(\lambda|s)$ via
\begin{align*}
    &p(\lambda_x|s) = \sum_{\lambda_z} p(\lambda_x, \lambda_z |s), \\
    & p(\lambda_z|s, \lambda_x = \hat{\lambda}_x) = \frac{p(\lambda_x = \hat{\lambda}_x, \lambda_z | s)}{p(\lambda_x = \hat{\lambda}_x |s)},
\end{align*}
enabling split MLD post-selection, using the MLD coset probabilities. Both post-selection schemes are visualized in Fig.~\ref{fig:post-selection-visualization}. To define a unified framework for post-selection under \gls{mld}, we define \gls{mld} post-selection as follows.
\begin{definition} \label{def:mld-post-selection}
    \gls{mld} post-selection: A post-selection scheme $f$ is an \gls{mld} post-selection scheme if syndromes are accepted or discarded based on the information available in the \gls{mld} coset probabilities and for post-selection parameter $c \in [0,1]$ we have $\underset{\lambda \in \Lambda}{\mathrm{max}} \ p(\lambda|s) > c$, i.e.~$f(c) \subset f_\mathrm{MLD}(c)$.
\end{definition}
Split \gls{mld} post-selection is an \gls{mld} post-selection scheme in the setting of Def.~\ref{def:mld-post-selection} and is stricter than standard \gls{mld} post-selection for the same post-selection parameter, i.e.~split \gls{mld} post-selection always discards at least as many syndromes as \gls{mld} post-selection. Hence:
\begin{corollary}
Let $s \in \mathcal{S}$ and $c \in [0,1]$. Then $f_\mathrm{NN}(c) \subset f_\mathrm{MLD}(c)$.
\end{corollary}

To quantify the fraction of discarded samples, we introduce the abort probability.
\begin{definition} \label{def:abort-probability}
Let $c \in [0,1]$ and $f: c \rightarrow \mathcal{P}(\mathcal{S})$ a post-selection function. Then the abort probability is defined by
\[
p_\mathrm{abort}(c, f) = \mathbb
P(f(c)^\mathcal{C}) = \sum_{s \notin f(c)} p(s).
\]
\end{definition}

The logical error rate $p_L$ and its complement, the success rate $p_\mathrm{succ}$, quantify how well \gls{mld} performs on the set of accepted syndromes, i.e. how likely a logical error is introduced after decoding on the set of accepted syndromes.
We find that post-selection schemes based on \gls{mld} directly impose an upper bound on the logical error rate, and the schemes are monotonous, meaning that discarding more syndromes leads to a lower logical error rate.
\begin{corollary} \label{cor:bound}
    Let $f$ be an MLD post-selection function. Let $c \in [0,1]$. Then $p_\mathrm{succ}^c > c$ or equivalently, $p_\mathrm{L}^c < 1-c$.
\end{corollary}

\begin{corollary} \label{cor:monotony}
    Let $f$ be an MLD post-selection function. Let $c,c^\prime \in [0,1]$ with $c < c^\prime$. Then $p_\mathrm{succ}^c < p_\mathrm{succ}^{c^\prime}$ or equivalently, $p_\mathrm{L}^c > p_\mathrm{L}^{c^\prime}$.
\end{corollary}
For both presented schemes, $c<0.25$ accepts all syndromes, while for $c=1$ all syndromes are discarded as $\underset{\lambda \in \Lambda}{\mathrm{max}} \ p(\lambda|s) = 1$ is only reached in the thermodynamic limit. Thus, the case $c<0.25$ reduces to \gls{mld} without post-selection and the previous corollary also proves that post-selection always decreases the logical error rate. 

With this, we formulate the main theorem of this chapter. \Gls{mld}-based post-selection is optimal in the sense that no other post-selection scheme can achieve a lower logical error rate at the same abort probability:
\begin{theorem} \label{thm:main}
Let $c, c^\prime \in [0, 1]$, $f$ an MLD post-selection function and $f^\prime: [0,1] \rightarrow \mathcal{P}(\mathcal{S})$ any other post-selection scheme. Then 
\[
p_\mathrm{abort}(c^\prime, f^\prime) \leq p_\mathrm{abort}(c, f) \Longrightarrow p_L(c^\prime, f^\prime) \geq p_L(c, f).
\]
\end{theorem}
In the proof of the theorem, we find that under \gls{mld} the logical error rate / success probability can be rewritten as (see App.~\ref{sec:post-selection-appendix})
\begin{equation}
    p_\mathrm{succ} = \sum_{s \in f(c)} \Tilde{p}(s) \  \underset{\lambda \in \Lambda}{\mathrm{max}} \ p(\lambda | s). \label{eq:p-succ}
\end{equation}
where $\Tilde{p}(s) = \frac{p(s)}{\sum_{s^\prime} p(s^\prime)}$ and we sum over $f(c)$, which is the set of accepted syndromes. For each syndrome, the probability of successfully correcting is equal $\underset{\lambda \in \Lambda}{\mathrm{max}} \ p(\lambda | s)$. 

The optimality stated in Thm.~\ref{thm:main} characterizes the trade-off between the abort probability and the logical error rate among accepted samples. As we demonstrate in Fig.~\ref{fig:pL-vs-pabort} in the appendix, this implies that discarding syndromes and committing logical errors are balanced such that the ratio \(\frac{p_L}{p_\mathrm{abort}}\) converges to a constant in the limit \(n \rightarrow \infty\).

\subsection{Thresholds under optimal post-selection} \label{sec:thresholds-post-selection}
\begin{figure*}
    \centering
    \includegraphics[width=\linewidth]{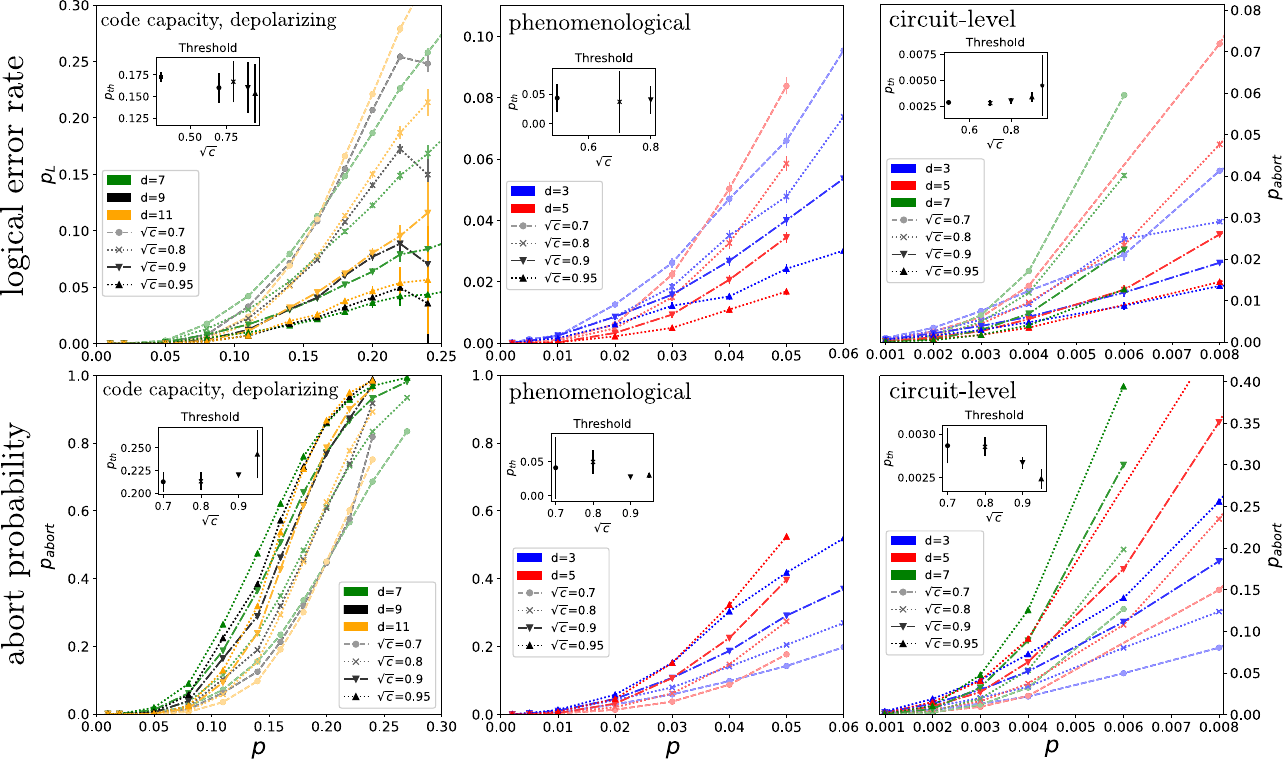}
    \caption[Decoding performance of our NN decoder with post-selection based on output probabilities.]{\justifying Decoding performance of our NN decoder with soft post-selection based on output probabilities. (upper panels) Logical error rates for post-selection thresholds \( \sqrt{c} = 0.7, 0.8, 0.9, 0.95 \). Higher thresholds reduce logical error rates. Inset: observed threshold vs.~post-selection parameter. (lower panels) Abort probability \( p_\mathrm{abort} \) increases with stricter post-selection and shows threshold behavior. Inset: threshold closely matches that of decoding without post-selection and is largely independent of the post-selection parameter \( c \). Black boxes mark crossings.
    }
    \label{fig:abort-results}
\end{figure*}
Below threshold, \gls{qec} codes can reliably protect the logical information by scaling up the encoding. This is proven in form of the threshold theorem \cite{aharonov_fault-tolerant_2008,kitaev_fault-tolerant_2003,knill_resilient_1998,shor_fault-tolerant_1997}:
\begin{equation}
\lim_{d \to \infty} p_L(d, p) =
\begin{cases} 
0, & \text{if } p < p_{\text{th}}, \\
const. > 0, & \text{if } p > p_{\text{th}}.
\end{cases}
\end{equation}
Under \gls{mld}, we find that the maximum coset probability converges to unity in probability, given the error rate is below threshold. 
\begin{theorem} \label{thm:abort-threshold}
    Let $g_d: \mathcal{S} \rightarrow [0,1], g_d(s)=\underset{\lambda \in \Lambda}{\mathrm{max}} \ p(\lambda | s)$ be the mapping of a syndrome to its maximum coset probability for a distance-$d$ \gls{qec} code. Then, for $p<p_{th}$ we have
    \begin{equation}
        g_n \overset{P}{\rightarrow} 1 \ \mathrm{as} \ d \rightarrow \infty.
    \end{equation}
\end{theorem}
Thus, a threshold for the abort probability under \gls{mld} post-selection emerges naturally. This threshold coincides with the threshold $p_{th}$ of standard \gls{qec} \cite{english_thresholds_2024}. For any post-selection parameter $c$, the abort probability vanishes as the probability of syndromes with $g_d < c$ vanishes.
\begin{equation}
    \forall c>0: \ P(|g_d - 1| > c) \rightarrow 0 \ \mathrm{as} \ d \rightarrow \infty.
\end{equation}
Additionally, from the threshold theorem, we follow directly that
\begin{align}
    \mathbb{E} g_d = \sum_s p(s) \ \underset{\lambda \in \Lambda}{\mathrm{max}} \ p(\lambda | s) \rightarrow 1.
\end{align}
This does not follow from Thm.~\ref{thm:abort-threshold}, as convergence in probability does not imply convergence of the expectation value.

\subsection{Soft post-selection on the network outputs} \label{sec:post-selection-results}
In this section, we numerically demonstrate the effectiveness of the \gls{mld} post-selection schemes introduced earlier. Since our model outputs the probabilities of both logical operators in an autoregressive manner, it naturally lends itself to the split \gls{mld} post-selection scheme. 

The post-selection parameter $c$ controls the strictness of the confidence threshold: for $c=1$ all syndromes are rejected; for $c=0$ all syndromes are accepted unconditionally. In practice, no syndromes are discarded for $\sqrt{c} \leq 0.5$, as the network needs to predict either $0$ (no logical operator) or $1$ (a logical $X$ or $Z$). Fig.~\ref{fig:abort-results} shows the logical error rate and abort probability of our neural network decoder under split \gls{mld} post-selection for $\sqrt{c} \in \{ 0.7, 0.8, 0.9, 0.95 \}$, evaluated under depolarizing noise in the code capacity setting, phenomenological noise, and circuit-level noise. As expected from Cor.~\ref{cor:monotony}, the logical error rate consistently decreases with increasing post-selection parameter $c$, accompanied by a corresponding increase in the abort probability.

\paragraph{Threshold behavior and extraction.} Both the logical error rate and the abort probability exhibit a clear threshold behavior, confirming that \gls{mld} post-selection preserves the threshold of standard decoding without post-selection. We extract threshold values via a finite-size scaling analysis using the ansatz $p \rightarrow (p - p_\mathrm{th}) d^{\frac{1}{\nu}}$, with uncertainties estimated by bootstrap resampling. The extracted thresholds for different post-selection parameters are shown in the insets of Fig.~\ref{fig:abort-results}, and remain largely independent of $c$. The logical error rate exhibits a less regular behavior than the abort probability; post-selection reduces the logical error rate particularly at high noise rates, where a greater number of syndromes are discarded, effectively preventing the logical error rate from increasing with the noise probability (see also Fig.~\ref{fig:compare-finite-size-effects}).

\paragraph{Finite-size effects.}
We observe notable finite-size effects in the form of sudden increases and drops, particularly at low distances $d=3,5$ (see Appendix~\ref{app:finite-size-effects}). Consequently, we exclude both distances from the threshold analysis in the code capacity setting. Under phenomenological and circuit-level noise, we observe similar effects; however, since our decoder was trained only on small code distances, we retain these instances in the analysis for consistency.

\paragraph{Scalability.}
A frequent criticism of (soft) post-selection is its limited scalability \cite{smith_mitigating_2024}. However, as we observe that both logical error rate and abort probability exhibit a threshold behavior, they diminish below threshold as the code distance increases. This demonstrates that \gls{mld} post-selection remains efficient in the large-code limit when compared to no post-selection, thereby enabling scalable post-selection based on the model confidence.

\section{Conclusion} \label{sec:conclusion}
\glsresetall

In this work, we have reformulated the \gls{ci} in terms of the conditional distribution $p(\lambda | s)$ over logical operators $\lambda$ given a syndrome $s$. This yields a novel interpretation of neural decoders trained with \gls{bce}: such training corresponds to maximizing the recoverable logical information, bounded above by the \gls{ci}. Thus, neural networks do not merely approximate \gls{mld}, but serve as variational approximators of the \gls{ci} itself. This interpretation links the training objective directly to the goal of preserving logical information through noisy quantum channels.

\begin{table*}[t]
\centering

\begin{tabular}{lccc}
\toprule
\textbf{Noise model} & \textbf{Estimated threshold} & \textbf{Reference threshold} & \textbf{Critical exponent} $\nu$ \\
\midrule
Depolarizing (code capacity) 
    & $p_\mathrm{depolarizing} = 0.172(1)$ 
    & $p_{th} = 0.189(3)$ \cite{bombin_strong_2012} 
    & $1.81(1)$ \\[6pt]
Phenomenological 
    & $p_\mathrm{phenomenological} = 5.15(1)\times 10^{-2}$ 
    & $p_{th} = 0.038$ \cite{wang_transformer-qec_2023}; $p_{th} = 0.06$ \cite{rispler_random_2024} 
    & $2.64(12)$ \\[6pt]
Circuit-level 
    & $p_\mathrm{circuit\text{-}level} = 3.00(1)\times 10^{-3}$ 
    & $p_{th} \approx 0.01$ \cite{fowler_surface_2012}; $p_{th} \approx 0.014$ \cite{rispler_random_2024} 
    & $1.19(4)$ \\
\bottomrule
\end{tabular}
\caption{\justifying Summary of estimated thresholds and critical exponents extracted using the finite size scaling ansatz \( p \rightarrow (p - p_{th}) d^{1/\nu} \) for three noise models. Reference values are provided for comparison.}
\label{tab:thresholds}
\end{table*}

Building on this theoretical insight, we have introduced a transformer-based neural network architecture designed to estimate $p(\lambda | s)$ for surface codes under various noise models, effectively tracking logical operators across noisy quantum channels, similar to a neural network decoder \cite{torlai_neural_2017, chamberland_deep_2018,varsamopoulos_decoding_2017,maskara_advantages_2019,wagner_symmetries_2020,meinerz_scalable_2022,overwater_neural-network_2022,cao_qecgpt_2023,wang_transformer-qec_2023,blue_machine_2025}. Our approach departs from prior works \cite{cao_qecgpt_2023,choukroun_deep_2024,wang_transformer-qec_2023,bausch_learning_2024,blue_machine_2025} in several important ways. First we track the maximum support logical operators, ensuring that all training data contributes to the learning task. Second, we employ a structure that is purely built upon the transformer encoder-decoder architecture \cite{vaswani_attention_2023} and use an autoregressive transformer decoder to model correlations between the logical operators. Third, we implement a divided space-time attention scheme for recurrent syndrome measurements, inspired by video understanding architectures \cite{bertasius_is_2021}, alternating between spatial and temporal attention layers. This improves expressiveness and balances sensitivity to data and measurement errors, overcoming the limitations of prior models, as previous architectures either bias models toward spacelike errors \cite{choukroun_deep_2024,bausch_learning_2024}, incur high runtime from full self-attention \cite{wang_transformer-qec_2023}, or limit the capture of non-local error chains by restricting attention windows \cite{blue_machine_2025}.

Our architecture is trained for depolarizing noise in the code capacity setting, phenomenological noise, and circuit-level noise models. To ensure effective training across these regimes, we designed a curriculum learning schedule that mitigates conditional collapse, $p(\lambda|s) \approx p(\lambda)$, and accelerates convergence, particularly in high-noise regions. We observe that the model performs best in the intermediate noise regime, where sufficient structure is present to learn from.

We have demonstrated that our model is able to successfully estimate the \gls{ci}, enabling \gls{qec} threshold and critical exponent estimation via finite-size scaling analysis. The results for all three noise models are summarized in Table~\ref{tab:thresholds}. These results highlight finite-size effects in small surface codes (e.g.,~$d=3,5$), particularly under circuit-level noise, though the observed scaling behavior supports the use of CI estimation as a practical method for characterizing decoder performance. In all noise regimes, our model also outperforms \gls{mwpm} \cite{fowler_minimum_2014} in terms of logical error rate.

Furthermore, we have explored soft post-selection strategies, where the information available in the syndrome is used to decide whether to discard runs with ambiguous logical inference \cite{smith_mitigating_2024,english_thresholds_2024}. We have proved that \gls{mld}-based post-selection is optimal and scalable: \gls{mld}-based post-selection achieves the lowest logical error rate among different post-selection schemes, balancing discarding syndromes and accepting logical errors. In the large-distance limit, the maximum coset probability converges to unity, leading to vanishing abort probabilities and aligning the thresholds under post-selection and for standard quantum computations. Based on this insight, we have proposed a novel split \gls{mld} post-selection scheme that treats uncertainty in both, $X$- and $Z$-type logical operators separately. This approach further reduces logical errors by discarding syndromes that exhibit a disproportionately high likelihood of inducing logical failures. We have performed split \gls{mld} post-selection based on the neural network output, yielding substantial improvements: the abort probabilities exhibit clear threshold behavior, and logical error rates are significantly lower than those achieved with standard decoding.

Altogether, our work further establishes transformer-based neural network decoders not only as competitive decoders, but also as computationally tractable approximators of information-theoretic quantities. This has significant implications. First, it shows that despite the hardness of \gls{qec} decoding, a neural network trained in reasonable time can estimate the coherent information with high accuracy. Second, it demonstrates that the \gls{ci} can serve as a practical, and physically meaningful performance metric for quantum error correction under general noise. Third, our autoregressive modeling of 
$p(\lambda | s)$ opens new pathways for decoding multiple logical qubits; by filtering based on the model's confidence, we introduce optimal soft post-selection into \gls{qec} protocols based on \gls{mld}.

While we have focused on surface codes, our method does not rely on their specific structure and can likely be adapted to color codes \cite{bombin_topological_2006}, hypergraph product codes \cite{Tillich_2014}, and other quantum low-density parity check (qLDPC) code families \cite{breuckmann_quantum_2021}. The architecture’s autoregressive nature is especially well-suited for codes encoding multiple logical qubits. Finally, adapting these methods to experimental circuit-level noise remains an important direction for future research.

\section*{Code and Data availability}
All code used to simulate the \gls{qec} circuits and to train our models is available at \url{https://github.com/Dominik5431/LearningOptimalQuantumErrorCorrectionThresholds}. The corresponding data, including the trained model checkpoints, can be accessed via \url{https://zenodo.org/records/19655235}. All machine learning models in this work have been implemented using PyTorch \cite{paszke_pytorch_2019} and numerical simulations have been carried out with NumPy \cite{harris_array_2020} and SciPy \cite{virtanen_scipy_2020}.

\begin{acknowledgments}
L.C. and M.M. gratefully acknowledge funding by the U.S. ARO Grant No. W911NF-21-1-0007. M.M. furthermore acknowledges funding from the European Union’s Horizon Europe research and innovation programme under grant agreement No 101114305 (“MILLENION-SGA1” EU Project), and the German Federal Ministry of Research, Technology and Space (BMFTR) as part of the Research Program Quantum Systems, research project 13N17317 (”SQale”), and MUNIQC-ATOMS (Grant No. 13N16070), as well as NeuQuant (Grant No. 13N17066). This research is also part of the Munich Quantum Valley (K-8 on Hardware-adapted Theory), which is supported by the Bavarian state government with funds from the Hightech Agenda Bayern Plus. M.M. acknowledges funding from
the ERC Starting Grant QNets through Grant No. 804247. L.C.~and M.M. also acknowledge support for the research that was sponsored by IARPA and the Army Research Office, under the Entangled Logical Qubits program, and was accomplished under Cooperative Agreement Number W911NF-23-2-0216. The views and conclusions contained in this document are those of the authors and should not be interpreted as representing the official policies, either expressed or implied, of IARPA, the Army Research Office, or the U.S. Government. The U.S. Government is authorized to reproduce and distribute reprints for Government purposes notwithstanding any copyright notation herein. M.M. acknowledges support from the Deutsche Forschungsgemeinschaft (DFG, German Research Foundation) under  Germany’s Excellence Strategy Cluster of Excellence Matter and Light for  Quantum Computing (ML4Q) EXC 2004/1 390534769. 
MS acknowledges support by the Helmholtz Initiative and Networking Fund, Grant No.~VH-NG-1711.
The
authors gratefully acknowledge the computing time provided to them at the NHR Center NHR4CES at RWTH
Aachen University (Project No. p0020074 and rwth1871). This is
funded by the Federal Ministry of Education and Research and the state governments participating on the
basis of the resolutions of the GWK for national high
performance computing at universities.
\end{acknowledgments}

\bibliography{bibliography2.bib}

\appendix

\clearpage
\section{The surface code}
\begin{figure}
    \centering
    \includegraphics[width=\linewidth]{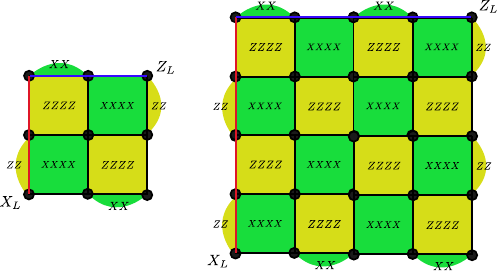}
    \caption[The (rotated) surface code: $d=3$ and $d=5$.]{\justifying The (rotated) surface code: $d=3$ and $d=5$. The data qubits (black dots) are sitting on the vertices of a square lattice of size $d$. Yellow (Green) plaquettes correspond to $Z$ ($X$) stabilizers. At the boundary of the respective logical operator, the weight-four stabilizer are reduced to weight-2 stabilizers. One possible $Z$ ($X$) type logical operator is shown in blue (red). The rotated surface code is a $[[d^2,1,d]]$ code.}
    \label{fig:surface-code}
\end{figure}
In this section, we describe the topological \gls{qec} code used throughout this work. Topological \gls{qec} codes are a prominent class of error-correcting codes that utilize the principles of topology to protect quantum information. Here, topology describes global features that remain invariant under local perturbations: the information is encoded in a way that spreads the information globally across the system, ensuring it is resilient to local disturbances such as local Pauli noise. The surface code is among the most prominent representatives of topological \gls{qec} \cite{kitaev_fault-tolerant_2003,dennis_topological_2002,fowler_surface_2012}. It encodes logical qubits into a two-dimensional lattice of physical qubits. It is particularly well-suited to hardware architectures with local interactions and has become one of the most promising candidates for fault-tolerant quantum computation.

In the rotated surface code, physical qubits are arranged on a square lattice of length $L$ with two types of stabilizer operators (see Fig.~\ref{fig:surface-code}): plaquette ($Z$-type) and star ($X$-type) operators. Each stabilizer acts on four (bulk) or two (boundary) neighboring qubits and enforces local parity constraints. Errors are detected by measuring changes in stabilizer outcomes, known as syndromes, which reveal the presence and location of bit-flip ($X$) or phase-flip ($Z$) errors.

Logical qubits are encoded by defining operators that span the lattice. Specifically, logical $X$- and $Z$-type operators correspond to strings of Pauli $X$ or $Z$ operators that cross the lattice from one boundary to the opposite. Due to the topological nature of the code, a logical error requires a chain of physical errors connecting opposite boundaries. The code distance is equal to the length of the boundary, $d=L$.

\section{Curriculum Learning}
\begin{figure*}
    \centering
    \includegraphics[width=\textwidth]{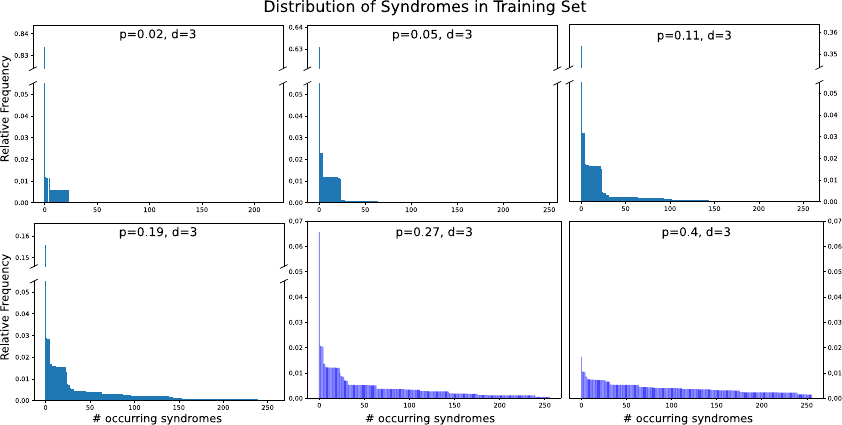}
    \caption[Illustration of the complexity of the syndromes in the training set in dependency of the noise strength of the depolarizing channel.]{\justifying Illustration of the increasing complexity of syndromes in the training set as a function of the depolarizing noise strength in the code capacity setting for a $d=3$ surface code. The plots show the frequency distribution of measured syndromes for $p=0.02$, $p=0.05$, $p=0.11$, $p=0.19$, $p=0.27$, and $p=0.4$. Syndromes are ordered by their frequency of occurrence, starting with the trivial syndrome which is the most likely syndrome as long as $p<0.5$. As the noise rate increases, the frequency of the trivial syndrome decreases, while easily-decodable syndromes (syndrome with only few excited stabilizers and accordingly few erroneous qubits) become more common. With further increase in noise, more complex syndromes with a high number of excitations occur more frequently and the syndrome distribution approaches a uniform distribution. The hierarchical structure of the syndrome becomes apparent.}
    \label{fig:curriculum-d3}
\end{figure*}
In Fig.~\ref{fig:curriculum-d3}, we illustrate the curriculum learning scheme using the syndrome distribution for a surface code of distance $d = 3$ under depolarizing noise in the code capacity setting (although this approach is not applied at $d = 3$ for this particular noise model). The figure shows the distribution of all $256$ possible syndromes, sorted and enumerated by their frequency of occurrence, across different noise strengths. For all values of $p$, the trivial (all-zero) syndrome is the most probable, as the likelihood of no error remains higher than that of any error as long as $p < 0.5$. As the noise strength increases, the frequency of the trivial syndrome decreases, while syndromes resulting from single-qubit errors become more prominent before eventually being surpassed by more complex error patterns.

When pretraining on low noise rates, the network learns to estimate logical operators for the simple syndromes first—such as the trivial syndrome and those resulting from single-qubit errors, thereby avoiding conditional collapse, where $p(\lambda | s) \approx p(\lambda)$ due to convergence to a poor local minimum. This pretrained model is then fine-tuned on higher noise rates, where it can refine the estimated probabilities for already encountered syndromes and begin to capture the correlations present in the more complex syndromes that dominate at higher noise levels. As a result, this curriculum learning strategy leads to improved convergence quality, avoids poor local minima, and reduces the overall computational cost.

When applying such a curriculum learning approach to system exhibiting a phase transition, one needs to be aware that an implicit bias may be introduced due to the training process starting in the ordered phase. A natural way to check for such bias would be to reverse this approach and start training from high noise rates. However, this is not feasible, as the complexity of the training set increases with the noise rate and does not reach its maximum at the phase transition. Nevertheless, we do not expect an implicit bias to arise, because of the hierarchical structure of the error patterns, as it is illustrated in Fig.~\ref{fig:curriculum-d3}. By increasing the noise rate, more complicated error patterns are slowly introduced while the likelihood of simple error chains is slowly decreased. The network is learning to track logical operators across the noise channel for error patterns of increasing complexity, which becomes increasingly difficult as the number of errors increases. The phase transition is only observed when considering a nonlinear function of $p(\lambda|s)$, such as the \gls{ci}, showing a characteristic crossing at the phase transition.

\section{Decoding strategy} \label{sec:decoding-strategy}
For estimating the \gls{ci}, our neural network approximates the conditional probabilities of the logical operators $p(\lambda | s)$. This already enables us to employ the network as a neural network decoder without further adaption. During decoding, the two logical operators are generated autoregressively. To determine the logical operators that define the recovery operations, two different strategies are possible. On one hand, we can select the most probable logical operator by rounding the network's prediction $\hat{\lambda}_j$; on the other hand, we can sample from the distribution $p(\lambda_j | s, \lambda_{<j})$. Both strategies have their advantages in different regimes. However, we show that it is beneficial to choose the most likely logical operator. A corresponding statement is also considered in Ref.~\cite{wichette_partition_2025} in more detail.

Assume, we measure a syndrome $s$ alongside the value $\lambda_x$ of the logical $X$-operator. First, if we decide on the most likely logical operator by thresholding our decision, whether this decision results in an error can be expressed as
\begin{align}
    N_\textrm{error}^{(th)} = \lambda_x \cdot \mathbb{I}(p(\lambda_x|s) \leq 0.5) + (1-\lambda_x) \cdot \mathbb{I}(p(\lambda_x|s) > 0.5),
\end{align}
where $\mathbb{I}$ denotes the Heaviside step function, which is $1$ for non-negative arguments and $0$ otherwise, $N_\textrm{error}^{(th)} = 1$ indicates a logical error and $N_\textrm{error}^{(th)}=0$ indicates a correct decision. $\lambda_x \in \{ 0,1 \}$ describes the decision for or against an $X$- or $Z$-type logical error. An error occurs only if $\lambda_x=0$ and $p(\lambda_x|s) > 0.5$, and vice versa. On the other hand, if we sample from the conditional probability $p(\lambda_x|s)$ of the logical operator, the indicator for a logical error can be written as
\begin{align}
    N_\textrm{error}^{(s)} = \lambda_x (1 - p(\lambda_x|s)) + (1 - \lambda_x) p(\lambda_x|s).
\end{align}
Consider the difference $N_\textrm{error}^{(th)} - N_\textrm{error}^{(s)}$ and distinguish two cases, $p(\lambda_x|s) \leq 0.5$ and $p(\lambda_x|s) > 0.5$.\\
\begin{widetext}
\begin{equation}
\begin{aligned}
    &1\textrm{st case:} \ p(\lambda_x|s) \leq 0.5: \\
    & \quad N_\textrm{error}^{(th)} - N_\textrm{error}^{(s)} = (1-\lambda_x) - \bigl( \lambda_x(1-p(\lambda_x|s)) + (1-\lambda_x)p(\lambda_x|s) \bigr) = (1- p(\lambda_x|s))(1-2\lambda_x) \\
    &2\textrm{nd case:} \ p(\lambda_x|s) > 0.5: \\
    & \quad N_\textrm{error}^{(th)} - N_\textrm{error}^{(s)} = \lambda_x - \bigl( \lambda_x(1-p(\lambda_x|s)) + (1-\lambda_x)p(\lambda_x|s) \bigr) = p(\lambda_x|s)(2\lambda_x-1)
\end{aligned}
\end{equation} 
\end{widetext}
Consider, as an example, the first case. For $\lambda_x=1$, taking the most likely operator is beneficial, while for $\lambda_x=0$, sampling is preferable since, due to stochasticity, we also draw the right decision $\lambda_x=0$ in some cases. However, since $\lambda_x=1$ will be more likely (with $p(\lambda_x|s) > 0.5$), the error introduced by sampling is higher than that introduced by thresholding. The same reasoning applies to the second case and can be formalized by taking the expectation over the syndrome distribution $p(s)$ to obtain the logical error rate $p_L$ for both strategies:
\begin{widetext}
\begin{equation}
\begin{aligned}
    p_L^{(th)} - p_L^{(s)} &= \sum_s p(s) \sum_{\lambda_x = 0,1} p(\lambda_x|s)\\
    &= \sum_s p(s) \Biggl( (1 - p)^2 \, \mathbb{I}(p > 0.5) - p(1-p) \, \mathbb{I}(p \leq 0.5)  - p(1-p) \, \mathbb{I}(p > 0.5) + p^2 \, \mathbb{I}(p \leq 0.5) \Biggr)\\
    &= \sum_s p(s) \Biggl( (1-p) (1-2p) \, \mathbb{I}(p > 0.5) + p(2p-1) \, \mathbb{I}(p \leq 0.5) \Biggr) < 0,
\end{aligned}
\end{equation} 
\end{widetext}
for $p(\lambda_x=0|s)=1-p$ and $p(\lambda_x=1|s)=p$. The argument can be straightforwardly generalized to several logical operators. Thus, we choose the most likely logical error in the following.

\section{Numerical calculations of the CI under circuit-level noise}
\begin{figure}
    \centering
    \includegraphics[width=0.97 \linewidth]{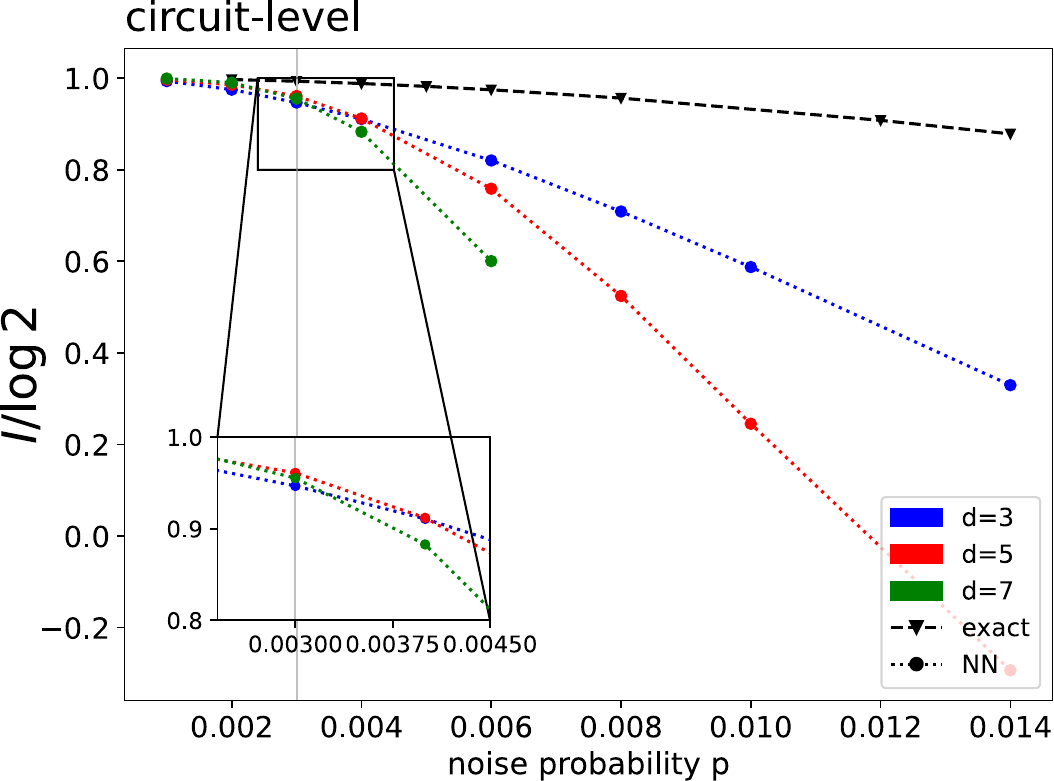}
    \caption{\justifying Comparison to the exact CI calculation in the circuit-level noise setting. We note that the previous method overestimates the \gls{ci} as computed with or neural network approach.}
    \label{fig:CI-circuit-level-exact}
\end{figure}
In Fig.~\ref{fig:results_CI}, we compare the \gls{ci} obtained under the phenomenological noise model with exact calculations performed following the method of Ref.~\cite{colmenarez_accurate_2024}.
For the circuit-level noise model at code distance $d=3$, we observe a discrepancy between our results and the one computed using the method in Ref.~\cite{colmenarez_accurate_2024}. We conjecture that this difference originates from the specific modeling of the density matrix $\rho_{Q}$ used in their approach. A comparison to the exact solution can be found in Fig.~\ref{fig:CI-circuit-level-exact}.

In Ref.~\cite{colmenarez_accurate_2024}, the density matrix $\rho_{Q}$ includes both the data and ancilla qubits, with the ancilla qubits not being measured. Instead, the full system of data and ancillas is treated coherently as the joint reference-system state. However, in order to more faithfully model the actual \gls{qec} procedure, the density matrix should instead be constructed as
\begin{equation}
    \rho_{RQ} = \sum_m \rho^m_{RQ} \otimes \ketbra{m}{m}_A,
\end{equation}
where $\rho^m_{RQ}$ denotes the density matrix of the data qubits and reference system, and the ancilla qubits $A$ are projected onto the measurement outcomes $m$. This construction reflects the projective nature of the stabilizer measurements performed in \gls{qec} circuits.

When the ancilla qubits are not measured, as is the case in the modeling of Ref.~\cite{colmenarez_accurate_2024}, coherent errors can propagate through the entangling two-qubit gates, particularly between $Z$- and $X$-type stabilizer ancillas. These ancilla qubits can remain entangled, leading to an overestimation of the \gls{ci}. Therefore, we believe that proper treatment of ancilla measurements is necessary for an accurate estimation of \gls{ci} under circuit-level noise.

\section{Phenomenological noise for distance-seven code}
In the phenomenological noise setting, we observe that training becomes significantly more difficult and computationally expensive for higher distances, similarly to the circuit-level noise case at similar code distances (e.g., $d=7$). However, for phenomenological noise the network often fails to converge to a good solution and instead gets trapped in local minima, where the \gls{ci} is not reliably attained. This reduced convergence performance may be attributed to the lower amount of structural information present in the data under phenomenological noise compared to circuit-level noise, making it more challenging for the network to learn meaningful patterns. Due to the high computational cost and unstable convergence behavior, we refrained from exhaustively optimizing the training in this regime. Nonetheless, for the sake of completeness, we report the obtained results in estimating the \gls{ci} in this section in Fig.~\ref{fig:CI-phenomenological-d7}.
\begin{figure}
    \centering
    \includegraphics[width=\linewidth]{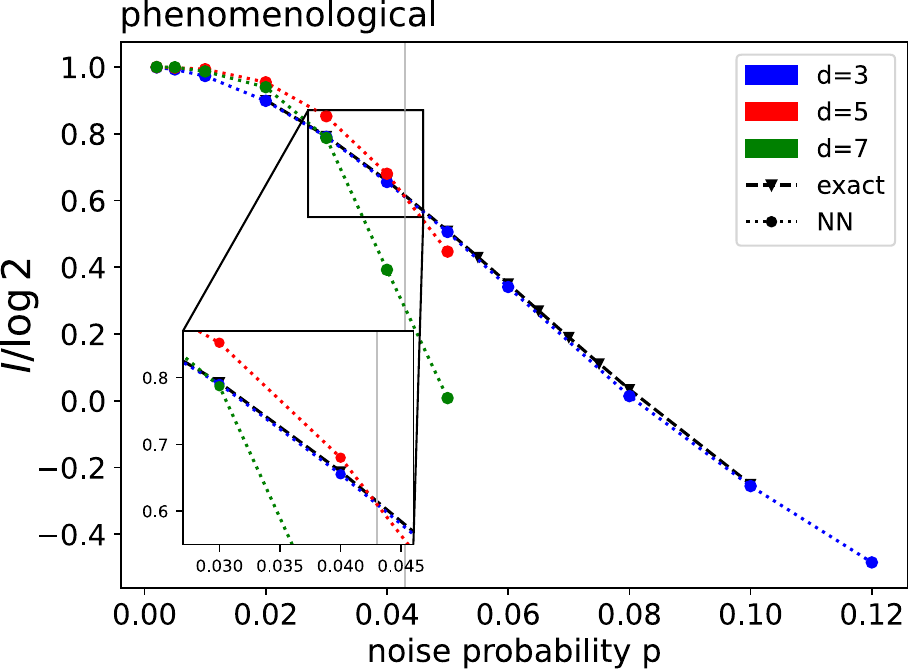}
    \caption{\justifying Estimates of the \gls{ci} for phenomenological noise, incl. $d=7$. We find that the training did not fully converged due to limited training time or got often stuck in suboptimal minima.}
    \label{fig:CI-phenomenological-d7}
\end{figure}

\section{Training loss} \label{app:loss}
To show how well the training converged or to verify that the neural network training has converged, we plot the training and validation loss over the course of training for the models used in our experiments. As shown in Fig.~\ref{fig:training_loss}, the loss decreases steadily and saturates, indicating stable convergence. The validation loss follows a similar trend, suggesting that the model does not overfit and generalizes well to unseen syndromes. These results confirm that the training process was successful and the reported performance is based on fully trained models.

\begin{figure*}
    \centering
    \includegraphics[width=0.8\linewidth]{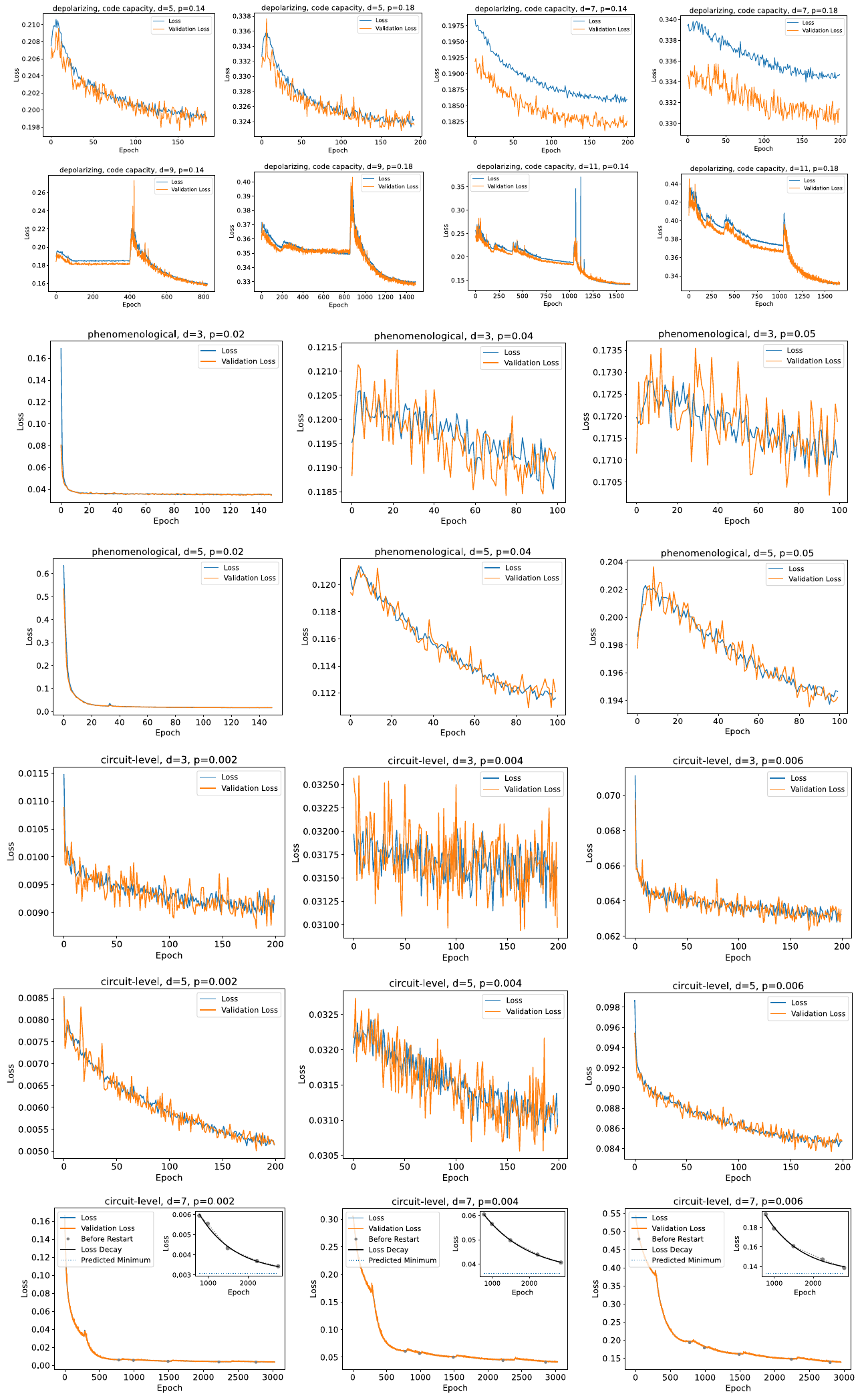}
    \caption{\justifying Training loss for selected, representative noise rates. We observe that training is most difficult around the threshold. For circuit level noise at distance $d=7$ we observe that the training did not fully converge yet. In the inset, we show to which value we expect the loss to go down.}
    \label{fig:training_loss}
\end{figure*}

\section{Theoretical perspective on MLD post-selection} \label{sec:post-selection-appendix}
In this section, we complete the theoretical results presented in Sec.~\ref{sec:post-selection} by proving the stated theorems and presenting additional information. In the following, we denote by $\mathcal{S}$ the set of all possible syndromes and by $p_L$ the logical error rate, which may depend on noise rate, post-selection parameter and post-selection scheme. Every probability measure $\mathbb{P}$ has to be understood as measure in the syndrome space. 

We recap the definition of the investigated \gls{mld} post-selection schemes. The selection functions $f$ are defined as follows.

\begin{definition}
Standard MLD post-selection. Let $c \in [0,1]$. Then the MLD post-selection function is given by 
\begin{equation}
    \begin{aligned}
        &f_\mathrm{MLD}: [0,1] \rightarrow \mathcal{P}(\mathcal{S}), \\ &f_\mathrm{MLD}(c) = \{ s \in \mathcal{S}: \underset{\lambda \in \Lambda}{\max} \  p(\lambda | s) > c\}.
    \end{aligned}
\end{equation}
\end{definition}

\begin{definition} 
Split MLD post-selection: Let $c \in [0,1]$. Then the split MLD post-selection function is given by 
\begin{equation}
\begin{aligned}
&f_\mathrm{NN}: \mathbb{R} \rightarrow \mathcal{P}(\mathcal{S}), \\ 
&f_\mathrm{NN}(c) = \{ s \in \mathcal{S}: \ (p(\lambda_x | s) > \sqrt{c} \ \lor \ p(\lambda_x | s) < 1 - \sqrt{c})  \\
&\quad  \land \ (p(\lambda_z | \lambda_x, s) > \sqrt{c} \ \lor \ p(\lambda_z | \lambda_x, s) < 1 - \sqrt{c}) \}.
    \end{aligned}
    \end{equation}
\end{definition}

\begin{remark}
    The idea for split \gls{mld} post-selection naturally emerges when autoregressively predicting the logical operators and applying a confidence filtering on the network output. However, given $p(\lambda|s)$ for $\lambda \in \{\mathbb{I}, \bar{X}, \bar{Y}, \bar{Z}$, the conditional probabilities can be calculated via
    \begin{align*}
        &p(\lambda_x|s) = \sum_{\lambda_z} p(\lambda_x, \lambda_z |s), \\
        & p(\lambda_z|s, \lambda_x = \hat{\lambda}_x) = \frac{p(\lambda_x = \hat{\lambda}_x, \lambda_z | s)}{p(\lambda_x = \hat{\lambda}_x |s)},
    \end{align*}
    enabling split \gls{mld} post-selection, using the \gls{mld} coset probabilities. Therefore, split \gls{mld} post-selection is a post-selection scheme in the following sense.
\end{remark}

\begin{definition} \label{def:app-post-selection}
    \gls{mld} post-selection: A post-selection scheme $f$ is an \gls{mld} post-selection scheme if syndromes are accepted or discarded based on the information available in the \gls{mld} coset probabilities and for post-selection parameter $c \in [0,1]$ we have $\underset{\lambda \in \Lambda}{\mathrm{max}} \ p(\lambda|s) > c$, i.e.~$f(c) \subset f_\mathrm{MLD}(c)$.
\end{definition}

\begin{definition}
Let $c \in [0,1]$ and $f: c \rightarrow \mathcal{P}(\mathcal{S})$ a post-selection function. Then the abort probability is defined by
\[
p_\mathrm{abort}(c, f) = \mathbb
P(f(c)^\mathcal{C}) = \sum_{s \notin f(c)} p(s).
\]
\end{definition}

Split MLD post-selection is stricter than MLD post-selection and therefore an MLD post-selection scheme in the setting of Def.~\ref{def:app-post-selection}.

\begin{corollary}
Let $s \in \mathcal{S}$. Then $f_\mathrm{NN}(c) \subset f_\mathrm{MLD}(c)$.
\end{corollary}
\begin{proof}
Let $s \in f_\mathrm{NN}(c)$ and $\lambda = \underset{\lambda \in \Lambda}{\mathrm{arg \, max}} \, p(\lambda | s) = (\hat{\lambda}_x, \hat{\lambda}_z)$. Then $p(\lambda_x = \hat{\lambda}_x | s) > \sqrt{c}$ and $p(\lambda_z = \hat{\lambda}_z | \hat{\lambda}_x, s) > \sqrt{c}$. Since $p(\lambda | s) = p(\lambda_x|s) p(\lambda_z | \lambda_x, s)$, we have $\underset{\lambda \in \Lambda}{\max} \  p(\lambda | s) > c$ and $s \in f_\mathrm{MLD}(c)$. 
\end{proof}

The logical error rate $p_L$ and its complement, the success rate $p_\mathrm{succ}$, quantify how well \gls{mld} performs on the set of accepted syndromes, i.e.~how likely a logical error is introduced after decoding on the set of accepted syndromes. 
\begin{corollary} \label{cor:bound}
    Let $f$ be an MLD post-selection function. Let $c \in [0,1]$. Then $p_\mathrm{succ}^c > c$ or equivalently, $p_\mathrm{L}^c < 1-c$.
\end{corollary}
\begin{proof}
    We start by writing the logical error rate
    \begin{align}
        p_L^c &= \sum_{E: \ s(E) \in f(c)} \Tilde{p}(E) \cdot p(\mathrm{fail} | s(E), E) \\
        &= \sum_{s \in f(c)} \Tilde{p}(s) \sum_{E} p(E|s) \cdot p(\mathrm{fail} | s(E), E) \\
        &= \sum_{s \in f(c)} \Tilde{p}(s) \cdot p(\mathrm{fail} | s),
    \end{align}
    where 
    \begin{align}
        \Tilde{p}(E) = \frac{p(E)}{\sum_{E^\prime: \ s(E) \in f(c)} p(E^\prime)},
    \end{align}
    is the probability of an error restricted to the accepted syndromes and 
    \begin{align} \label{eq:syndrome-normalized}
        \Tilde{p}(s) = \frac{p(s)}{\sum_{s^\prime \in f(c)} p(s^\prime)},
    \end{align}
    is the probability of observing a syndrome $s$ normalized by the set of accepted syndromes. Under MLD, we have 
    \begin{align*}
    &p(\mathrm{fail} | s) = \sum_{\lambda \in \Tilde{\Lambda}(s)} p(\lambda | s) = 1 - \underset{\lambda \in \Lambda}{\mathrm{max}} \ p(\lambda|s) \\ 
    & \quad \mathrm{for} \  \Tilde{\Lambda}(s) = \Lambda \setminus \{ \underset{\lambda \in \Lambda}{\mathrm{arg \ max}} \ p(\lambda | s) \}.
    \end{align*}
    This means, we can rewrite the success probability of decoding, $p_\mathrm{succ} = 1 - p_L$ as
    \begin{align}
        p_\mathrm{succ} = \sum_{s \in f(c)} \Tilde{p}(s) \  \underset{\lambda \in \Lambda}{\mathrm{max}} \ p(\lambda | s). \label{eq:p-succ}
    \end{align}
    Since we sum over syndromes $s \in f(c)$, we have $\underset{\lambda \in \Lambda}{\mathrm{max}} \ p(\lambda | s) > c$ and therefore
    \begin{align}
        p_\mathrm{succ} > c \sum_{s \in f(c)} \Tilde{p}(s) = c.
    \end{align}
\end{proof}

\begin{corollary}
    Let $f$ be an MLD post-selection function and $c,c^\prime \in [0,1]$ with $c < c^\prime$. Then $p_\mathrm{succ}^c < p_\mathrm{succ}^{c^\prime}$ or equivalently, $p_\mathrm{L}^c > p_\mathrm{L}^{c^\prime}$.
\end{corollary}
\begin{proof}
    Starting from Eq.~\eqref{eq:p-succ}, we have
    \begin{widetext}
        \begin{equation}
            \begin{aligned}
        p_\mathrm{succ} &= \sum_{s \in f(c)} \Tilde{p}^c(s) \  \underset{\lambda \in \Lambda}{\mathrm{max}} \ p(\lambda | s) \\
        &= \sum_{s \in f(c^\prime)} \Tilde{p}^c(s) \  \underset{\lambda \in \Lambda}{\mathrm{max}} \ p(\lambda | s) + \sum_{s \in f(c) \setminus f(c^\prime)} \Tilde{p}^c(s) \  \underset{\lambda \in \Lambda}{\mathrm{max}} \ p(\lambda | s) \\
        &\leq \sum_{s \in f(c^\prime)} \Tilde{p}^c(s) \  \underset{\lambda \in \Lambda}{\mathrm{max}} \ p(\lambda | s) + c^\prime \sum_{s \in f(c) \setminus f(c^\prime)} \Tilde{p}^c(s) \\
        &= p_\mathrm{succ}^{c^\prime} \frac{\sum_{s^\prime \in f(c^\prime)} p(s^\prime)}{\sum_{s \in f(c)} p(s)} + c^\prime \frac{\sum_{s^\prime \in f(c) \setminus f(c^\prime)} p(s^\prime)}{\sum_{s \in f(c)} p(s)} \\
        &\leq p_\mathrm{succ}^{c^\prime} \frac{\sum_{s^\prime \in f(c)} p(s^\prime)}{\sum_{s \in f(c)} p(s)} = p_\mathrm{succ}^{c^\prime},
    \end{aligned}
    \end{equation}
    \end{widetext}
    using the normalized syndrome distribution Eq.~\eqref{eq:syndrome-normalized}. The index $c$ denotes normalization of the syndrome density with respect to the post-selection parameter $c$. We used that $\Tilde{p}^{c^\prime}(s) \sum_{s^\prime \in f(c^\prime)} p(s^\prime) = \Tilde{p}^c(s) \sum_{s^\prime \in f(c)} p(s^\prime)$ and the result of Cor.~\ref{cor:bound}.
\end{proof}
As the case $c<0.25$ reduces to MLD decoding without post-selection, the previous corollary also proved that post-selection always decreases the logical error rate. With this, we formulate the main theorem of this chapter.

\begin{theorem} \label{thm:main-app}
Let $c, c^\prime \in [0, 1]$, $f$ an MLD post-selection function and $f^\prime: [0,1] \rightarrow \mathcal{P}(\mathcal{S})$ any other post-selection scheme. Then 
\[
p_\mathrm{abort}(c^\prime, f^\prime) \leq p_\mathrm{abort}(c, f) \Longrightarrow p_L(c^\prime, f^\prime) \geq p_L(c, f).
\]
\end{theorem}

\begin{proof}
    Since the logical error rate is not well defined in the case of $p_\mathrm{abort} = 1$, we assume without loss of generality $\mathbb{P}(s \in f(c))>0$ and $\mathbb{P}(s \in f^\prime(c^\prime))>0$. We assume that decoding is performed via \gls{mld}, i.e.~$p(\mathrm{succ} | s) = \underset{\lambda \in \Lambda}{\mathrm{max}} \ p(\lambda | s)$.
    Let $f^\prime$ be another post-selection scheme with parameter $c^\prime$ such that $p_\mathrm{abort}(c^\prime, f^\prime) \leq p_\mathrm{abort}(c, f)$.
    Then
    \begin{align}
        p_\mathrm{succ} &=  \sum_{s \in f(c)} \Tilde{p}(s) \  \underset{\lambda \in \Lambda}{\mathrm{max}} \ p(\lambda | s)\\
        &= \frac{\sum_{s \in f(c)} p(s) \ \underset{\lambda \in \Lambda}{\mathrm{max}} \ p(\lambda | s)}{\sum_{s^\prime \in f(c)} p(s^\prime)} \\
        &= \frac{\sum_{s \in f(c)} p(s) \ \underset{\lambda \in \Lambda}{\mathrm{max}} \ p(\lambda | s)}{1 - p_\mathrm{abort}}.
    \end{align}
    Distinguish two cases:\\
    (i) $p_\mathrm{abort} = p_\mathrm{abort}^\prime$: Then either $f(c) = f^\prime (c^\prime) \implies p_\mathrm{succ} = p_\mathrm{succ}^\prime$. Or $f(c) \neq f^\prime (c^\prime)$. Then there exist one or several $s$ such that $\underset{\lambda \in \Lambda}{\mathrm{max}} \ p(\lambda | s) \leq c$. Thus, $p_\mathrm{succ} > p_\mathrm{succ}^\prime$. \\
    (ii) $p_\mathrm{abort} > p_\mathrm{abort}^\prime$: Rewrite
    \begin{widetext}
        \begin{equation}
            \begin{aligned}
        p_\mathrm{succ} &= \frac{\sum_{s \in f(c)} p(s) \ \underset{\lambda \in \Lambda}{\mathrm{max}} \ p(\lambda | s)}{1 - p_\mathrm{abort}} \\
        &> \frac{\sum_{s \in f(c)} p(s) \ \underset{\lambda \in \Lambda}{\mathrm{max}} \ p(\lambda | s)}{1 - p_\mathrm{abort}^\prime} \\
        &= \frac{\sum_{s \in f(c) \cap f^\prime (c^\prime)} p(s) \ \underset{\lambda \in \Lambda}{\mathrm{max}} \ p(\lambda | s) + \sum_{s \in f(c) \setminus f^\prime (c^\prime)} p(s) \ \underset{\lambda \in \Lambda}{\mathrm{max}} \ p(\lambda | s)}{1 - p_\mathrm{abort}^\prime} \\
        &> \frac{\sum_{s \in f(c) \cap f^\prime (c^\prime)} p(s) \ \underset{\lambda \in \Lambda}{\mathrm{max}} \ p(\lambda | s) + c \sum_{s \in f(c) \setminus f^\prime (c^\prime)} p(s)}{1 - p_\mathrm{abort}^\prime} \\
        &> \frac{\sum_{s \in f(c) \cap f^\prime (c^\prime)} p(s) \ \underset{\lambda \in \Lambda}{\mathrm{max}} \ p(\lambda | s) + c \sum_{s \in f^\prime (c^\prime) \setminus f(c)} p(s)}{1 - p_\mathrm{abort}^\prime} \\
        & \geq \frac{\sum_{s \in f(c) \cap f^\prime (c^\prime)} p(s) \ \underset{\lambda \in \Lambda}{\mathrm{max}} \ p(\lambda | s) + \sum_{s \in f^\prime (c^\prime) \setminus f(c)} p(s)  \ \underset{\lambda \in \Lambda}{\mathrm{max}} \ p(\lambda | s)}{1 - p_\mathrm{abort}^\prime} \\
        &= \frac{\sum_{s \in f^\prime(c^\prime)} p(s) \ \underset{\lambda \in \Lambda}{\mathrm{max}} \ p(\lambda | s)}{1 - p_\mathrm{abort}^\prime} = p_\mathrm{succ}^\prime.
    \end{aligned}
        \end{equation}
    \end{widetext}
    Thus, $p_L(c^\prime, f^\prime) < p_L(c, f)$.
\end{proof}

In words, Thm.~\ref{thm:main-app} establishes that any MLD-based post-selection scheme is optimal: no other post-selection strategy can achieve a lower logical error rate without discarding additional syndromes. An MLD-based post-selection scheme chooses the subset of syndromes that yields the smallest possible logical error rate for the given abort probability.
In the next remarks, some examples of sub-optimal decoding and post-selection are discussed, where Thm.~\ref{thm:main-app} is not fulfilled. 

\begin{remark}
    In Ref.~\cite{english_thresholds_2024} a similar result is stated. However, we find that their theorem is wrong. It is stated: \\
    \textnormal{Post-selection that aborts if MLD returns a maximum coset probability less than some $c \in [0, 1]$ is optimal in the following sense. Such post-selection partitions the errors $\mathcal{E}$ into an abort set $\mathcal{E}_\mathrm{abort}$ of some cardinality $|\mathcal{E}_\mathrm{abort}| = r$. This partition strictly upper bounds $p_\mathrm{succ}$ for any abort set with equivalent cardinality $r$.}\\
    From our previous calculations, we infer that $\mathbb{P}_\mathrm{succ} = [\underset{\sigma \in \mathcal{P}^{\otimes k}}{\mathrm{max}} Z_{\sigma E}]_{E_\mathrm{accept}}$ in Appendix B describes the expectation value over all syndromes. Thus we interpret 
    \[
    \mathbb{P}_\mathrm{succ} = [\underset{\sigma \in \mathcal{P}^{\otimes k}}{\mathrm{max}} Z_{\sigma E}]_{E_\mathrm{accept}} = \sum_{s \in f(c)} \Tilde{p}(s) \  \underset{\lambda \in \Lambda}{\mathrm{max}} \ p(\lambda | s).
    \]
    Now consider the (exemplary, hand-crafted) situation depicted in Table~\ref{tab:example}. In this example, there are four distinct syndromes, each triggered by an equal number of possible error chains. We compare two different post-selection schemes: first, the standard \gls{mld} post-selection, which Ref.~\cite{english_thresholds_2024} claims is optimal in the previously defined sense; and second, an alternative scheme that is designed to accept syndromes $s_1$ and $s_3$, while discarding the other two. Both post-selection schemes thus partition the set of errors into accept and reject sets of equal cardinality. We note that $p_\mathrm{succ} < p_\mathrm{succ}^\prime$, which contradicts the result of Ref.~\cite{english_thresholds_2024}. However, since $p_\mathrm{abort} < p_\mathrm{abort}^\prime$, our own result in Thm.~\ref{thm:main-app} remains valid. This highlights that the quantity constraining the choice of post-selection schemes should be the abort probability, rather than the cardinality, i.e. number of aborted errors, of the abort partition.
\end{remark}

\begin{table}[]
    \centering
    \caption{\justifying For $c=0.8$. Some noise model to match the given values. MLD post-selection accepts $s_1$ and $s_2$. The other post-selection scheme is designed to accept $s_1$ and $s_3$.}
    \begin{tabular}{c|c|c|c|c}
        $s$ & $s_1$ & $s_2$ & $s_3$ & $s_4$ \\\hline\hline
        $p(s)$ & $0.8$ & $0.1$ & $0.001$ & $0.099$ \\
        $\underset{\lambda \in \Lambda}{\mathrm{max}} \ p(\lambda | s)$ & $0.99$ & $0.9$ & $0.25$ & $0.25$ \\
        $\Tilde{p}(s)_\mathrm{MLD}$ & $\frac{8}{9}$ & $\frac{1}{9}$ & - & -  \\
        $\Tilde{p}(s)_\mathrm{other}$ & $\frac{800}{801}$ & - & $\frac{1}{801}$ & - \\\hline
        $p_\mathrm{succ,MLD}$ & \multicolumn{4}{|c}{$0.9800$}  \\
        $p_\mathrm{succ,other}$ & \multicolumn{4}{|c}{$\approx 0.9891$} \\\hline
        $p_\mathrm{abort,MLD}$ & \multicolumn{4}{|c}{$0.100$} \\
        $p_\mathrm{abort,other}$ & \multicolumn{4}{|c}{$0.199$}  \\
    \end{tabular}
    \label{tab:example}
\end{table}

Standard \gls{mld} post-selection and split \gls{mld} post-selection are not equal.
\begin{remark}
    Let $s \in \mathcal{S}$. Then $f_\mathrm{MLD}(c) \not \subset f_\mathrm{NN}(c)$.\\
    Example: $c=0.64$, $\sqrt{c}=0.8$.\\
    \gls{mld} coset probabilities: $p(1) = 0.294$, $p(X)=0.686$, $p(Y)=0.014$, $p(Z)=0.006 \implies $ accepted \\
    Conditional probabilities: $p(\lambda_x|s) = 0.7$, $p(\lambda_z|\lambda_x,s) = 0.02 \implies$ discarded
\end{remark}

\begin{remark}
    As we have proven in Theorem~\ref{thm:main-app}, all \gls{mld} post-selection schemes are optimal. This might appear contradictory to our claim in Definition~\ref{def:split-post-selection} that the split \gls{mld} post-selection scheme outperforms the standard approach. The key lies in understanding the nature of the post-selection decision. For certain values of the physical error rate~$p$, all \gls{mld} schemes with $f(c) \subset f_\mathrm{MLD}(c)$ will discard the same initial set of syndromes. The only difference between them is the order in which additional syndromes are discarded to further reduce the logical error rate.\\
    As noted in Section~\ref{sec:soft-post-selection}, the maximum logical error rate (under \gls{mld}) is typically reached at high noise levels when only the trivial syndrome is accepted. Because the trivial syndrome provides no additional information—such as indicating an imbalance in the certainty between the two logical operators—any \gls{mld} scheme that accepts only this syndrome will ultimately yield the same performance bound.\\
    In summary, with respect to the trade-off between logical error rate and abort probability, both standard and split \gls{mld} post-selection schemes are equally optimal. However, the split \gls{mld} scheme begins discarding syndromes earlier that contribute disproportionately to the logical error rate—particularly those associated with a significant imbalance in the confidence between the two logical operators.
\end{remark}

\begin{remark}
    The upper bound on the logical error rate under \gls{mld} post-selection, $p_L^c < 1-c$, is strict for both post-selection schemes presented. From Fig.~\ref{fig:post-selection-comparison} it becomes apparent that the highest logical error rate is reached when all syndrome but the trivial syndrome are discarded. For discarding the trivial syndrome, both MLD and split MLD post-selection are equivalent. Hence, both reach this strict upper bound once $p(1|s) = 1-c$. As decoding the trivial syndrome will always have the highes confidence, for high noise rates, the trivial syndrome is the only syndrome that is not discarded. If one is given the coset probability $p(1|s_0)$, it might be possible to derive a sharper bound on the logical error rate excluding the effect of the trivial syndrome.
\end{remark}

\begin{figure}[h]
    \centering
    \includegraphics[width=0.8\linewidth]{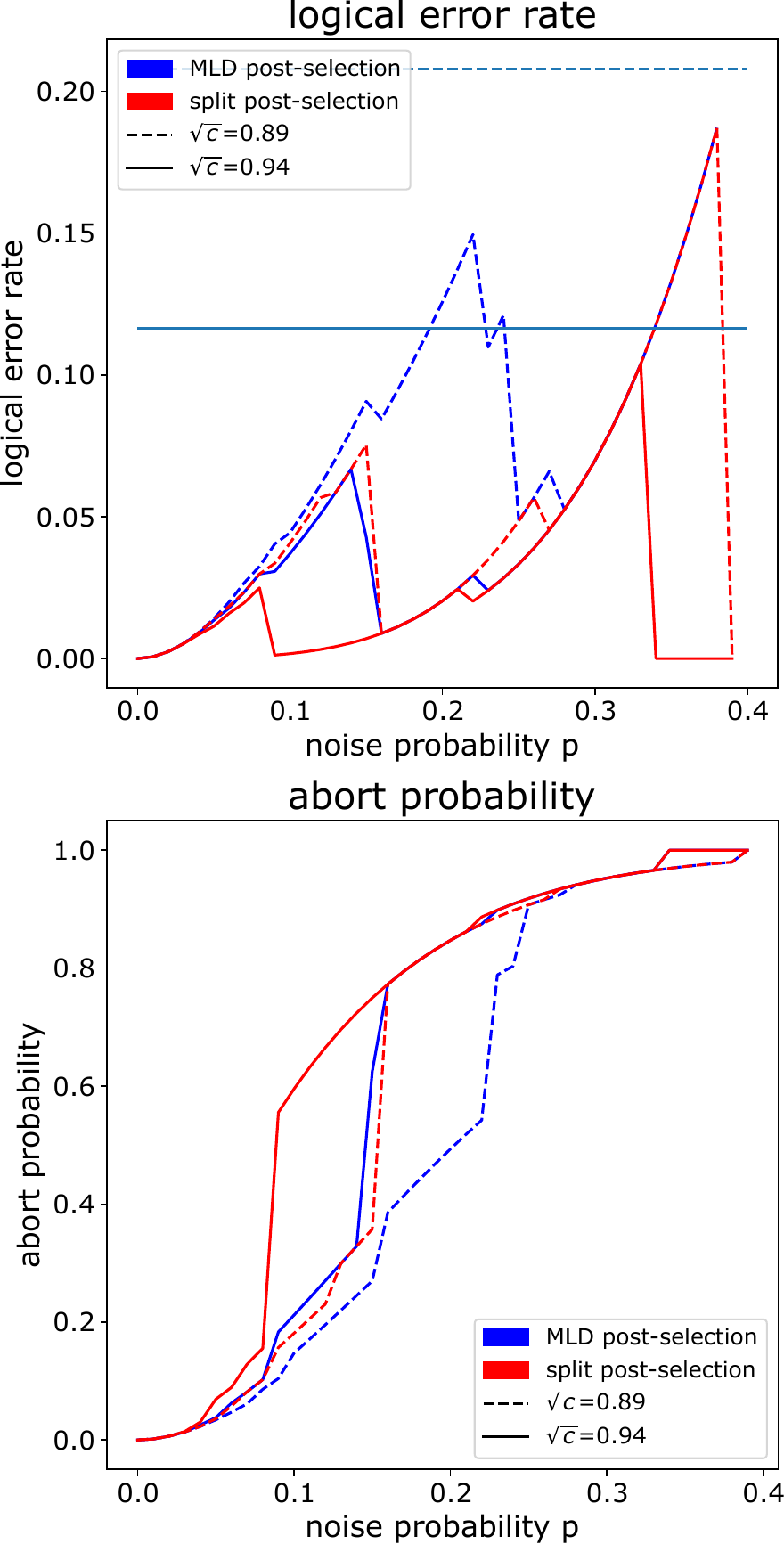}
    \caption{\justifying Comparison of the two post-selection schemes, standard \gls{mld} post-selection and split \gls{mld} post-selection, presented above. On the left, the upper bound on the logical error rate is shown.}
    \label{fig:post-selection-comparison}
\end{figure}

We find that the optimality established in Theorem~\ref{thm:main-app} reflects a balanced trade-off between discarding syndromes and accepting logical errors.
\begin{remark}
    We conjecture that each optimal post-selection scheme that fulfills Thm.~\ref{thm:main-app} has the property that $\underset{n \rightarrow \infty}{\lim} \frac{p_L}{p_{abort}} = 1-c$. We were not able to prove this conjecture but provide numerical evidence in Fig.~\ref{fig:pL-vs-pabort}. We leave the proof for future work.
\end{remark}

\begin{figure}
    \centering
    \includegraphics[width=\linewidth]{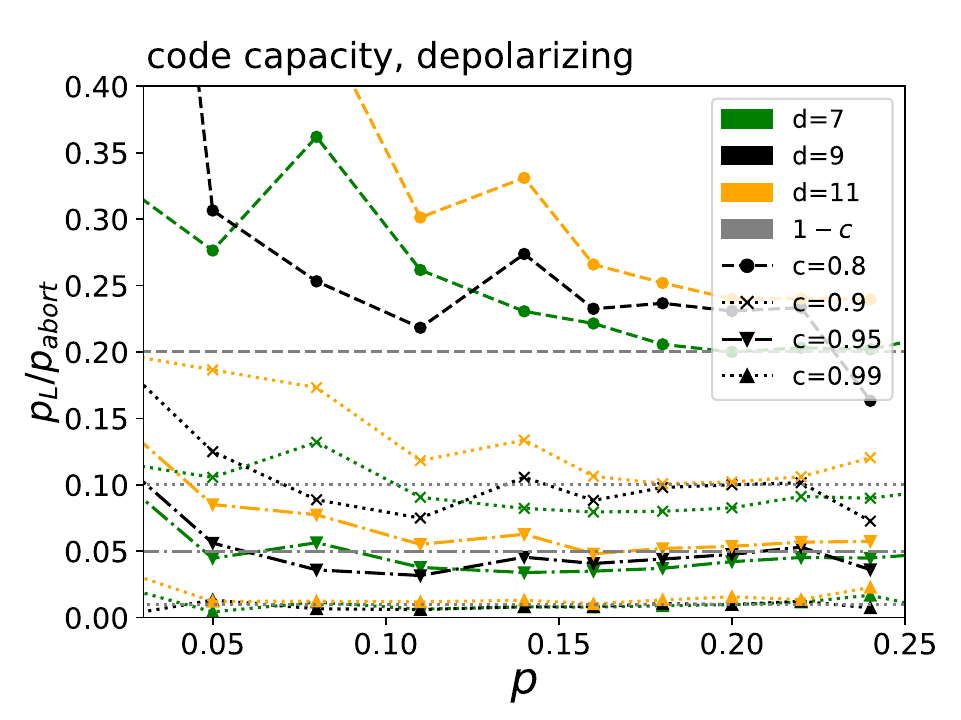}
    \caption{\justifying Logical error rate vs abort probability. Optimal post-selection schemes are characterized by a horizontal line, weighting equally between abort probability and logical error rate. Suboptimal schemes will not show a horizontal line and rather a $p$-dependence, as thresholds for logical error rate and abort probability do not coincide. We conjecture that $\underset{n \rightarrow \infty}{\lim} \frac{p_L}{p_{abort}} = 1-c$.}
    \label{fig:pL-vs-pabort}
\end{figure}

The abort threshold and the threshold of standard quantum computation coincide, as the maximum coset probability converges to unity in probability. This follows from the threshold theorem \cite{aharonov_fault-tolerant_2008,kitaev_fault-tolerant_2003,knill_resilient_1998,shor_fault-tolerant_1997}:
\begin{equation}
\lim_{d \to \infty} p_L(d, p) =
\begin{cases} 
0, & \text{if } p < p_{\text{th}}, \\
c > 0, & \text{if } p > p_{\text{th}}.
\end{cases}
\end{equation}

\begin{theorem}
    Let $g_n: \mathcal{S} \rightarrow [0,1], g_n(s)=\underset{\lambda \in \Lambda}{\mathrm{max}} \ p(\lambda | s)$ be the mapping of a syndrome to its maximum coset probability for an $n$-qubit \gls{qec} code. Then, for $p<p_{th}$ we have
    \begin{equation}
        g_n \overset{P}{\rightarrow} 1 \ \mathrm{as} \ n \rightarrow \infty.
    \end{equation}
\end{theorem}
\begin{proof}
    Let $p<p_{th}$, $c>0$, and $\epsilon > 0$. From the threshold theorem, we know $p_L \rightarrow 0$ as $n \rightarrow \infty$. Equivalently, $p_\mathrm{succ} \rightarrow 1$ as $n \rightarrow \infty$. Thus, choose $\epsilon_1 = \delta \epsilon$ and let $N \in \mathbb{N}$ such that $\forall n>N$
    \begin{equation}
    \begin{aligned}
        \epsilon_1 &> |p_\mathrm{succ} - 1| \\
        &= 1 - \sum_s \underset{\lambda \in \Lambda}{\mathrm{max}} \ p(\lambda|s) \\
        &= \sum_s p(s) \bigl(1 - g_n(s) \bigr) \\
        &> c \sum_{s: \ 1-g_n(s)>c} p(s) + \sum_{s: \ 1-g_n(s)\leq c} p(s) \bigl(1 - g_n(s) \bigr) \\
        &\geq c \sum_{s: \ 1-g_n(s)>c} p(s).
    \end{aligned}
    \end{equation}
    Thus, we have
    \begin{equation}
        P(|g_n - 1| > c) = \sum_{s:1-g_n(s)>c} p(s) < \frac{\epsilon_1}{c} = \epsilon.
    \end{equation}
\end{proof}
Thus, for any post-selection parameter $c$, the abort probability vanishes as the probability of syndromes with $g_n < c$ vanishes, as convergence in probability can be rewritten as
\begin{equation}
    \forall c>0: \ P(|g_n - 1| > c) \rightarrow 0 \ \mathrm{as} \ n \rightarrow \infty.
\end{equation}

With that, we identify a region where \gls{mld} post-selection is scalable whereas the modified \gls{mwpm} decoder from Ref.~\cite{smith_mitigating_2024} is not scalable, as also mentioned in Ref.~\cite{english_thresholds_2024}.
\begin{remark}
    Smith et al.~\cite{smith_mitigating_2024} introduce a post-selecting decoder by modifying the \gls{mwpm} algorithm to return both the most likely and the next-most likely, logically inequivalent correction, which have a weight difference $\Delta$. Then syndromes are discarded if $1 - \frac{\Delta}{d} > c$. The proposed post-selection scheme is further described in Sec.~\ref{sec:other-post-selection}. As this post-selection scheme is not based on the \gls{mld} cosets, thresholds of logical error rate and abort probability are different. Thus, the decoder is only scalable if $p<p_\mathrm{th, log}$ and $p<p_\mathrm{th, abort}$. The scalable region of optimal post-selection, i.e.~our scheme, compared to the modified \gls{mwpm} decoder is illustrated in Fig.~\ref{fig:mwpm-comparison}. 
\end{remark}

\begin{figure}
    \centering
    \includegraphics[width=\linewidth]{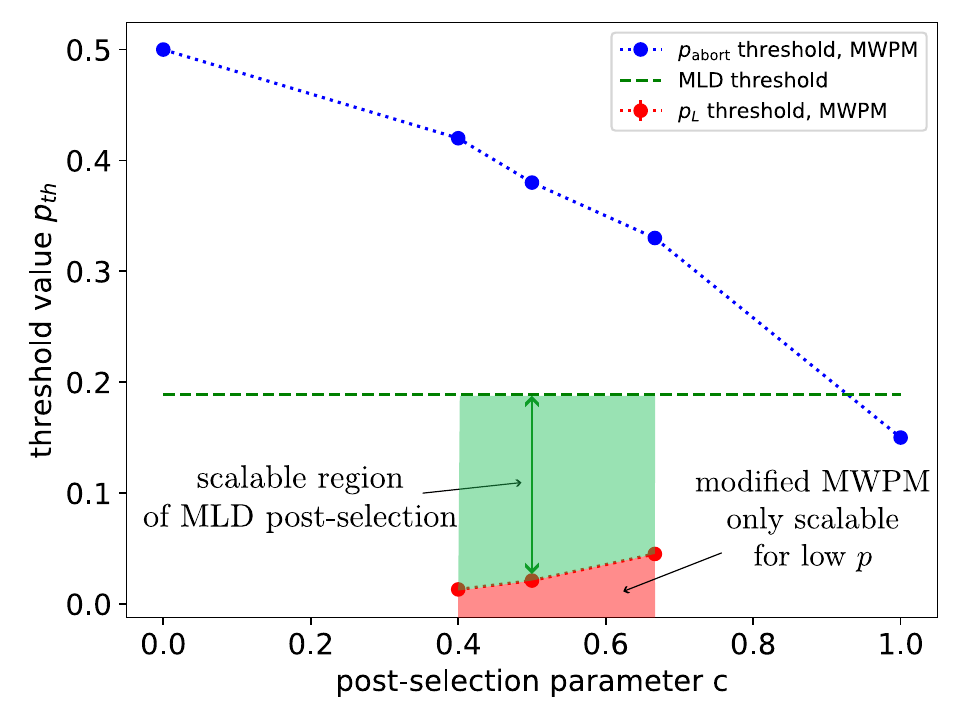}
    \caption{\justifying Illustration of the scalable region of the modified MWPM decoder compared to a post-selected ML decoder.}
    \label{fig:mwpm-comparison}
\end{figure}

Comparing to $\frac{p_L}{p_\mathrm{abort}}$ of our post-selected neural network decoder, we conjecture that suboptimal post-selection schemes do not show a horizontal line, i.e.~discarding and accepting logical errors is not balanced across different noise strengths. E.g.~for the modified \gls{mwpm} decoder \cite{smith_mitigating_2024}, we expect $\frac{p_L}{p_\mathrm{abort}}$ to converge to zero for $p_\mathrm{th,abort}<p<p_\mathrm{th,log}$. 

\section{Finite-size effects of post-selection} \label{app:finite-size-effects}
\begin{figure}
    \centering
    \includegraphics[width=0.8\linewidth]{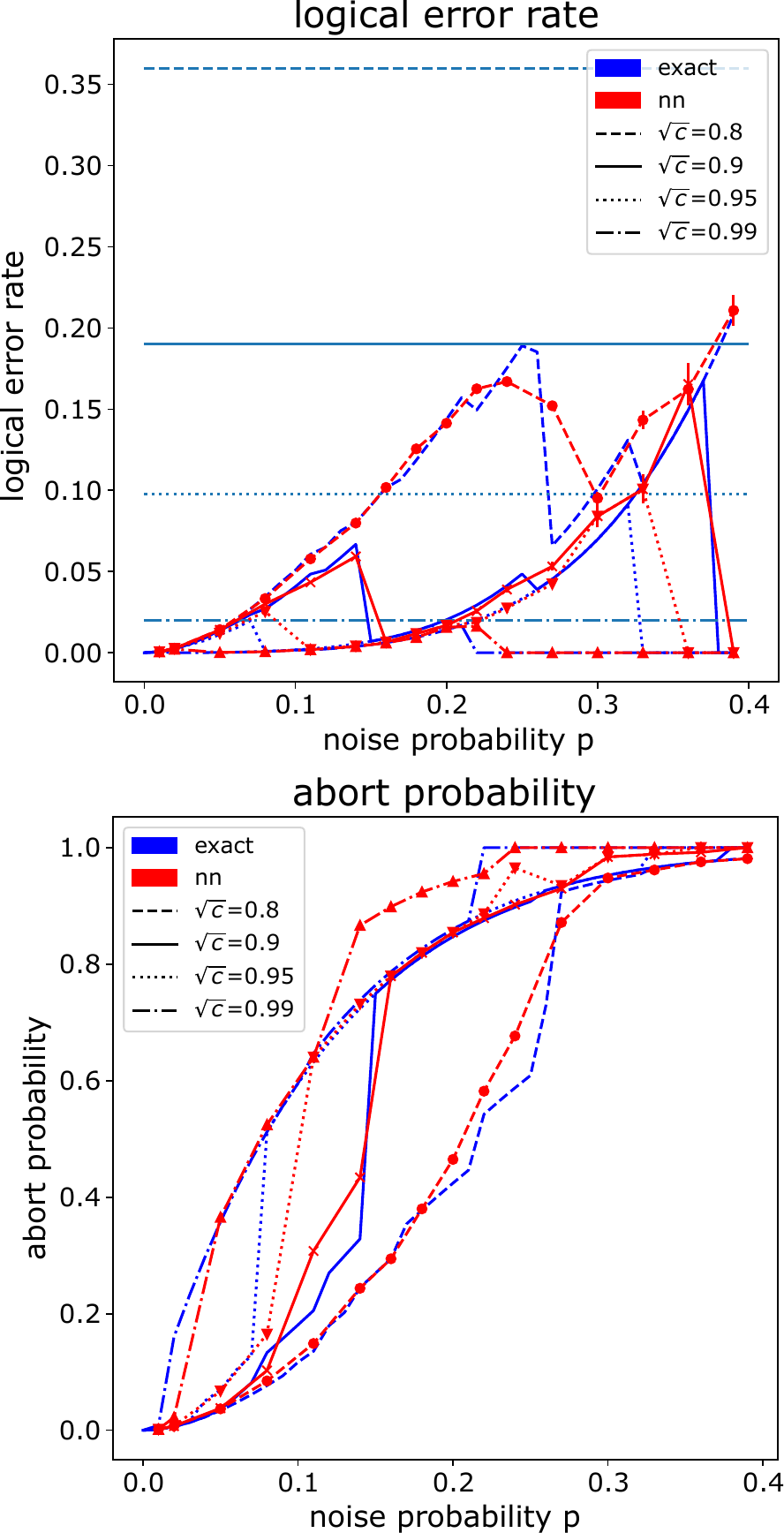}
    \caption{\justifying Finite size effects of post-selection. Comparison of the post-selected neural network decoder for distance $d=3$ in the code capacity setting, compared to the theoretical expectation, obtained from calculating $p(\lambda | s)$ for all possible syndromes $s$. Strong finite size effects are apparent in both the network output and the theoretical expectation.}
    \label{fig:compare-finite-size-effects}
\end{figure}
The abort probability and logical error rate under post-selection exhibit strong finite-size effects. As the physical error rate~$p$ increases, the coset probabilities shift, causing certain syndromes to be discarded. This leads to abrupt jumps in the abort probability and sudden drops in the logical error rate. While these effects are already visible in Fig.~\ref{fig:post-selection-comparison}, we further illustrate them using our $d=3$ neural network decoder. In particular, we compare the theoretically expected abort probability and logical error rate under split \gls{mld} post-selection with the results obtained from our neural network decoder in Fig.~\ref{fig:compare-finite-size-effects}.

\section{Other post-selecting decoders} \label{sec:other-post-selection}
Soft post-selection schemes have previously been investigated for the surface code under depolarizing noise by Smith et al.~\cite{smith_mitigating_2024} and for the toric code under bit-flip noise by English et al.~\cite{english_thresholds_2024}. Both works report a strong dependence of the threshold on the post-selection parameter. In the following, we briefly review their approaches and explain why our findings do not exhibit this behavior.

Smith et al.~\cite{smith_mitigating_2024} use a modified version of the \gls{mwpm} decoder, which returns not only the minimum weight correction $C$, but also all logically inequivalent minimum weight corrections $C_L$ satisfying $C C_L = L$ for logical operators $L$. Let $\Delta$ be the weight difference between the two most likely logically inequivalent corrections. Then the syndrome is discarded if
\begin{equation}
    1 - \frac{\Delta}{d} > c, \label{eq:mwpm-post-selection}
\end{equation}
for a code of distance $d$ and post-selection parameter $0 < c \leq 1$. We note that this post-selection rule is independent of the noise strength of the depolarizing channel, as only the weight difference between the two most likely inequivalent corrections are considered. 
For post-selection based on the probability of the \gls{mld} cosets, the coset probabilities depend on the noise strength $p$ of the error channel. This implies, formulating the former scheme in terms of \gls{mld} post-selection, the post-selection parameter $c$ would be noise-dependent. Since the post-selection condition Eq.~\eqref{eq:mwpm-post-selection} enforces a constant scaling, we conjecture that this comes at the cost of a threshold that depends on the post-selection parameter $c$. However, since this criterion does not incorporate the actual likelihoods of the error configurations under the physical noise model, it cannot achieve true maximum-likelihood performance and is therefore suboptimal.

In Smith et al.~\cite{english_thresholds_2024}, their post-selection rule is based on a heuristic scheme. A magnetization $M=\frac{1}{|\mathcal{S}|} \sum_i s_i$ is defined on the toric code syndromes by summing over all stabilizers of the syndrome. A syndrome $S$ is then discarded if
\begin{equation}
    M(S) < -1 + 2c,
\end{equation}
where excited stabilizers are assigned a value of $-1$ and non-excited a value of $+1$. Therefore, syndromes with a sufficiently high density of excited stabilizers are discarded, based on the heuristic assumption that decoding becomes too complicated if too many stabilizers are excited. For the toric code, the expectation of the syndrome mean $M$ is independent of the code distances and that there is a one-to-one correspondence between expected syndrome mean and noise strength of the bit-flip channel, given by
\begin{equation}
    M(S) = 1 - 2 \bigl( 4 p (1-p)^3 + 4 (1 - p) p^3 \bigr).
\end{equation}
Hence, choosing a post-selection parameter $c$ effectively determines the threshold $p_{th}$ of the abort rate via
\begin{equation}
    -1 + 2c = 1 - 2 \bigl( 4 p_{th} (1-p_{th})^3 + 4 (1 - p_{th}) p_{th}^3 \bigr).
\end{equation}
As the code distance determines the variance of the syndrome mean, the abort probability exhibits a threshold behavior that depends on the post-selection parameter $c$, as shown in Ref.~\cite{english_thresholds_2024}. This leads to a logical error rate threshold that also varies with $c$. This approach is likewise suboptimal, as the magnetization-based criterion is only a heuristic proxy for decoder confidence and does not explicitly account for the full statistical structure of the noise or syndrome correlations.

\section{Scaling of the post-selected maximum likelihood decoder}
\begin{figure*}
    \centering
    \includegraphics[width=\linewidth]{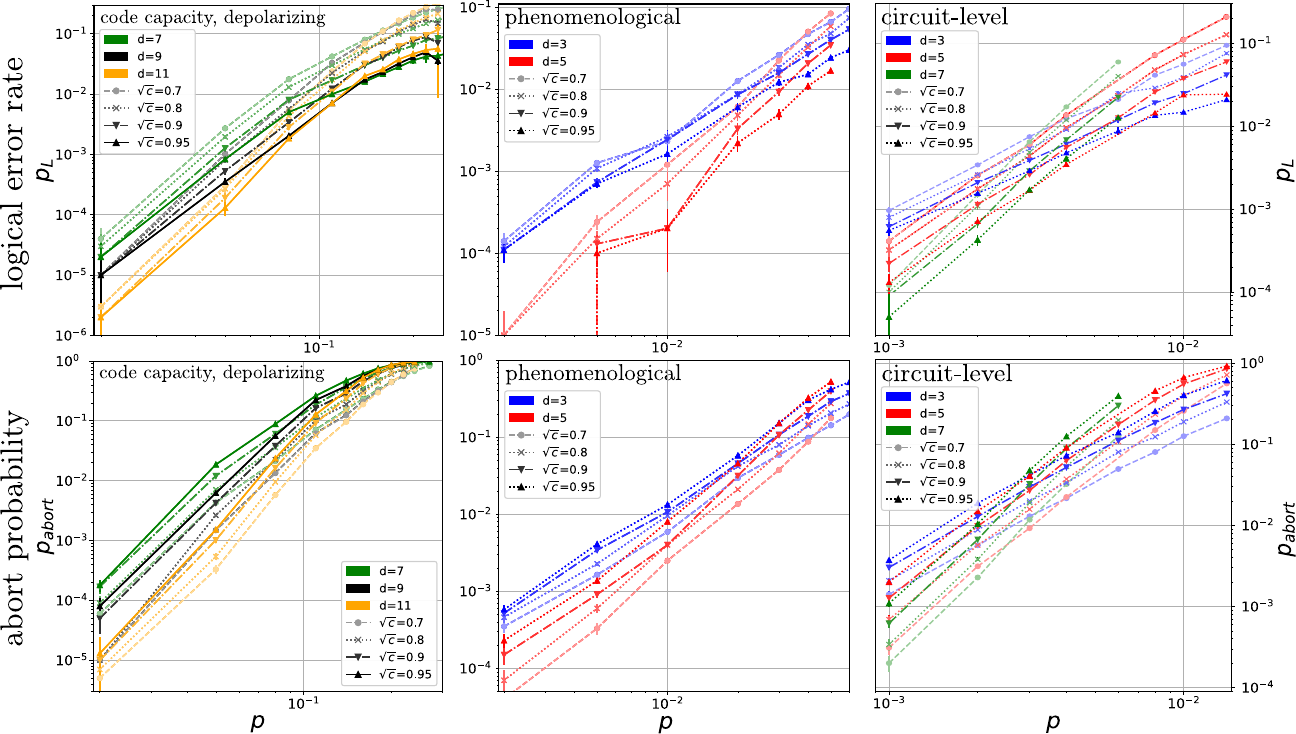}
    \caption{\justifying Decoding performance of our NN decoder with soft post-selection based on output probabilities. Complementary to Fig.~\ref{fig:abort-results} we show the scaling of logical error rate (upper panels) and abort probability (lower panels). Scaling of both logical error rate and abort probability depend on noise rate and post-selection strength.}
    \label{fig:abort-scaling}
\end{figure*}
Fig.~\ref{fig:abort-scaling} shows the scaling of both logical error rate and abort probability on log-log axes, complementing the results in Fig.~\ref{fig:abort-results}. Unlike standard quantum computation, the scaling is no longer constant due to the applied \gls{mld} post-selection scheme. In contrast to Ref.~\cite{smith_mitigating_2024}, we do not impose fixed scaling behavior. Instead, both the logical error rate and abort probability exhibit threshold behavior that aligns with the threshold observed in standard quantum computation.

\section{QEC-circuits} \label{sec:circuits}
In Figs.~\ref{fig:depolarizing-circuit}, \ref{fig:phenomenological-circuit}, and \ref{fig:circuit-level-circuit} we show the exact circuit for the readout of stabilizers and logical operators according to Fig.~\ref{fig:CI-setting-msmts} in the code capacity setting, for phenomenological noise, and for circuit-level noise.

\begin{figure*}[ht]
    \centering
    \rotatebox{90}{\includegraphics[width=0.85\textheight]{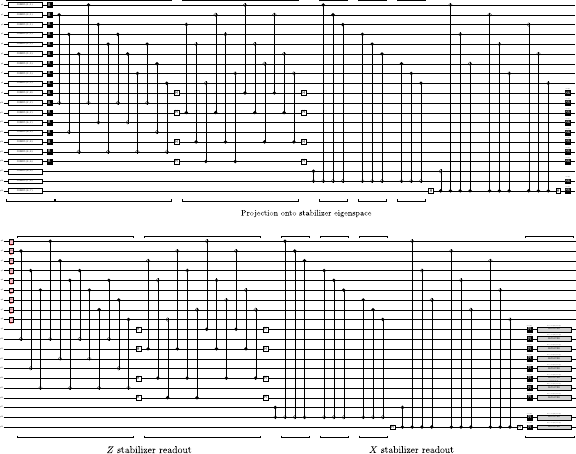}}
    \caption[Circuit of a QEC cycle under depolarizing noise.]{Circuit of a \gls{qec} cycle under depolarizing noise.}
    \label{fig:depolarizing-circuit}
\end{figure*}

\begin{figure*}[ht]
    \centering
    \rotatebox{90}{\includegraphics[width=0.95\textheight]{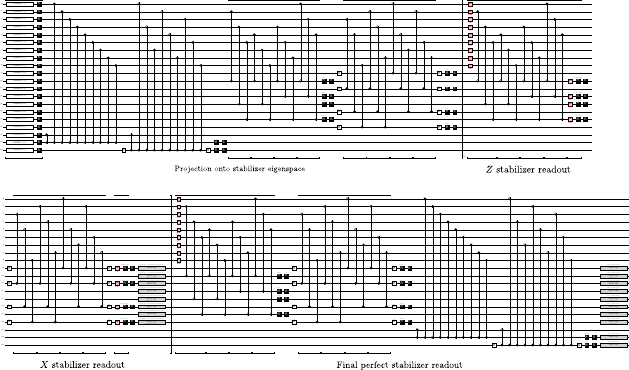}}
    \caption[Circuit of a QEC cycle under phenomenological noise.]{Circuit of a \gls{qec} cycle under phenomenological noise.}
    \label{fig:phenomenological-circuit}
\end{figure*}

\begin{figure*}[ht]
    \centering
    \rotatebox{90}{\includegraphics[width=0.95\textheight]{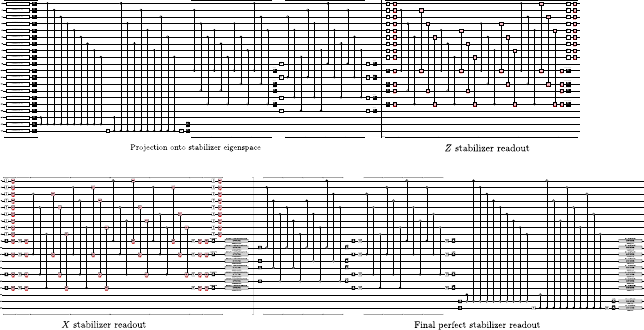}}
    \caption[Circuit of a QEC cycle under circuit-level noise.]{Circuit of a \gls{qec} cycle under circuit-level noise.}
    \label{fig:circuit-level-circuit}
\end{figure*}

\section{Hyperparameters}\label{app:hyperparameters}

The hyperparameters used in this work are shown in Table~\ref{tab:hyperparameters}.

\clearpage
\begin{table*}
    \centering
    \caption[Noise strengths for training the network.]{\justifying Noise strengths for training the network. Shown are the noise rates for depolarizing noise, phenomenological noise and single-parameter circuit-level noise. For larger distances, the model is only trained for noise rates slightly exceeding the observed threshold.}
    \begin{tabular}{l|l}
        \hline
        noise model & noise strength $p$ \\\hline \hline
        depolarizing & $\{0.01, 0.02, 0.05, 0.08, 0.11, 0.14, 0.16, 0.18, 0.2, 0.22, 0.24, 0.27, 0.3, 0.33, 0.36, 0.39 \}$ \\
        phenomenological &$ \{ 0.002, 0.005, 0.01, 0.02, 0.03, 0.04, 0.05, 0.06, 0.08, 0.1, 0.12 \}$ \\
        circuit-level & $\{ 0.001, 0.002, 0.003, 0.004, 0.006, 0.008, 0.01, 0.014 \}$ \\
        \hline
    \end{tabular}
    \label{tab:noise-strengths}
\end{table*}
\begin{table*}
    \centering
    \caption[Network and training hyperparameters.]{\justifying Network and training hyperparameters. The training follows a cosine decay schedule, after which a warm restart is applied depending on the convergence achieved during the previous cycle.}
    \begin{tabular}{l|l|c}
        \hline
        Module & Hyperparameter & Value \\\hline\hline
        Optimizer & Method & AdamW \cite{loshchilov_decoupled_2019} \\
        & Initial learning rate & $10^{-4}$ - $10^{-5}$ \\ 
        & Final learning rate & $10^{-8}$ \\
        & Scheduler & Cosine decay (opt.~warm restarts) \\
        & Weight decay & $10^{-5}$ \\
        & Batch size & $1000$ \\
        & Epochs & $200$ - $750$ before opt.~restart \\
        & Samples per epoch & $1\,000\,000$ \\\hline
        QEC-Transformer & Layers & $3$ \\ 
        & Embedding dimension $D$ & $128$ ($256$) \\
        & Feedforward dimension & $128$ ($256$) \\
        & Heads & $8$ \\
        & Stabilizer positional encoding & Fixed, 2D/3D \\\hline
        Patched Encoding & Patch distance & $3$ \\ 
        & Embedding convolutions & $3$ \\
        & Kernel size & $2$ \\
        & Stride & $2$ \\ \hline 
    \end{tabular}
    \label{tab:hyperparameters}
\end{table*}

\end{document}